\documentclass[11pt, a4paper, final]{article}

\usepackage{amsmath, amssymb, amsthm}
\allowdisplaybreaks[2]
\numberwithin{equation}{section}
\usepackage{xcolor}
\usepackage[T1]{fontenc}   
\usepackage{lmodern}       
\usepackage{exscale}       
\usepackage{mathtools}
\usepackage{mathrsfs}
\usepackage{bbm}
\usepackage{interval}
\intervalconfig{soft open fences}
\usepackage{physics}
\usepackage{enumitem}
\usepackage{braket}
\usepackage{siunitx}
\usepackage{geometry}
\usepackage{hyperref}
\usepackage{xparse}
\usepackage{xifthen}
\usepackage{tcolorbox}
\tcbuselibrary{theorems, skins}
\usepackage{subcaption}

\theoremstyle{plain}
\newtheorem{theorem}{Theorem}[section]
\newtheorem{lemma}[theorem]{Lemma}
\newtheorem{corollary}[theorem]{Corollary}
\newtheorem{proposition}[theorem]{Proposition}
\theoremstyle{definition}
\newtheorem{definition}[theorem]{Definition}
\theoremstyle{remark}
\newtheorem{remark}[theorem]{Remark}

\usepackage[english]{datetime2}
\usepackage{fancyhdr}
\usepackage{delarray}
\usepackage{ifdraft}
\usepackage[inline]{showlabels}
\ifdraft
  {
    \fancyhead[LO,RE]{\textit{\DTMnow}}
    \chead{\fbox{DRAFT}}
    \fancyhead[LE, RO]{}
    
    \setlength{\headheight}{15pt}
    \pagestyle{fancy}
  }
  {
    \pagestyle{plain}
  }

\newcommand{\Napier}{\mathord{\mathrm{e}}}
\newcommand{\imag}{\mathord{\mathrm{i}}}
\newcommand{\conj}[1]{{#1}^*}
\newcommand{\defeq}{\coloneqq}
\newcommand{\eqdef}{\eqqcolon}
\DeclarePairedDelimiter{\card}{\lvert}{\rvert}
\providecommand{\KroneckerDelta}[2]{\delta_{#1,#2}}
\providecommand{\ind}[1]{\mathbbm{1}_{\set{#1}}}
\DeclarePairedDelimiter{\floor}{\lfloor}{\rfloor}

\DeclareMathOperator{\sgn}{sgn}

\newcommand{\conv}{\mathbin{\ast}}        
\newcommand{\Conv}{\mathbin{\star}}       
\newcommand{\ft}[1]{\hat{#1}}             

\newcommand{\flt}[1]{\hat{#1}}
\newcommand{\FLT}[1]{\widehat{#1}}
\newcommand{\DLaplacian}[1]{\Delta_{#1}}  
\newcommand{\nnInt}{\mathbb{Z}_+}         
\NewDocumentCommand{\opSpace}{O{d}}{%
  \ifthenelse{\isempty{#1}}{%
    \mathbb{L}%
  }{%
    \mathbb{L}^{#1}%
  }%
}
\NewDocumentCommand{\opTime}{}{\nnInt}
\NewDocumentCommand{\opSpaceTime}{O{d}}{\opSpace[#1]\times\opTime}
\DeclareMathOperator{\TimeOf}{t}

\usepackage{tikz-lace}
  \usetikzlibrary{external}%
  \tikzsetexternalprefix{tikz-img/}%
  \tikzsetfigurename{lace-}%
  \tcbset{shield externalize}  

\title{Mean-field behavior of nearest-neighbor oriented percolation on the BCC lattice above $8+1$ dimensions}

\author{%
  Lung-Chi Chen\thanks{Department of Mathematical Sciences, National Chengchi University, Taipei, Taiwan} \thanks{Mathematics Division, National Center for Theoretical Sciences, Taipei, Taiwan}%
  \and Satoshi Handa\thanks{atama plus Inc.}%
  \and Yoshinori Kamijima\thanks{Mathematics Division, National Center for Theoretical Sciences, Taipei, Taiwan}%
}
\date{\today}

\begin{document}

\maketitle

\begin{abstract}
  \noindent
  In this paper, we consider nearest-neighbor oriented percolation with independent Bernoulli bond-occupation probability on the $d$-dimensional body-centered cubic (BCC) lattice $\opSpace$ and the set of non-negative integers $\opTime$.
  Thanks to the orderly structure of the BCC lattice, we prove that the infrared bound holds on $\opSpaceTime$ in all dimensions $d\geq 9$.
  As opposed to ordinary percolation, we have to deal with complex numbers due to asymmetry induced by time-orientation, which makes it hard to bound the bootstrap functions in the lace-expansion analysis.
  By investigating the Fourier-Laplace transform of the random-walk Green function and the two-point function, we derive the key properties to obtain the upper bounds and resolve a problematic issue in Nguyen and Yang's bound.
  The issue is caused by the fact that the Fourier transform of the random-walk transition probability can take the value $-1$.
\end{abstract}

\section{Introduction}

\subsection{Motivation}\label{sec:motivation}

In 1957, Broadbent and Hammersley~\cite{bh56} introduced oriented percolation (directed percolation)
in the context of physical phenomena, such as the wetting of porous medium.
It is well-known that oriented percolation exhibits critical phenomena in the vicinity of critical points.
It was first predicted by Obukhov~\cite{o80} that the upper critical dimension $d_\mathrm{c}+1$ for oriented percolation
equals $4+1$ (= spatial + temporal dimension), above which some quantities display the power law, and its exponents take mean-field values.
For example,
there exists the critical point $p_\mathrm{c}$ such that the susceptibility (the expected cluster size)
$\chi_p$ is finite if $p<p_\mathrm{c}$ (as $p\uparrow p_\mathrm{c}$, $\chi_p$ diverges at least as fast as
$(p_\mathrm{c} - p)^{-1}$ by the second inequality in \eqref{eq:differentialInequality-susceptibility} below).
It is believed that $\chi_p$ shows power-law behavior as $(p_\mathrm{c} - p)^{-\gamma}$
with the critical exponent $\gamma$.
Nguyen and Yang~\cite{ny93} proved via the infrared bound and the lace expansion that
spread-out oriented percolation with independent Bernoulli bond-occupation probability
on $\mathbb{Z}^d\times\opTime$ exhibits mean-field behavior in dimensions $d + 1 > 4 + 1$,
and there exists a sufficiently high dimension $d_0\gg 4$
such that nearest-neighbor oriented percolation on $\mathbb{Z}^d\times\opTime$ also exhibits it in $d + 1 \geq d_0 + 1$,
where $\opTime\defeq\mathbb{N}\sqcup\{0\}$,
and $\sqcup$ is a disjoint union.
The former supports the prediction of the upper critical dimension $d_\mathrm{c}=4$.

As a side note, we mention the results of the critical exponents in $d\leq 4$.
It is predicted that their values in $d + 1 = 4 + 1$ equal mean-field critical exponents with logarithmic corrections.
So far, the mathematically nonrigorous renormalization group method~\cite{js04}
and the computer-oriented approach~\cite{g09} suggest that the prediction is true.
There are almost only numerical results regarding the critical exponents for oriented percolation in $d + 1 = 1 + 1$, $2+1$ or $3+1$.
One can see the conjectural values in, e.g., \cite{o04}.

Similar results for ordinary percolation on $\mathbb{Z}^d$ are known.
It was first predicted in \cite{t74} that its upper critical dimension is $6$.
In the nearest-neighbor model, Hara and Slade~\cite{hs94} proved mean-field behavior in $d\geq 19$
via the lace expansion which was first derived in the seminal paper~\cite{hs90p}.
The best result $d\geq 11$ on $\mathbb{Z}^d$ was shown by Fitzner and van der Hofstad~\cite{fh17p}
via the non-backtracking lace expansion~\cite{fh17nobl}.
If we do not insist on the simple cubic lattice $\mathbb{Z}^d$,
a closer value $d\geq 9$ to the critical dimension on the body-centered cubic (BCC) lattice $\mathbb{L}^d$
is obtained by the second author, the third author and Sakai~\cite{Handa-Kamijima-Sakai}.

In this paper, we consider nearest-neighbor oriented percolation with independent Bernoulli bond-occupation probability
on $\opSpaceTime$ in high dimensions.
It is known what dimension the lace expansion works for nearest-neighbor ordinary percolation,
whereas Nguyen and Yang did not take care of the exact value of the dimension $d_0$ in which their analysis works for nearest-neighbor oriented percolation.
Our purpose is to specify $d_0$ and to prove the infrared bound in $d_0$.
As a result, we prove that the infrared bound holds on $\opSpaceTime[d\geq 9]$.

Section~\ref{sec:improvedBounds} and Section~\ref{sec:bounds-basicDiagrams} contain two novel tasks (see Remark~\ref{rem:role-restricting-g2} and Remark~\ref{rem:finiteness-weightedBubble} for the details, respectively).
Nguyen and Yang~\cite{ny93} used the inequality $\ft D_L(k) + 1 > 0$ for all $k$,
where $D_L$ is the random-walk transition probability for the spread-out model,
to bound $g_2$ in our notation below.
They used a similar inequality despite the nearest-neighbor model, but the inequality does not hold
because the Fourier transform $\ft D(k)$ of the random-walk transition probability can take $-1$ for some $k\in\interval{-\pi}{\pi}^d$.
The structure of the BCC lattice helps us to solve this problem (see also Proposition~\ref{prp:restricting-g2}).
Moreover, Nguyen and Yang showed the finiteness of the weighted open triangle diagram (the derivative of the open triangle diagram),
which corresponds to $\flt V_{p, 1}^{(0, 1)}(k)$ in this paper, only for $d>8$.
We observe that the period of $\ft D(k)^2$ equals that of $\ft D(2k)$.
This equality yields the finiteness of $\flt V_{p, m}^{(\lambda, \rho)}(k)$ in Lemma~\ref{lem:basic-diagrams} for $d>4$.

\subsection{Model}

We describe the model that we deal with in this paper.
First, we define the BCC lattice.
The $d$-dimensional BCC lattice $\opSpace$ is a graph that contains the origin $o=(0, \dots, 0)$
and is generated by the set of neighbors $\mathscr{N}^d = \{x = (x_1, \dots, x_d)\in\mathbb{Z}^d \mid \prod_{i=1}^{d}\abs{x_i} = 1\}$.
Although we should generally distinguish a graph from its vertex set,
for convenience, we denote by ``$x\in\opSpace$'' that $x$ belongs to the vertex set of the BCC lattice.
It was first used in \cite{Handa-Kamijima-Sakai} in the context of lace expansions and has some good properties as we see below.
On $\opSpace$, note that the cardinality, denoted by $\card{\bullet}$, of $\mathscr{N}^d$ equals $2^d$ (while it equals $2d$ on $\mathbb{Z}^d$).
The $d$-dimensional random-walk transition probability (1-step distribution) $D(x)$ is expressed as the product of $1$-dimensional random-walk transition probabilities:
\begin{equation}
  \label{eq:rwTransitonProbability}
  D(x)
  \defeq \frac{1}{\card{\mathscr{N}^d}}\ind{x\in\mathscr{N}^d}
  = \prod_{j=1}^{d}\frac{1}{2}\KroneckerDelta{\abs{x_j}}{1},
\end{equation}
where $\ind{\bullet}$ is the indicator function, and $\KroneckerDelta{\bullet}{\bullet}$ is the Kronecker delta.
This property is useful to compute the numerical values of the random-walk quantities defined as
\begin{equation}
  \label{eq:def-rwQuantities}
  \varepsilon_i^{(\nu)} = \sum_{n=\nu}^{\infty} D^{\conv 2n}(o) \times
  \begin{cases}
    1 & [i = 1],\\
    \left(n - \nu + 1\right) & [i = 2]
  \end{cases}
\end{equation}
for $\nu\in\nnInt$,
where $D^{\conv n}(x) = (D^{\conv (n-1)} \conv D)(x) \defeq \sum_{y\in\opSpace}D^{\conv (n-1)}(y) D(x - y)$ denotes the convolution on $\opSpace$.
See Appendix~\ref{sec:numericalComputations-rwQuantities} for the details of the computation.
Let $\ft D(k)$ be the Fourier transform of $D\colon \opSpace\to\mathbb{R}$ defined as
\begin{equation}
  \label{eq:Fourier-rwTransition}
  \ft D(k)
  \defeq \sum_{x\in\opSpace} D(x) \Napier^{\imag k\cdot x}
  = \prod_{j=1}^{d}\cos k_j
\end{equation}
for $k=(k_1, \dots, k_d)\in\mathbb{T}^d$, where $k\cdot x$ is the Euclidian inner product,
and $\mathbb{T}^d\defeq\interval{-\pi}{\pi}^d$ is the $d$-dimensional torus of side length $2\pi$ in the Fourier space.
In the equality in the above, we have used \eqref{eq:rwTransitonProbability}.

Second, we define nearest-neighbor oriented percolation.
We call $\opSpaceTime$ a space-time in which we write a vertex in the bold font,
i.e., $\vb*{x}=(x, t)\in\opSpaceTime$, while we write a vertex in space $\opSpace$
in the normal font, i.e., $x\in\opSpace$.
For convenience, $x$ and $\TimeOf(\vb*{x})$ denote the spatial and temporal components of $\vb*{x}$, respectively.
A bond is an ordered pair $((x, t), (y, t+1))$ of two vertices in $\opSpaceTime$.
Each bond $(\vb*{x}, \vb*{y})$ is
occupied with the probability
\[
  q_p(\vb*{y} - \vb*{x})\defeq pD(y - x)\KroneckerDelta{\TimeOf(\vb*{y}) - \TimeOf(\vb*{x})}{1}
\]
and vacant with the probability $1 - q_p(\vb*{y} - \vb*{x})$ independent of the others,
where $p\in\interval{0}{\norm{D}_\infty^{-1}}$ is the percolation parameter,
and $\norm{D}_\infty=\sup_{x\in\opSpace}\abs{D(x)}$.
Unlike much literature, note that $p$ alone does not mean probability.
Let $\mathbb{P}_p$ be the associated probability measure with such bond variables, and denote its expectation by $\mathbb{E}_p$.

For $n\in\mathbb{N}$, a (time-oriented) path of length $t$ is defined to be a sequence
$(\vb*{\omega}_0, \vb*{\omega}_1, \dots, \vb*{\omega}_t)$ of vertices in $\opSpaceTime$
such that $\TimeOf(\vb*{\omega}_s) - \TimeOf(\vb*{\omega}_{s-1})=1$ for $s=1, \dots, t$.
Let $\mathscr{W}(\vb*{x}, \vb*{y})$ be the set of all paths of length $t=\TimeOf(\vb*{y})-\TimeOf(\vb*{x})$
from $\vb*{x}=\vb*{\omega}_0$ to $\vb*{y}=\vb*{\omega}_t$.
By convention, $\mathscr{W}(\vb*{x}, \vb*{x})\equiv\set{(\vb*{x})}$.
We say that a path $\va*{\omega}=(\vb*{\omega}_0, \dots, \vb*{\omega}_t)\in\mathscr{W}(\vb*{x}, \vb*{y})$
of length $t=\TimeOf(\vb*{y})-\TimeOf(\vb*{x})$ is occupied if either $\vb*{x}=\vb*{y}$ or every bond $(\vb*{\omega}_{s-1}, \vb*{\omega}_s)$
for $s=1, \dots, t$ is occupied.
We say that $\vb*{x}$ is connected to $\vb*{y}$, denoted by $\vb*{x}\rightarrow\vb*{y}$,
if there is an occupied path $\va*{\omega}\in\mathscr{W}(\vb*{x}, \vb*{y})$.
Then, we define the two-point function as
\begin{equation*}
  \varphi_p(\vb*{x})
  = \mathbb{P}_p(\vb*{o}\rightarrow\vb*{x})
  = \mathbb{P}_p\qty(\bigcup_{\va*{\omega}\in\mathscr{W}(\vb*{o}, \vb*{x})} \set{\va*{\omega}\text{ is occupied}})
\end{equation*}
for $\vb*{x}=(x, t)\in\opSpaceTime$, where $\vb*{o}\defeq (o, 0)$.
The percolation probability and the susceptibility are defined as
\begin{equation*}
  \Theta_p = \mathbb{P}_p(\card{\mathcal{C}(\vb*{o})} = \infty)
  \qand
  \chi_p = \mathbb{E}_p\bigl[\card*{\mathcal{C}(\vb*{o})}\bigr] = \sum_{\vb*{x}\in\opSpaceTime} \varphi_p(\vb*{x})
\end{equation*}
respectively, where $\mathcal{C}(\vb*{o}) = \set{\vb*{x}\in\opSpaceTime | \vb*{o}\rightarrow\vb*{x}}$.
Also, the magnetization is defined as
\[
  M_{p, h} = \mathbb{E}_p\qty[1 - \Napier^{- \card{\mathcal{C}(\vb*{o})} h}],
\]
which is a resemblance to the spontaneous magnetization in ferromagnetic models.

Third, we describe critical phenomena.
The critical point is defined as
\[
  p_\mathrm{c}
  = \inf\Set{p\in\interval{0}{\norm{D}_\infty^{-1}} | \Theta_p > 0}
  = \sup\Set{p\in\interval{0}{\norm{D}_\infty^{-1}} | \chi_p < \infty}.
\]
Note that the second equality in the definition of $p_\mathrm{c}$ is quite nontrivial.
In this sense, the critical point for $\Theta_p$ and $\chi_p$ are often written as $p_\mathrm{H}$ and $p_\mathrm{T}$, respectively.
However, Menshikov \cite{m86} and Aizenman and Barsky \cite{ab87} independently proved that $p_\mathrm{H}=p_\mathrm{T}$
for any translation invariant percolation models.
Recently, Duminil-Copin and Tassion \cite{dct16} found a particularly
simple proof of the uniqueness on arbitrary locally finite vertex-transitive infinite graphs.
This uniqueness also holds on $\opSpace$, hence we do not distinguish between these critical values.

The power laws against the critical exponents $\beta$, $\gamma$ and $\delta$ are defined as
\begin{gather*}
  \Theta_p  \underset{p \downarrow p_\mathrm{c}}{\asymp} \qty(p - p_\mathrm{c})^\beta,\\
  \chi_p \underset{p \uparrow p_\mathrm{c}}{\asymp} \qty(p_\mathrm{c} - p)^{-\gamma},\\
  M_{p_\mathrm{c}, h} \underset{h \uparrow \infty}{\asymp} h^{-1 / \delta}.
\end{gather*}
Here,
\begin{itemize}
  \item $f(x) \asymp g(x)$ as $ x\to a$ denotes that there exist constants $\delta_0$, $C_1$ and $C_2$ such that, for any $x$ satisfying $\abs{x - a} < \delta_0$, $C_1 g(x) \leq f(x) \leq C_2 g(x)$,
  \item and $f(n) \asymp g(n)$ as $ n\to \infty$ denotes that there exist constants $n_0$, $C_1$ and $C_2$ such that, for any $n$ satisfying $n \geq n_0$, $C_1 g(n) \leq f(n) \leq C_2 g(n)$.
\end{itemize}

Finally, we mention the Fourier-Laplace transform of $\varphi_p(\vb*{x})$ as a preliminary for the main result.
Let
\[
  \Phi_p(k; t) = \sum_{x\in\opSpace}\varphi_p(x, t)\Napier^{\imag k\cdot x}
\]
be the Fourier transform of $\varphi_p(x, t)$ with respect to $x$.
The Markov property of oriented percolation implies that $\{\log\Phi_p(k; t) / t\}_{t=1}^{\infty}$ is a subadditive sequence,
hence there exists a constant $m_p^{-1}=\lim_{t\uparrow\infty}\Phi_p(k; t)^{1/t}$ depending on $p$ by \cite[Appendix.~II]{g99}
with $m_{p_\mathrm{c}}=1$ (See \cite[Section~1.2]{cs08} or \cite[Section~1]{ny95}).
For every $p<p_\mathrm{c}$, the inequality $m_p>1$ holds.
$m_p$ is the radius of convergence of the Laplace transform (the generating function) of $\Phi_p(0; t)$,
defined by
\begin{equation*}
  \flt\varphi_p(k, z)
  \defeq \sum_{t\in\opTime}\Phi_p(k; t) z^t
  = \sum_{(x, t)\in\opSpaceTime} \varphi_p(x, t) \Napier^{\imag k\cdot x} z^t
\end{equation*}
for $k\in\mathbb{T}^d$ and $z\in\mathbb{C}$ satisfying $\abs{z}\in\interval[open right]{0}{m_p}$.
Also, let
\[
  Q_p(x, t) = p^t D^{\conv t}(x) \ind{t\in\opTime}
\]
be the random-walk two-point function.
By convention, the $0$-fold convolution denotes the Kronecker delta: $D^{\conv 0}(x) = \KroneckerDelta{o}{x}$.
The Laplace transform of $Q_1(x, t)$ gives the well-known random-walk Green function as
\begin{equation*}
  S_p(x) \defeq \sum_{t\in\opTime} Q_1(x, t) p^t = \sum_{t\in\opTime} p^t D^{\conv t}(x)
\end{equation*}
for $x\in\opSpace$ and $p\in\interval{0}{1}$.
Notice that the radius of convergence of $S_p(x)$ is $1$, and in particular $S_1(x)$ is well-defined in a proper limit when $d>2$.
By Bool's inequality, one can easily see that $\varphi_p(x, t) \leq Q_p(x, t)$ for every $0\leq p<1$.
Taking the sum of both sides leads to $\chi_p \leq (1 - p)^{-1}$, which implies $p_\mathrm{c}\geq 1$.

\subsection{Main result}

Aizenman and Newman~\cite{an84} showed that the triangle condition is a sufficient condition
for percolation models to exhibit mean-field behavior (or $\gamma=1$).
Barsky and Aizenman~\cite{ba91} extended/reformulated the Aizenman-Newman triangle condition
in order to deal with oriented percolation in a unified manner.
They also showed that other critical exponents $\delta$ and $\beta$ take
the mean-field values $2$ and $1$, respectively, if the triangle condition is satisfied.

The triangle condition expresses that
\begin{equation}
  \label{eq:triangleCondition}
  \lim_{R\to\infty} \bigtriangleup_{p_\mathrm{c}}(R) = 0,
\end{equation}
where
\[
  \bigtriangleup_p(R) = \sup\Set{
    \sum_{\vb*{y}\in\opSpaceTime} \varphi_p^{\Conv 2}(\vb*{y}) \varphi_p(\vb*{y} - \vb*{x})
    | \norm{\vb*{x}}_2 \geq R
  },
\]
the norm of $\vb*{x}=(x, t)$ denotes $\norm{(x, t)}_2 = (\sum_{i=1}^{d} \abs{x_i}^2 + t^2)^{1 / 2}$,
and $\varphi_p^{\Conv n}(x, t) = (\varphi_p^{\Conv (n-1)} \Conv \varphi_p)(x, t) \defeq \sum_{(y, s)\in\opSpaceTime} \allowbreak \varphi_p^{\Conv (n-1)}(y, s) \varphi_p(x - y, t - s)$ denotes the $n$-fold convolution on $\opSpaceTime$.
If \eqref{eq:triangleCondition} holds, then solving the differential inequalities~\cite{ab87,an84,ba91,cc86}
\begin{gather}
  \label{eq:differentialInequality-susceptibility}
  \epsilon_p(R) \qty\big(
    1 - \bigtriangleup_{p_\mathrm{c}}(R)
  ) \chi_p^2
    \leq \dv{\chi_p}{p}
    \leq \chi_p^2,\\
  \label{eq:differentialInequality-percolation}
  \epsilon_p'(R) \qty\big(
    1 - f_p \bigtriangleup_{p_\mathrm{c}}(R)
  ) M_{p, h} \pdv{M_{p, h}}{h}
    \leq M_{p, h} - h \pdv{M_{p, h}}{h}
    \leq \frac{p}{1 - p \norm{D}_\infty} M_{p, h}^2 \pdv{M_{p, h}}{h} + M_{p, h}^2
\end{gather}
imply $\beta=\gamma=1$ and $\delta=2$,
where $\epsilon_p(R)$, $\epsilon_p'(R)$ and $f_p$ are model-dependent functions with $1 / \epsilon_p(R)$, $1 / \epsilon_p'(R)$ and $f_p$ uniformly bounded in a neighborhood of $p_\mathrm{c}$.
Note that $\Theta_p = \lim_{h\downarrow 0} M_{p, h}$.

To verify \eqref{eq:triangleCondition}, we use the infrared bound which is our main theorem
(cf., \cite[Theorem~1]{ny93} and \cite[Theorem~2]{ny95}).
\begin{theorem}[Infrared bound]\label{thm:infraredBound}
  For nearest-neighbor oriented percolation on $\opSpaceTime[d\geq 9]$ with independent bond statuses,
  there exists a model-dependent constant $K\in\interval[open]{0}{\infty}$ such that
  \begin{equation}
    \label{eq:infraredBound}
    \abs{\flt\varphi_p(k, z)}
    \leq \frac{K}{\abs*{1 - \Napier^{\imag\arg z} \ft D(k)}}
    = K \abs{\ft S_{\Napier^{\imag\arg z}} (k)}
  \end{equation}
  uniformly in $p\in\interval[open right]{0}{p_\mathrm{c}}$, $k\in\mathbb{T}^d$
  and $z\in\mathbb{C}$ with $\abs{z}\in\interval[open right]{1}{m_p}$.
\end{theorem}

\begin{corollary}[{\cite{ny93}}]\label{cor:triangleCondition}
  If the infrared bound \eqref{eq:infraredBound} holds, then the triangle condition \eqref{eq:triangleCondition} is satisfied.
\end{corollary}
\begin{proof}
  Fix $\epsilon\in\interval[open]{0}{1 \wedge \frac{d-4}{6}}$, where $a\wedge b$ denotes $\min\set{a, b}$.
  By the Hausdorff-Young inequality and \eqref{eq:infraredBound},
  there is a constant $C_{\epsilon}$ depending on $\epsilon$ such that
  \begin{align}
    &\qty(
      \sum_{\vb*{x}} \abs{\sum_{\vb*{y}} \varphi_{p}^{\Conv 2}(\vb*{y}) \varphi_{p}(\vb*{y} - \vb*{x})}^{1+1/\epsilon}
    )^{\epsilon/(1+\epsilon)}\\
    &\leq C_\epsilon \qty(
      \int_{\mathbb{T}^d}\frac{\dd[d]{k}}{(2\pi)^d} \int_{\mathbb{T}}\frac{\dd{\theta}}{2\pi}
      \abs{\flt\varphi_p(k, \Napier^{\imag\theta})^2 \flt\varphi_p(k, \Napier^{-\imag\theta})}^{1+\epsilon}
    )^{1/(1+\epsilon)} \notag\\
    &= C_\epsilon \qty(
      \int_{\mathbb{T}^d}\frac{\dd[d]{k}}{(2\pi)^d} \int_{\mathbb{T}}\frac{\dd{\theta}}{2\pi}
      \abs{\flt\varphi_p(k, \Napier^{\imag\theta})}^{3 (1+\epsilon)}
    )^{1/(1+\epsilon)} \notag\\
    &\leq C_\epsilon \qty(
      \int_{\mathbb{T}^d}\frac{\dd[d]{k}}{(2\pi)^d} \int_{\mathbb{T}}\frac{\dd{\theta}}{2\pi}
      \frac{K^{3 (1+\epsilon)}}{\abs\big{1 - \Napier^{\imag\theta} \ft D(k)}^{3 (1+\epsilon)}}
    )^{1/(1+\epsilon)}.
    \label{eq:consequence-Hausdorff-Young}
  \end{align}
  In the equality, we have used $\flt\varphi_p(k, \conj{z}) = \conj{\flt\varphi_p(k, z)}$ for every $k\in\mathbb{T}^d$
  and $z\in\mathbb{C}$, where $\conj{z}$ is the complex conjugate of $z$.
  The integrability of the right-most side in the above follows from the Taylor series of $(1 - \ft D(k)) + \ft D(k) (1 - \Napier^{\imag\theta})$ around the singularities $\set{(k, \theta) \in \mathbb{T}^{d+1} | 1 - \Napier^{\imag\theta} \ft D(k) = 0}$ and
  \begin{align*}
    &\iint_{\Set{(k, \theta)\in\mathbb{T}^{d+1} | \norm{k}_2 \leq \delta_1,\ \abs{\theta} \leq \delta_2}}
      \frac{1}{\qty\big(\norm{k}_2^2 + \abs{\theta})^{3 (1+\epsilon)}}
      \frac{\dd[d]{k}}{(2\pi)^d}\frac{\dd{\theta}}{2\pi}\\
    &= 2 \int_{\Set{k\in\mathbb{T}^d | \norm{k}_2 \leq \delta_1}} \frac{\dd[d]{k}}{(2\pi)^d}
      \int_{0}^{\delta_2}
      \frac{1}{\qty\big(\norm{k}_2^2 + \theta)^{3 (1+\epsilon)}}
      \frac{\dd{\theta}}{2\pi}\\
    &= \frac{2}{2 + 3\epsilon} \int_{\Set{k\in\mathbb{T}^d | \norm{k}_2 \leq \delta_1}}
      \eval{\frac{1}{\qty\big(\norm{k}_2^2 + \theta)^{2 + 3\epsilon}}}_{\theta=\delta_2}^{0}
      \frac{\dd[d]{k}}{(2\pi)^d}\\
    &\lesssim \int_{\Set{k\in\mathbb{T}^d | \norm{k}_2 \leq \delta_1}}
      \frac{1}{\norm{k}_2^{2 (2 + 3\epsilon)}}
      \frac{\dd[d]{k}}{(2\pi)^d}
    \lesssim \int_{0}^{\delta_1} r^{d - 5 - 6\epsilon} \dd{r}
    < \infty
  \end{align*}
  for $\delta_1, \delta_2 < 1$,
  where $f(x) \lesssim g(x)$ means that there exists a constant $C < \infty$ such that $f(x) \leq C g(x)$.
  Since the upper bound on \eqref{eq:consequence-Hausdorff-Young} does not depend on $p$,
  $\sum_{\vb*{y}} \varphi_{p_\mathrm{c}}^{\Conv 2}(\vb*{y}) \varphi_{p_\mathrm{c}}(\vb*{y} - \vb*{x})$
  is a $(1+1/\epsilon)$-summable function over $\opSpaceTime$.
  This immediately implies
  \[
    \sum_{\vb*{y}} \varphi_{p_\mathrm{c}}^{\Conv 2}(\vb*{y}) \varphi_{p_\mathrm{c}}(\vb*{y} - \vb*{x})
    \xrightarrow[\norm{\vb*{x}}_2 \uparrow \infty]{} 0,
  \]
  and completes the proof.
\end{proof}

As we explained above, the triangle condition and the differential inequalities \eqref{eq:differentialInequality-susceptibility} and \eqref{eq:differentialInequality-percolation} yield the next corollary.
It is well-known that $\Theta_p$ is continuous at $p_\mathrm{c}$ (i.e., $\Theta_{p_\mathrm{c}}=0$) in the case of oriented percolation~\cite{bg90,gh02}.
Refer \cite[Proof of Proposition~3.1]{an84} to see how to imply the inequalities \eqref{eq:inequalities-susceptibility} from \eqref{eq:differentialInequality-susceptibility}.
Specifically, integrating the differential inequalities leads to their inequalities.
Refer \cite[Proof of Lemma~5.1]{ab87} and \cite[Proof of Proposition~4.1]{ba91} to show \eqref{eq:inequalities-percolationProbability} and \eqref{eq:inequalities-magnetization}.
As a side note, the shorter proof of a lower bound on $\Theta_p$ than \cite{ab87,cc86} is given by \cite{dct16}.

\begin{corollary}[{\cite{ba91}}]
  If the triangle condition \eqref{eq:triangleCondition} is satisfied, then there exist constants $\{C_i\}_{i=1}^{6}\subset\interval[open]{0}{\infty}$ such that, in the vicinity of $p_\mathrm{c}$,
  \begin{gather}
    C_1 (p - p_\mathrm{c}) \vee 0 \leq \Theta_p \leq C_2 (p - p_\mathrm{c}) \vee 0,
    \label{eq:inequalities-percolationProbability}\\
    \frac{C_3}{p_\mathrm{c} - p} \leq \chi_p \leq \frac{C_4}{p_\mathrm{c} - p},
    \label{eq:inequalities-susceptibility}
  \end{gather}
  and, for small $h\geq 0$,
  \begin{equation}
    \label{eq:inequalities-magnetization}
    C_5 h^{1/2} \leq M_{p_\mathrm{c}, h} \leq C_6 h^{1/2},
  \end{equation}
  where $a\vee b$ denotes $\max\set{a, b}$.
  Therefore, $\beta=\gamma=1$ and $\delta=2$.
\end{corollary}

We obtain upper bounds on $p_\mathrm{c}$ as a by-product through the proof of Theorem~\ref{thm:infraredBound}.
See \eqref{eq:numericalBound-g1} below (whose numerical value is rounded up to the fourth decimal place due to significant figures).
Incidentally, one can obtain a more precise estimate because the lace expansion also provides an asymptotic expansion for $p_\mathrm{c}$ \cite{hs95,hm22,hs05-asymp,hs06}.
This topic is outside the reach of this paper, so that we have not taken care of this subject.

\begin{corollary}
  For each dimension $d$, an upper bound on the critical point $p_\mathrm{c} (\geq 1)$ for nearest-neighbor oriented percolation on $\opSpaceTime$ is given by the following table.
  \begin{table}[htbp]
    \centering
    \begin{tabular}{|c|cccc|}
      \hline
      $d+1$ & $9+1$ & $10+1$ & $11+1$ & $12+1$\\
      $p_\mathrm{c} \leq$ & \num{1.000110} & \num{1.000039} & \num{1.000014} & \num{1.000005}\\
      \hline
    \end{tabular}
  \end{table}
\end{corollary}

\subsection{Remark}

In this paper, we show mean-field behavior for nearest-neighbor
oriented percolation on $\opSpaceTime[d\geq 9]$.
As compared to the analysis on $\mathbb{Z}^{d\gg 4}\times\opTime$ by Nguyen and Yang~\cite{ny93},
the current analysis identifies an upper bound $9$ on the critical dimension $d_\mathrm{c}$.
To begin with, since \cite{ny93} contains a problematic issue due to the existence of $k$ such that $\ft D(k) = -1$ (as we mentioned in the last paragraph of Section~\ref{sec:motivation}),
their bound corresponding to \eqref{eq:improvedBound-g2} below requires extra lemmas for any dimensions even if $d$ is sufficiently large.
To resolve the issue, we prove Proposition~\ref{prp:restricting-g2} and Lemma~\ref{lem:upperBound-Green} below.
It is crucial to reduce the range of the supremum in $g_2$ (See Section~\ref{sec:proof-main}) by symmetry in the Fourier space, hence the assumption of the inequality \eqref{eq:upperBound-Green} is satisfied.

To go down to the desired spatial $5$ dimensions, we must improve our analysis in various aspects.
Some of the ideas are summarized as follows.
\begin{enumerate}[label=(\roman*), listparindent=\parindent]
  \item
  As we explain in Remark~\ref{rem:computer-assisted-proof} below, the spatial dimension $9$ in Theorem~\ref{thm:infraredBound} depends on how to look for the parameters $\{K_i\}_{i=1}^{3}$ of a bootstrap argument (See Section~\ref{sec:proof-main} for details).
  We prove the main theorem only in $d+1\geq 9+1$ because we were not able to find the parameters satisfying the bootstrap argument in $d+1=8+1$.
  There is a possibility that one can update our result if one searches for the parameters carefully, but it is more important to improve our bounds for $\{g_i\}_{i=1}^{3}$, the lace expansion coefficients, and so on than accuracy in computers.

  \item
  In Lemma~\ref{lem:precise-diagrammaticBounds} below, we paid attention to the coefficient $1/2$ of the diagram containing four $q_p$'s.
  If one isolates the diagrams containing six $q_p$'s and specifies their coefficients,
  then one can obtain better upper bounds than those in this paper.
  The upper bounds on $B_{p, m}^{(\lambda, \rho)}$ and $T_{p, m}^{(\lambda, \rho)}$ in Lemma~\ref{lem:basic-diagrams} below
  may be helpful for the improvement.
  However, although their random-walk counterparts, such as $\sum_{\vb*{x}} (q_p^{\Conv\lambda}\Conv\varphi_p^{\Conv 2})(\vb*{x}) (q_p^{\Conv\rho}\Conv\varphi_p)(\vb*{x})$,
  are decreasing in $\lambda$ and $\rho$,
  the bound on $T_{p, m}^{(\lambda, \rho)}$ may attain the minimum at some $\lambda_*, \rho_*\in\mathbb{N}$,
  due to the exponentially growing factor $(p (m \vee 1))^{\lambda + \rho}$.
  That is a reason why we did not isolate the diagrams containing six $q_p$'s in the paper.

  \item
  If one does not use the Laplace transform, then our result may be improved.
  We were not able to avoid using bad triangle inequality estimates due to complex numbers,
  e.g., not splitting contributions of $\Pi_p^\mathrm{even}$ and $\Pi_p^\mathrm{odd}$ in the second half in Lemma~\ref{lem:precise-diagrammaticBounds} below.
  We would like to neglect either $\Pi_p^\mathrm{even}$ or $\Pi_p^\mathrm{odd}$ because the odd terms in \eqref{eq:def-sumOfLaceExpansionCoefficients} below are negative.
  To this end, the inductive approach to the lace expansion \cite{hs02,hs03} may be usable.
  However, this method is shown for spread-out models.
  An extension to nearest-neighbor models is required.

  \item
  Even if some improvements are found for ordinary percolation, such methods are not always applicable to oriented percolation.
  For example, the non-backtracking lace expansion~\cite{fh17nobl,fh17p} achieved success for ordinary percolation.
  Its expansion coefficients express the perturbation of the non-backtracking random walk while the (ordinary) lace expansion coefficients express that of the random walk.
  The non-backtracking random walk is useful to approximate ordinary-percolation clusters, but we guess that the property is not useful to approximate oriented-percolation clusters.
  Any paths on an oriented-percolation cluster are not self-avoiding paths when we consider their projection onto the space.
  Therefore, different methods between ordinary and oriented percolation are often desired.
\end{enumerate}

\subsection{Proof of the main result}\label{sec:proof-main}

Our proof of Theorem~\ref{thm:infraredBound} is based on a bootstrap analysis.
In this paper, we do not take care of the details of the method.
See either \cite[Lemma~5.9]{s06} or \cite[Lemma~8.1]{hh17}.
Thanks to it, it suffices for us to verify the following three propositions,
of which in particular we use the lace expansion in the third one.

Let
\begin{multline}
  \flt U_\mu(k, l) = \bigl(1 - \hat D(k)\bigr) \Biggl(\frac{\abs*{\ft S_{\mu}(l + k)} + \abs*{\ft S_{\mu}(l - k)}}{2} \abs{\ft S_{\mu}(l)}\\
  + \bigl(1 - \hat D(2l)\bigr) \abs{\ft S_{\mu}(l)} \abs{\ft S_{\mu}(l + k)} \abs{\ft S_{\mu}(l - k)}\Biggr)
  \label{eq:2nd-derivative-RW}
\end{multline}
for $k, l\in\mathbb{T}^d$ and $\mu\in\mathbb{C}$ with $\abs{\mu}\in\interval{0}{1}$, and let
\begin{equation}
  \label{eq:scaledParameter}
  \mu_p(z) = \qty(1 - \flt\varphi_p(0, \abs{z})^{-1}) \Napier^{\imag \arg z}
\end{equation}
for $p\in\interval[open right]{0}{p_\mathrm{c}}$ and $z\in\mathbb{C}$ with $\abs{z}\in\interval[open right]{1}{m_p}$.
We define the bootstrap functions $\{g_i\}_{i=1}^{3}$ as
\begin{gather}
  g_1(p, m) \defeq p (m \vee 1),\\
  g_2(p, m) \defeq \sup_{\substack{k\in\mathbb{T}^d,\\ z\in\mathbb{C}\colon\abs{z}\in\set{1, m}}}\frac{\abs{\flt\varphi_p(k, z)}}{\abs*{\flt S_{\mu_p(z)}(k)}},\\
  g_3(p, m) \defeq \sup_{\substack{k, l\in\mathbb{T}^d,\\ z\in\mathbb{C}\colon\abs{z}\in\set{1, m}}} \frac{\abs\big{\frac{1}{2}\DLaplacian{k}\bigl(\flt q_p(l, z) \flt\varphi_p(l, z)\bigr)}}{\flt U_{\mu_p(z)}(k, l)},
  \label{eq:op-bootstrapFunction3}
\end{gather}
where
\[
  \DLaplacian{k}\ft f(l) = \ft f(l + k) + \ft f(l - k) - 2\ft f(l)
\]
for a function $\ft f$ on $\mathbb{T}^d$ is the double discrete derivative.
In \eqref{eq:op-bootstrapFunction3}, the supremum near $k=0$ should be interpreted as the supremum over the limit as $\norm{k}_2\to 0$.
The following propositions are the sufficient condition of \cite[Lemma~5.9]{s06}.
Thus, if Propositions~\ref{prp:op-continuity}--\ref{prp:op-bootstrapArgument} hold, then
it is true that there exists $K_i$ such that $g_i(p, m) < K_i$ for each $i=1, 2, 3$.
In particular, the bound on $g_2(p, m)$ and Lemma~\ref{lem:mu-bound} below immediately imply Theorem~\ref{thm:infraredBound}.

\begin{proposition}[Continuity]\label{prp:op-continuity}
  The functions $\{g_i\}_{i=1}^{3}$ are continuous
  in $m\in\interval[open right]{1}{m_p}$ for every $p\in\interval[open right]{0}{p_\mathrm{c}}$,
  and the functions $\{g_i(p, 1)\}_{i=1}^{3}$ are continuous in $p\in\interval[open right]{0}{p_\mathrm{c}}$.
\end{proposition}

\begin{proposition}[Initial conditions]\label{prp:op-initialCondition}
  There are model-dependent finite constants $\{K_i\}_{i=1}^{3}$ such that $g_i(0, 1)<K_i$ for $i=1, 2, 3$.
\end{proposition}

\begin{proposition}[Bootstrap argument]\label{prp:op-bootstrapArgument}
  For nearest-neighbor oriented percolation on $\opSpaceTime[d\geq 9]$,
  we fix both $p\in\interval[open]{0}{p_\mathrm{c}}$ and $z\in\mathbb{C}$ with $\abs{z}\in\interval[open]{1}{m_p}$
  and assume $g_i(p, m)\leq K_i$ for each $i=1, 2, 3$, where $\{K_i\}_{i=1}^{3}$ are the same constants as in Proposition~\ref{prp:op-initialCondition}.
  Then, the stronger inequalities $g_i(p, m)<K_i$, $i=1, 2, 3$, hold.
\end{proposition}

When the constants in Proposition~\ref{prp:op-initialCondition} are denoted by $\{\bar K_i\}_{i=1}^{3}$, it is also possible to take $\{K_i\}_{i=1}^{3}$ in Proposition~\ref{prp:op-bootstrapArgument} such that $K_i \geq \bar K_i$.
However, since we would like to obtain upper bounds as precisely as possible in the vicinity of the critical point $p_\mathrm{c}$, the option is not very useful to bound $\{g_i(p, 1)\}_{i=1}^{3}$ above.
In many cases, taking the same constants as in the initial conditions is the best selection.

It is not hard to apply the method in \cite[Section~3]{Handa-Kamijima-Sakai} to Proposition~\ref{prp:op-continuity},
hence we omit the proof.
One can easily verify the continuities of $g_1(p, m)$ and $g_2(p, m)$.
For $g_3(p, m)$, notice that we use the Markov property for oriented percolation instead of the nested expectations for ordinary percolation.
Thanks to this, the proof is a little simplified.
One can see the complete proof in \cite{k21}.
The proof of Proposition~\ref{prp:op-initialCondition} is also easy, so that we provide it here:

\begin{proof}[Proof of Proposition~\ref{prp:op-initialCondition}]
  Clearly, $g_1(0, 1) = 0$.  By definitions, $\flt\varphi_0(k, z) = \sum_{(x, t)} (\KroneckerDelta{o}{x}\KroneckerDelta{0}{t} + \mathbb{P}_0((o, 0) \rightarrow (x, t) \neq (o, 0))) \Napier^{\imag k\cdot x} z^t = 1$,
  $\mu_0(z) = (1 - \flt\varphi_0(0, \abs{z})) \Napier^{\imag \arg z} = 0$
  and hence $\ft S_{\mu_0(z)}(k) = 1$.
  These lead to $g_2(0, 1) = 1$ and $g_3(0, 1) = 0$.
  Therefore, the initial conditions hold for arbitrary $K_i>1$, $i=1, 2, 3$.
\end{proof}

We must take the values $\{K_i\}_{i=1}^{3}$ satisfying Proposition~\ref{prp:op-bootstrapArgument}.
See the proof of Proposition~\ref{prp:op-bootstrapArgument} in Section~\ref{sec:proof-bootstrapArgument} for the specific values.

In the rest of this paper, we focus on the proof of Proposition~\ref{prp:op-bootstrapArgument},
which is shown in the following steps:
\begin{enumerate}
  \item\label{enum:improvedBound}
    Bound the bootstrap functions $\{g_i(p, m)\}_{i=1}^{3}$ above by the lace expansion.
  \item\label{enum:diagrammaticBound}
    Bound the lace expansion coefficients ($\{\pi_p^{(N)}\}_{N=0}^{\infty}$) above
    by basic diagrams ($B_{p, m}^{(\lambda, \rho)}$, $T_{p, m}^{(\lambda, \rho)}$ and $\ft V_{p, m}^{(\lambda, \rho)}(k)$).
  \item\label{enum:basicDiagram}
    Bound the basic diagrams above by the bootstrap hypotheses
    $g_i(p, m) \leq K_i$ for $i=1, 2, 3$ and the random-walk quantities \eqref{eq:def-rwQuantities}.
  \item\label{enum:completingbootstrap}
    Verify the stronger inequalities $g_i(p, m) < K_i$ for $i=1, 2, 3$
    by combining the bounds in Steps~\ref{enum:improvedBound}--\ref{enum:basicDiagram} and
    substituting specifically numerical values into $\{K_i\}_{i=1}^{3}$.
\end{enumerate}
We explain Step~\ref{enum:improvedBound} in Section~\ref{sec:improvedBounds},
Step~\ref{enum:diagrammaticBound} in Section~\ref{sec:diagrammaticBounds},
Step~\ref{enum:basicDiagram} in Section~\ref{sec:bounds-basicDiagrams}
and Step~\ref{enum:completingbootstrap} in Section~\ref{sec:proof-bootstrapArgument}.

\section{Diagrammatic bounds on the expansion coefficients}\label{sec:diagrammaticBounds}

\subsection{Review of the lace expansion}

The lace expansion was first derived by Brydges and Spencer~\cite{bs85} for weakly self-avoiding walk.
In the literature, there are three different ways to derive the lace expansion for oriented percolation.
The first is to directly apply Brydges and Spencer's method due to the Markov property~\cite{ny93},
the second is to use inclusion-exclusion relations and nested expectations~\cite{hs90p},
and the third is to use inclusion-exclusion relations and the Markov property~\cite{s01}.
In this paper, although we do not show its proof,
we implicitly use the third approach and its representations of $\{\pi_p^{(N)}\}_{N=0}^{\infty}$.

The lace expansion gives a similar recursion equation for the two-point function $\varphi_p$
to that for the random-walk two-point function $Q_p$, which is
\[
  Q_p(\vb*{x}) = \KroneckerDelta{\vb*{o}}{\vb*{x}} + (q_p\Conv Q_p)(\vb*{x})
\]
for $\vb*{x}=(x, t)\in\opSpaceTime$, where $\KroneckerDelta{\vb*{o}}{\vb*{x}} = \KroneckerDelta{o}{x}\KroneckerDelta{o}{t}$
(cf., the recursion equation for the random-walk Green function $S_p$).

\begin{proposition}[Lace expansion]
  For any $p<p_\mathrm{c}$ and $N\in\nnInt$, there exist model-dependent nonnegative functions $\{\pi_p^{(n)}\}_{n=0}^{N}$
  on $\opSpaceTime$ such that, if we define $\Pi_p^{(N)}$ as
  \[
    \Pi_p^{(N)}(\vb*{x}) = \sum_{n=0}^{N}(-1)^n \pi_p^{(n)}(\vb*{x}),
  \]
  we obtain the recursion equation
  \begin{equation}
    \label{eq:op-laceExpansion}
    \varphi_p(\vb*{x}) = \KroneckerDelta{o}{\vb*{x}} + \Pi_p^{(N)}(\vb*{x})
      + \qty(\qty\big(\KroneckerDelta{\vb*{o}}{\bullet} + \Pi_p^{(N)}) \Conv q_p \Conv \varphi_p)(\vb*{x})
      + (-1)^{N+1} R_p^{(N+1)}(\vb*{x}),
  \end{equation}
  where the remainder $R_p^{(N+1)}(\vb*{x})$ obeys the bound
  \begin{equation*}
    0 \leq R_p^{(N+1)}(\vb*{x}) \leq \qty(\pi_p^{(N)} \Conv \varphi_p)(\vb*{x}).
  \end{equation*}
\end{proposition}

Assume that $\lim_{N\to\infty}\pi_p^{(N)}(\vb*{x})=0$ for every $\vb*{x}\in\opSpaceTime$.
This will be proved to hold by the absolute convergence of the alternating series of $\{\pi_p^{(N)}\}_{N=1}^{\infty}$
from Lemma~\ref{lem:precise-diagrammaticBounds} below.
Then, $R_p^{(N)}(\vb*{x}) \to 0$ as $N \to \infty$ for every $\vb*{x}\in\opSpaceTime$, so that the limit of the series
\begin{equation}
  \label{eq:def-sumOfLaceExpansionCoefficients}
  \Pi_p(\vb*{x}) \defeq \lim_{N\to\infty} \Pi_p^{(N)}(\vb*{x}) = \sum_{N=0}^{\infty}(-1)^N \pi_p^{(N)}(\vb*{x})
\end{equation}
is well-defined.

To state the precise expression of $\{\pi_p^{(N)}\}_{N=0}^{\infty}$, we introduce some notation.
Below, $\underline{\vb*{b}}$ and $\overline{\vb*{b}}$ denote the bottom and top of a bond $\vb*{b}$, respectively,
that is $\vb*{b}=(\underline{\vb*{b}}, \overline{\vb*{b}})$.
\begin{definition}
  Fix a bond configuration
  and let $\vb*{x}, \vb*{y}, \vb*{u}, \vb*{v}\in\opSpaceTime$.
  \begin{enumerate}[label=(\roman*)]
    \item
      Given a bond $\vb*{b}$, we define $\tilde{\mathcal{C}}^{\vb*{b}}(\vb*{x})$ to be the set of vertices
      connected to $\vb*{x}$ in the new configuration obtained by setting $\vb*{b}$ to be vacant.
    \item
      We say that a bond $(\vb*{u}, \vb*{v})$ is pivotal
      for $\vb*{x} \rightarrow \vb*{y}$
      if $\vb*{x} \rightarrow \vb*{u}$ occurs in $\tilde{\mathcal{C}}^{(\vb*{u}, \vb*{v})}(\vb*{x})$
      (i.e., $\vb*{x}$ is connected to $\vb*{u}$ without using $(\vb*{u}, \vb*{v})$)
      and if $\vb*{v}\rightarrow \vb*{y}$ occurs in the complement of $\tilde{\mathcal{C}}^{(\vb*{u}, \vb*{v})}(\vb*{x})$.
      Let $\mathtt{piv}(\vb*{x}, \vb*{y})$ be the set of pivotal bonds for the connection
      from $\vb*{x}$ to $\vb*{y}$.
    \item
      We say that $\vb*{x}$ is doubly connected to $\vb*{y}$, denoted by $\vb*{x}\Rightarrow\vb*{y}$,
      if either $\vb*{x}=\vb*{y}$ or $\vb*{x}\rightarrow\vb*{y}$ and $\mathtt{piv}(\vb*{x}, \vb*{y})=\varnothing$.
  \end{enumerate}
\end{definition}

Let
\[
  \tilde E_{\va*{b}_N}^{(N)}(\vb*{x}) = \Set{\vb*{o} \Rightarrow \underline{\vb*{b}}_1} \cap \bigcap_{i=1}^{N} E\qty(\vb*{b}_i, \underline{\vb*{b}}_{i+1}; \tilde{\mathcal{C}}^{\vb*{b}_i}(\overline{\vb*{b}}_{i-1})),
\]
where $\va*{b}_N = (\vb*{b}_1, \dots, \vb*{b}_N)$ is an ordered set of bonds on $\opSpaceTime$, $\underline{\vb*{b}}_0 = \vb*{o}$, $\overline{\vb*{b}}_{N+1} = \vb*{x}$ and
\begin{equation*}
  E(\vb*{b}, \vb*{x}; A) = \underbrace{\Set{\vb*{b} \rightarrow \vb*{x} \in A}}_{= \set{\vb*{b}\text{ is occupied}} \cap \set{\overline{\vb*{b}} \rightarrow \vb*{x}} \cap \set{\vb*{x}\in A}} \setminus \bigcup_{\vb*{b}'\in\mathtt{piv}(\overline{\vb*{b}}, \vb*{x})} \Set{\underline{\vb*{b}}' \in A}
\end{equation*}
for a bond $\vb*{b}$, a vertex $\vb*{x}\in\opSpaceTime$ and a set $A\subset\opSpaceTime$.
According to \cite[Proposition~4.1]{s01}, the lace expansion coefficients $\{\pi_p^{(N)}\}_{N=0}^{\infty}$ are defined as
\begin{equation*}
  \pi_p^{(N)}(\vb*{x}) =
  \begin{dcases}
    \mathbb{P}_p\qty(\vb*{o} \Rightarrow \vb*{x}) - \KroneckerDelta{\vb*{o}}{\vb*{x}} & [N=0],\\
    \sum_{\va*{b}_N} \mathbb{P}_p\qty(\tilde E_{\va*{b}_N}^{(N)}(\vb*{x})) & [N\geq 1].
  \end{dcases}
\end{equation*}

\subsection{Diagrammatic bounds}

We state upper bounds on the Fourier-Laplace transform of the sum
\eqref{eq:def-sumOfLaceExpansionCoefficients} of the alternating series of the lace expansion coefficients
in terms of basic diagrams.
To do so, we first introduce notations.
Let $\varphi_p^{(m)}(x, t) = m^t \varphi_p(x, t)$.
Recall the notation $\vb*{x}\in\opSpaceTime$.
We define basic diagrams as, for $\lambda, \rho\in\mathbb{N}$, $m>0$ and $k\in\mathbb{T}^d$,
\begin{gather}
  B_{p, m}^{(\lambda, \rho)} \defeq \sup_{\vb*{x}}\sum_{\vb*{y}} \qty(q_p^{\Conv \lambda} \Conv \varphi_p)(\vb*{y}) \qty(m^\rho q_p^{\Conv \rho} \Conv \varphi_p^{(m)})(\vb*{y} - \vb*{x}),
  \label{eq:def-bubble}\\
  T_{p, m}^{(\lambda, \rho)} \defeq \sup_{\vb*{x}}\sum_{\vb*{y}} \qty(q_p^{\Conv \lambda} \Conv \varphi_p^{\Conv 2})(\vb*{y}) \qty(m^\rho q_p^{\Conv \rho} \Conv \varphi_p^{(m)})(\vb*{y} - \vb*{x}),
  \label{eq:def-triangle}\\
  \flt V_{p, m}^{(\lambda, \rho)}(k) \defeq \sup_{\vb*{x}}\sum_{\vb*{y}} \qty(q_p^{\Conv \lambda} \Conv \varphi_p)(\vb*{y}) \qty(1 - \cos k\cdot y) \qty(m^\rho q_p^{\Conv \rho} \Conv \varphi_p^{(m)})(\vb*{y} - \vb*{x}).
  \label{eq:def-weighted-bubble}
\end{gather}
We represent the transition probability $q_p$ by a pair of either parallel lines or a line and a dot,
and we represent the two-point function $\varphi_p$ by a line segment.
For example,
\begin{equation}
  \label{eq:diagrammaticRepresentation}
  \begin{gathered}
    \varphi_p(\vb*{x}) =
    \begin{tikzpicture}[op diagram]
      \laceDraw (0,0) (0,2);
      \lacePutLabel[anchor=north] (0,0) {$\vb*{o}$};
      \lacePutLabel[anchor=south] (0,2) {$\vb*{x}$};
    \end{tikzpicture},\qquad
    \qty(q_p^{\Conv 2} \Conv \varphi_p)(\vb*{x}) =
    \begin{tikzpicture}[op diagram]
      \laceDraw[first=2] (0,0) (0,2);
      \lacePutLabel[anchor=north] (0,0) {$\vb*{o}$};
      \lacePutLabel[anchor=south] (0,2) {$\vb*{x}$};
    \end{tikzpicture},\\
    B_{p, m}^{(1, 2)} = \sup_{\vb*{x}}
    \begin{tikzpicture}[op diagram]
      \coordinate (O) at (-1.3,0);
      \coordinate (X) at (1.3,0.5);
      \coordinate (Y) at (0,2.5);
      \laceDraw[first=1, endpoint-shape=line] (O) (Y) node[vertex] {};
      \laceDraw[first=2, endpoint-shape=line] (X) (Y);
      \draw[decorate, decoration={brace, amplitude=5pt, raise=0.7ex}] (Y) -- (X) node [midway, sloped, above, yshift=1.5ex] {$m$};
      \node[anchor=north] at (O) {$\vb*{o}$};
      \node[anchor=north] at (X) {$\vb*{x}$};
    \end{tikzpicture},\qquad
    T_{p, m}^{(2, 1)} = \sup_{\vb*{x}}
    \begin{tikzpicture}[op diagram]
      \coordinate (O) at (-1.3,0);
      \coordinate (X) at (1.3,0.5);
      \coordinate (Y) at (1.3,2.5);
      \coordinate (W) at (-1.3,2);
      \laceDraw[first=2, endpoint-shape=line] (O) (W) node[vertex] {};
      \draw (W) -- (Y) node[vertex] {};
      \laceDraw[first=1, endpoint-shape=line] (X) (Y);
      \draw[decorate, decoration={brace, amplitude=5pt, raise=0.7ex}] (Y) -- (X) node [midway, sloped, above, yshift=1.5ex] {$m$};
      \node[anchor=north] at (O) {$\vb*{o}$};
      \node[anchor=north] at (X) {$\vb*{x}$};
    \end{tikzpicture},\qquad
    \flt V_{p, m}^{(2, 2)}(k) = \sup_{\vb*{x}}
    \begin{tikzpicture}[op diagram]
      \coordinate (O) at (-1.3,0);
      \coordinate (X) at (1.3,0.5);
      \coordinate (Y) at (0,2.5);
      \laceDraw[first=2, endpoint-shape=line] (O) (Y) node[vertex] {};
      \laceDraw[first=2, endpoint-shape=line] (X) (Y);
      \draw[decorate, decoration={brace, amplitude=5pt, raise=0.7ex}] (O) -- (Y) node [midway, sloped, above, yshift=1.5ex] {$1 - \cos k\cdot$};
      \draw[decorate, decoration={brace, amplitude=5pt, raise=0.7ex}] (Y) -- (X) node [midway, sloped, above, yshift=1.5ex] {$m$};
      \node[anchor=north] at (O) {$\vb*{o}$};
      \node[anchor=north] at (X) {$\vb*{x}$};
    \end{tikzpicture},
  \end{gathered}
\end{equation}
where the unlabeled vertices (short lines and dots) are summed over $\opSpaceTime$.
The segments emphasized by the braces mean weighted two-point functions or weighted space-time transition probabilities: $\varphi_p$ or $q_p$ multiplied by $m$'s or $1 - \cos k\cdot\bullet$.
Time increases from the beginning point to the ending point of a line.
This representation is useful to bound $\{\pi_p^{(N)}\}_{N=0}^{\infty}$ above.
Moreover, we divide $\Pi_p(\vb*{x})$ into two parts:
\begin{gather*}
  \Pi_p^\mathrm{even}(\vb*{x}) \coloneqq \sum_{N=1}^{\infty}\pi_p^{(2N)}(\vb*{x}),\\
  \Pi_p^\mathrm{odd}(\vb*{x}) \coloneqq \pi_p^{(1)}(\vb*{x}) - \pi_p^{(0)}(\vb*{x}) + \sum_{N=1}^{\infty}\pi_p^{(2N+1)}(\vb*{x}).
\end{gather*}
Although the latter contains the zeroth lace expansion coefficient, we name it ``odd'' for convenience.

Next, we show upper bounds on the lace expansion coefficients $\{\pi_p^{(N)}\}_{N=0}^{\infty}$.
Taking the sum of the bounds in Lemma~\ref{lem:ksp-bound} over $N\in\nnInt$ immediately implies Lemma~\ref{lem:precise-diagrammaticBounds}.
We state only their statements in this subsection, and we prove them in the next subsection.

\begin{lemma}\label{lem:ksp-bound}
  Let $N$ be an integer greater than or equal to $3$.
  When we multiply the upper bounds in Lemma~\ref{lem:xsp-bound} by $m^t$,
  \begin{align}
    \label{eq:sum-diagrammaticBounds-m-01}
    &\begin{aligned}[b]
      \abs{\flt\pi_p^{(0)}(0, m) - \flt\pi_p^{(1)}(0, m)} \leq {} \MoveEqLeft[0]
        \frac{1}{2} B_{p, m}^{(2, 2)} + B_{p, 1}^{(2, 2)} B_{p, m}^{(0, 2)} + \frac{3}{2} B_{p, 1}^{(2, 2)} B_{p, m}^{(1, 3)}\\
        & + 3 \qty(B_{p, m}^{(2, 2)})^2 + 3 B_{p, m}^{(2, 1)} T_{p, m}^{(2, 2)},
    \end{aligned}\\
    \label{eq:sum-diagrammaticBounds-m-2}
    &\begin{aligned}[b]
      \flt\pi_p^{(2)}(0, m) \leq {} \MoveEqLeft[0]
        B_{p, m}^{(2, 2)} B_{p, m}^{(1, 3)} + 2 \qty(B_{p, m}^{(2, 2)})^2 + 2 B_{p, m}^{(2, 1)} T_{p, m}^{(2, 2)}\\
        & + \frac{1}{2} B_{p, 1}^{(2, 2)} B_{p, m}^{(1, 3)} T_{p, m}^{(1, 2)} + B_{p, m}^{(2, 1)} T_{p, m}^{(1, 1)} T_{p, m}^{(1, 2)}\\
        & + \frac{1}{2} B_{p, 1}^{(2, 2)} B_{p, m}^{(1, 3)} T_{p, m}^{(2, 1)} + B_{p, m}^{(2, 1)} T_{p, m}^{(1, 1)} T_{p, m}^{(2, 1)},
    \end{aligned}\\
    \label{eq:sum-diagrammaticBounds-m-N}
    &\flt\pi_p^{(N)}(0, m) \leq
      \qty(B_{p, m}^{(1, 1)} + \frac{1}{2} B_{p, 1}^{(2, 2)} B_{p, m}^{(1, 3)} + T_{p, m}^{(1, 1)} B_{p, m}^{(1, 1)}) \qty(2 T_{p, m}^{(1, 1)})^{N - 1}.
  \end{align}
  When we multiply the upper bounds in Lemma~\ref{lem:xsp-bound} by $m^t$ and $t$,
  \begin{align}
    \label{eq:sum-diagrammaticBounds-mt-01}
    &\begin{aligned}[b]
      \MoveEqLeft \sum_{(x, t)} \abs{\pi_p^{(0)}(x, t) - \pi_p^{(1)}(x, t)} m^t t \leq
        \frac{1}{2} \qty(B_{p, m}^{(2, 2)} + T_{p, m}^{(2, 2)}) + B_{p, 1}^{(2, 2)} \qty(B_{p, m}^{(0, 2)} + T_{p, m}^{(0, 2)})\\
        & + \frac{3}{2} B_{p, 1}^{(2, 2)} \qty(2 B_{p, m}^{(1, 3)} + T_{p, m}^{(1, 3)})\\
        & + 6 B_{p, m}^{(2, 2)} \qty(B_{p, m}^{(2, 2)} + T_{p, m}^{(2, 2)})
        + 3 \qty(m T_{p, m}^{(1, 2)} T_{p, 1}^{(2, 2)} + B_{p, m}^{(1, 2)} T_{p, m}^{(2, 2)} + T_{p, m}^{(2, 2)} T_{p, m}^{(2, 1)}),
    \end{aligned}\\
    \label{eq:sum-diagrammaticBounds-mt-2}
    &\begin{aligned}[b]
      \MoveEqLeft \sum_{(x, t)} \pi_p^{(2)}(x, t) m^t t \leq
        B_{p, 1}^{(2, 2)} \qty(2 B_{p, m}^{(1, 3)} + T_{p, m}^{(1, 3)})\\
        & + 4 B_{p, m}^{(2, 2)} \qty(B_{p, m}^{(2, 2)} + T_{p, m}^{(2, 2)})
        + 2 \qty(m T_{p, m}^{(1, 2)} T_{p, 1}^{(2, 2)} + B_{p, m}^{(1, 2)} T_{p, m}^{(2, 2)} + T_{p, m}^{(2, 2)} T_{p, m}^{(2, 1)})\\
        & + \frac{1}{2} B_{p, 1}^{(2, 2)} \qty(3 B_{p, m}^{(1, 3)} T_{p, m}^{(1, 2)} + 2 T_{p, m}^{(1, 3)} T_{p, m}^{(1, 2)})
        + \qty(3 T_{p, m}^{(1, 2)} T_{p, m}^{(1, 1)} T_{p, m}^{(2, 1)} + B_{p, m}^{(1, 2)} T_{p, m}^{(1, 1)} T_{p, m}^{(2, 1)})\\
        & + \frac{1}{2} B_{p, 1}^{(2, 2)} \qty(2 B_{p, m}^{(1, 3)} T_{p, m}^{(2, 1)} + 2 T_{p, m}^{(1, 3)} T_{p, m}^{(2, 1)})
        + 3 T_{p, m}^{(2, 1)} T_{p, m}^{(1, 1)} T_{p, m}^{(2, 1)},
    \end{aligned}\\
    \label{eq:sum-diagrammaticBounds-mt-N}
    &\begin{aligned}[b]
      \MoveEqLeft \sum_{(x, t)} \pi_p^{(N)}(x, t) m^t t\\
      \leq {} &
        \qty(T_{p, m}^{(1, 1)} + \frac{1}{2} (2 B_{p, m}^{(2, 2)} B_{p, m}^{(1, 3)} + B_{p, m}^{(2, 2)} T_{p, m}^{(3, 1)}) + 2 \qty(T_{p, m}^{(1, 1)})^2) \qty(2 T_{p, m}^{(1, 1)})^{N - 1}\\
        & + \qty(N - 2) \qty(T_{p, m}^{(1, 1)} + \frac{1}{2} B_{p, m}^{(2, 2)} T_{p, m}^{(1, 3)} + \qty(T_{p, m}^{(1, 1)})^2) \qty(2 T_{p, m}^{(1, 1)})^{N - 2}\\
        & + \qty(T_{p, m}^{(1, 1)} + \frac{1}{2} B_{p, m}^{(2, 2)} T_{p, m}^{(1, 3)} + \qty(T_{p, m}^{(1, 1)})^2) \qty(2 T_{p, m}^{(1, 1)})^{N - 1}.
    \end{aligned}
  \end{align}
  When we multiply the upper bounds in Lemma~\ref{lem:xsp-bound} by $m^t$ and $1 - \cos k\cdot x$,
  \begin{align}
    \label{eq:sum-diagrammaticBounds-mcos-01}
    &\begin{aligned}[b]
      \MoveEqLeft \sum_{(x, t)} \abs{\pi_p^{(0)}(x, t) - \pi_p^{(1)}(x, t)} m^t \qty(1 - \cos k\cdot x)\\
      \leq {} &
        \frac{1}{2} \flt V_{p, m}^{(2, 2)}(k) + 2 \qty(\flt V_{p, 1}^{(2, 2)}(k) B_{p, m}^{(0, 2)} + B_{p, 1}^{(2, 2)} \flt V_{p, m}^{(1, 2)}(k))\\
        & + \frac{3}{2} B_{p, m}^{(2, 2)} \flt V_{p, m}^{(3, 1)}(k)\\
        & + 12 B_{p, m}^{(2, 2)} \flt V_{p, m}^{(2, 2)}(k)
        + 6 \qty(\flt V_{p, m}^{(1, 2)}(k) T_{p, m}^{(2, 2)} + T_{p, m}^{(2, 2)} \flt V_{p, m}^{(2, 1)}(k)),
    \end{aligned}\\
    \label{eq:sum-diagrammaticBounds-mcos-2}
    &\begin{aligned}[b]
      \MoveEqLeft \flt\pi_p^{(2)}(0, m) - \flt{\pi}_p^{(2)}(k, m)
      \leq
        B_{p, m}^{(2, 2)} \flt V_{p, m}^{(3, 1)}(k)\\
        & + 8 B_{p, m}^{(2, 2)} \flt V_{p, m}^{(2, 2)}(k)
        + 4 \qty(\flt V_{p, m}^{(1, 2)}(k) T_{p, m}^{(2, 2)} + T_{p, m}^{(2, 2)} \flt V_{p, m}^{(2, 1)}(k))\\
        & + B_{p, m}^{(2, 2)} \qty(\flt V_{p, m}^{(3, 1)}(k) T_{p, m}^{(2, 1)} + T_{p, m}^{(3, 1)} \flt V_{p, m}^{(2, 1)}(k))\\
        & + 3 \qty(\flt V_{p, m}^{(1, 2)}(k) T_{p, m}^{(1, 1)} T_{p, m}^{(2, 1)} + T_{p, m}^{(1, 2)} \flt V_{p, m}^{(1, 1)}(k) T_{p, m}^{(2, 1)} + T_{p, m}^{(1, 2)} T_{p, m}^{(1, 1)} \flt V_{p, m}^{(2, 1)}(k))\\
        & + B_{p, m}^{(2, 2)} \qty(\flt V_{p, m}^{(3, 1)}(k) T_{p, m}^{(1, 2)} + T_{p, m}^{(3, 1)} \flt V_{p, m}^{(1, 2)}(k))\\
        & + 3 \qty(\flt V_{p, m}^{(1, 2)}(k) T_{p, m}^{(1, 1)} T_{p, m}^{(1, 2)} + T_{p, m}^{(1, 2)} \flt V_{p, m}^{(1, 1)}(k) T_{p, m}^{(1, 2)} + T_{p, m}^{(1, 2)} T_{p, m}^{(1, 1)} \flt V_{p, m}^{(1, 2)}(k)),
    \end{aligned}\\
    \label{eq:sum-diagrammaticBounds-mcos-N}
    &\begin{aligned}[b]
      \MoveEqLeft \flt\pi_p^{(N)}(0, m) - \flt{\pi}_p^{(N)}(k, m)\\
      \leq {} &
        \qty(N + 1)
        \Biggl(
          \qty(\flt V_{p, m}^{(1, 1)}(k) + \frac{1}{2} B_{p, m}^{(2, 2)} \flt V_{p, m}^{(1, 3)}(k) + 2 T_{p, m}^{(1, 1)} \flt V_{p, m}^{(1, 1)}(k)) \qty(2 T_{p, m}^{(1, 1)})^{N - 1}\\
          & + 2 \qty(N - 2) \qty(T_{p, m}^{(1, 1)} + \frac{1}{2} B_{p, m}^{(2, 2)} T_{p, m}^{(1, 3)} + \qty(T_{p, m}^{(1, 1)})^2) \flt V_{p, m}^{(1, 1)}(k) \qty(2 T_{p, m}^{(1, 1)})^{N - 2}\\
          & + 2 \qty(T_{p, m}^{(1, 1)} + \frac{1}{2} B_{p, m}^{(2, 2)} T_{p, m}^{(1, 3)} + \qty(T_{p, m}^{(1, 1)})^2) \qty(2 T_{p, m}^{(1, 1)})^{N - 2} \flt V_{p, m}^{(1, 1)}(k)
        \Biggr).
    \end{aligned}
  \end{align}
\end{lemma}

\begin{lemma}\label{lem:precise-diagrammaticBounds}
  Suppose that $2 T_{p, m}^{(1, 1)} < 1$.
  Then, the alternative series \eqref{eq:def-sumOfLaceExpansionCoefficients} absolutely converges, and the following inequalities hold for $p< p_\mathrm{c}$, $m<m_p$ and $k\in\mathbb{T}^d$:
  \begin{align}
    &\begin{aligned}[b]
      \MoveEqLeft \flt\Pi_p^\mathrm{even}(0, m) \leq
      \flt\pi_p^{(2)}(0, m)\\
      & + \qty(B_{p, m}^{(1, 1)} + \frac{1}{2} B_{p, 1}^{(2, 2)} B_{p, m}^{(1, 3)} + T_{p, m}^{(1, 1)} B_{p, m}^{(1, 1)}) \frac{\qty\big(2 T_{p, m}^{(1, 1)})^3}{1 - \qty\big(2 T_{p, m}^{(1, 1)})^2},
    \end{aligned}
    \label{eq:diagrammaticBounds-evenLaceZero}\\
    &\begin{aligned}[b]
      \MoveEqLeft[1] \flt\Pi_p^\mathrm{odd}(0, m) \leq
      \abs{\flt\pi_p^{(0)}(0, m) - \flt\pi_p^{(1)}(0, m)}\\
      & + \qty(B_{p, m}^{(1, 1)} + \frac{1}{2} B_{p, 1}^{(2, 2)} B_{p, m}^{(1, 3)} + T_{p, m}^{(1, 1)} B_{p, m}^{(1, 1)}) \frac{\qty\big(2 T_{p, m}^{(1, 1)})^2}{1 - \qty\big(2 T_{p, m}^{(1, 1)})^2},
    \end{aligned}
    \label{eq:diagrammaticBounds-oddLaceZero}\\
    &\begin{aligned}[b]
      \MoveEqLeft \sum_{(x, t)} \qty(\Pi_p^\mathrm{even}(x, t) + \Pi_p^\mathrm{odd}(x, t)) m^t t\\
      \leq {} & \sum_{(x, t)} \qty(\abs{\pi_p^{(0)}(x, t) - \pi_p^{(1)}(x, t)} + \pi_p^{(2)}(x, t)) m^t t\\
      & + \qty(T_{p, m}^{(1, 1)} + \frac{1}{2} (2 B_{p, m}^{(2, 2)} B_{p, m}^{(1, 3)} + B_{p, m}^{(2, 2)} T_{p, m}^{(3, 1)}) + 2 \qty(T_{p, m}^{(1, 1)})^2) \frac{\qty\big(2 T_{p, m}^{(1, 1)})^2}{1 - 2 T_{p, m}^{(1, 1)}}\\
      & + \qty(T_{p, m}^{(1, 1)} + \frac{1}{2} B_{p, m}^{(2, 2)} T_{p, m}^{(1, 3)} + \qty(T_{p, m}^{(1, 1)})^2) \frac{2 T_{p, m}^{(1, 1)}}{\qty\big(1 - 2 T_{p, m}^{(1, 1)})^2}\\
      & + \qty(T_{p, m}^{(1, 1)} + \frac{1}{2} B_{p, m}^{(2, 2)} T_{p, m}^{(1, 3)} + \qty(T_{p, m}^{(1, 1)})^2) \frac{\qty\big(2 T_{p, m}^{(1, 1)})^2}{1 - 2 T_{p, m}^{(1, 1)}},
    \end{aligned}
    \label{eq:diagrammaticBounds-laceEmbedded}\\
    &\begin{aligned}[b]
      \MoveEqLeft \sum_{(x, t)} \qty(\Pi_p^\mathrm{even}(x, t) + \Pi_p^\mathrm{odd}(x, t)) m^t \qty(1 - \cos k\cdot x)\\
      \leq {} & \sum_{(x, t)} \abs{\pi_p^{(0)}(x, t) - \pi_p^{(1)}(x, t)} m^t \qty(1 - \cos k\cdot x)
      + \qty(\flt\pi_p^{(2)}(0, m) - \flt\pi_p^{(2)}(k, m))\\
      & + \qty(\flt V_{p, m}^{(1, 1)}(k) + \frac{1}{2} B_{p, m}^{(2, 2)} \flt V_{p, m}^{(1, 3)}(k) + 2 T_{p, m}^{(1, 1)} \flt V_{p, m}^{(1, 1)}(k)) \frac{8 \qty\big(T_{p, m}^{(1, 1)})^2 \qty\big(2 - 3 T_{p, m}^{(1, 1)})}{\qty\big(1 - 2 T_{p, m}^{(1, 1)})^2}\\
      & + 2 \qty(T_{p, m}^{(1, 1)} + \frac{1}{2} B_{p, m}^{(2, 2)} T_{p, m}^{(1, 3)} + \qty(T_{p, m}^{(1, 1)})^2) \flt V_{p, m}^{(1, 1)}(k) \frac{8 T_{p, m}^{(1, 1)} \qty\big(1 - T_{p, m}^{(1, 1)})}{\qty\big(1 - 2 T_{p, m}^{(1, 1)})^3}\\
      & + 2 \qty(T_{p, m}^{(1, 1)} + \frac{1}{2} B_{p, m}^{(2, 2)} T_{p, m}^{(1, 3)} + \qty(T_{p, m}^{(1, 1)})^2) \flt V_{p, m}^{(1, 1)}(k) \frac{4 T_{p, m}^{(1, 1)} \qty\big(2 - 3 T_{p, m}^{(1, 1)})}{\qty\big(1 - 2 T_{p, m}^{(1, 1)})^2}.
    \end{aligned}
    \label{eq:diagrammaticBounds-LaceDiff}
  \end{align}
  Here, upper bounds on $\abs*{\flt\pi_p^{(0)}(0, m) - \flt\pi_p^{(1)}(0, m)}$, $\flt\pi_p^{(2)}(0, m)$, etc. are given by Lemma~\ref{lem:ksp-bound}.
\end{lemma}

The proof of Lemma~\ref{lem:ksp-bound} is based on the diagrammatic representations of $\{\pi_p^{(N)}\}_{N=0}^{\infty}$ and the methods in \cite{Handa-Kamijima-Sakai,hs05,s07op}.
The farmer is to bound $\{\pi_p^{(N)}\}_{N=0}^{\infty}$ above in terms of sums of products of two-point functions.
The latter is to split events of percolation connections into some events depending on whether or not each line collapses.
Lemma~\ref{lem:ksp-bound} is a well-known result in this sense, hence we explain its proof in Appendix~\ref{sec:proof-diagrammaticBounds}.
However, these upper bounds are a little sharper than previous research in a certain sense because we take care of the explicit form of not only the leading terms but also the remainder terms.
They are often expressed as the order $\order{d^{-1}}$ relying on $d\gg d_{\mathrm{c}}$ (or the order $\order{L^{-d}}$ relying on $L\gg 1$ in the case of the spread-out model).
We believe that the upper bounds are worth showing in this paper.

In particular, each of the first term in \eqref{eq:sum-diagrammaticBounds-m-01}, \eqref{eq:sum-diagrammaticBounds-mt-01} and \eqref{eq:sum-diagrammaticBounds-mcos-01} is a leading term.
The coefficient $1/2$ results from \cite[Equation (3.11)]{hs05}.
The sum $\lambda + \rho$ of the superscript of $B_{p, m}^{(\lambda, \rho)}$ (similarly, $T_{p, m}^{(\lambda, \rho)}$ and $\flt V_{p, m}^{(\lambda, \rho)}(k)$) means how many $q_p$'s the diagram contains at least.
For example, the leading term $B_{p, m}^{(2, 2)} / 2$ in \eqref{eq:sum-diagrammaticBounds-m-01} contains at least four $q_p$'s.
If one isolates the diagrams containing six $q_p$'s from error terms and specifies their coefficients, then one can obtain more precise upper bounds than Lemma~\ref{lem:ksp-bound}.
However, too many $q_p$'s aggravate upper bounds on the basic diagrams due to many multiplicative constants $K_1$'s (See Lemma~\ref{lem:basic-diagrams}).
That is a reason why we isolate the diagrams containing only four $q_p$'s in this paper.

\section{Diagrammatic bounds on the bootstrap functions}\label{sec:improvedBounds}

In this section, using the lace expansion, we show upper bounds on $\{g_i\}_{i=1}^{3}$
in terms of the lace expansion coefficients.
As a consequence, we obtain the next lemma.
A sufficient condition \eqref{eq:smallCondition-laceCoefficients} below is to be verified after we complete the bootstrap argument.
Recall $\Pi_p(\vb*{x}) = \Pi_p^\mathrm{even}(\vb*{x}) - \Pi_p^\mathrm{odd}(\vb*{x})$.

\begin{lemma}\label{lem:improvedBounds}
  Suppose that $2 T_{p, m}^{(1, 1)} < 1$
  and that $B_{p, m}^{(\lambda, \rho)}$, $T_{p, m}^{(\lambda, \rho)}$, $\flt V_{p, m}^{(\lambda, \rho)}(k)$ for any $\rho, \lambda\in\mathbb{N}$ are so small that the inequality
  \begin{equation}
    \label{eq:smallCondition-laceCoefficients}
    \sum_{n=0}^{\infty} \flt\pi_p(0, m) < 1
  \end{equation}
  holds.
  Then, we have
  \begin{gather}
    g_1(p, m) \leq \frac{1}{1 - \flt\Pi_p^\mathrm{odd}(0, m)},
    \label{eq:improvedBound-g1}\\
    \begin{aligned}[b]
      \MoveEqLeft g_2(p, m)\\
      \leq {} & \frac{1 + \flt\Pi_p^\mathrm{even}(0, m) + \flt\Pi_p^\mathrm{odd}(0, m)}{1 - \flt\Pi_p^\mathrm{odd}(0, m)}
        + \frac{2 K_2 \flt\Pi_p^\mathrm{even}(0, m)}{1 - \flt\Pi_p^\mathrm{odd}(0, m)}\\
        & + \frac{2 K_1 K_2}{1 - \flt\Pi_p^\mathrm{odd}(0, m)}\\
        & \times \qty\bigg(
          \pi \sum_{(x, t)} \qty\Big(
            \Pi_p^\mathrm{even}(x, t) + \Pi_p^\mathrm{odd}(x, t)
          )
          m^t t
        )\\
        & \quad \vee \qty\bigg(
          \sum_{(x, t)} \qty\Big(
            \Pi_p^\mathrm{even}(x, t) + \Pi_p^\mathrm{odd}(x, t)
          )
          m^t \frac{1 - \cos k\cdot x}{1 - \ft D(k)}
        ),
    \end{aligned}
    \label{eq:improvedBound-g2}\\
    \begin{aligned}[b]
      \MoveEqLeft g_3(p, m)\\
      \leq {} & \Biggl( 1 \vee \frac{g_2(p, m)}{1 - \flt\Pi_p^\mathrm{even}(0, m) - \flt\Pi_p^\mathrm{odd}(0, m)}\Biggr)^3 K_1^2\\
      &\quad \times \Biggl(
        1 + 2 \qty(\flt\Pi_p^\mathrm{even}(0, m) + \flt\Pi_p^\mathrm{odd}(0, m))\\
        & + 2 \norm{
          \frac{
            \qty\big(\flt\Pi_p^\mathrm{even}(0, m) + \flt\Pi_p^\mathrm{odd}(0, m))
            - \qty\big(\flt\Pi_p^\mathrm{even}(\bullet, m) + \flt\Pi_p^\mathrm{odd}(\bullet, m))
          }{
            1 - \ft D(\bullet)
          }
        }_\infty
      \Biggr)^2.
    \end{aligned}
    \label{eq:improvedBound-g3}
  \end{gather}
\end{lemma}

We show the proofs of \eqref{eq:improvedBound-g1}--\eqref{eq:improvedBound-g3} in Section~\ref{sec:improvedBound-g1}--\ref{sec:improvedBound-g3}, respectively, with required lemmas.
Their lemmas are proved in Section~\ref{sec:proof-lemmas-improvedBounds} in a lump sum.

Assume the absolute convergence of \eqref{eq:def-sumOfLaceExpansionCoefficients} (in other words, \eqref{eq:smallCondition-laceCoefficients}).
Then, taking the Fourier-Laplace transform on both sides in \eqref{eq:op-laceExpansion} implies that
\begin{equation}
  \label{eq:FourierLaplace-recursionEquation}
  \flt\varphi_p(k, z) = \frac{1 + \flt\Pi_p(k, z)}{1 - \flt q_p(k, z) \qty\big(1 + \flt\Pi_p(k, z))}.
\end{equation}
We use this expression many times.

\subsection{An upper bound for $g_1(p, m)$}\label{sec:improvedBound-g1}

\begin{proof}[Proof of \eqref{eq:improvedBound-g1} in Lemma~\ref{lem:improvedBounds}]
  Note that
  $\flt\varphi_p(0, m) \geq 1$
  since $m\geq 1$.
  This inequality and \eqref{eq:FourierLaplace-recursionEquation} yield
  \begin{equation*}
    \flt\varphi_p(0, m) = \frac{1 + \flt\Pi_p(0, m)}{1 - p m \qty\big(1 + \flt\Pi_p(0, m))}
    \iff pm = \frac{1}{1 + \flt\Pi_p(0, m)} - \flt\varphi_p(0, m)^{-1}.
  \end{equation*}
  Thus, we obtain
  \begin{equation}
    g_1(p, m) \leq \frac{1}{1 + \flt\Pi_p^\mathrm{even}(0, m) - \flt\Pi_p^\mathrm{odd}(0, m)} \leq \frac{1}{1 - \flt\Pi_p^\mathrm{odd}(0, m)}.
  \end{equation}
\end{proof}

\subsection{An upper bound for $g_2(p, m)$}\label{sec:improvedBound-g2}

As preliminaries, we show the statements of one proposition and two lemmas in the following.
The next proposition plays the most important role.
Notice that $\set{k\in\mathbb{T}^d | \ft D(k) \geq 0} = \interval{-\pi / 2}{\pi / 2}^d$ on $\opSpace$.

\begin{proposition}\label{prp:restricting-g2}
  On the BCC lattice $\opSpace$, it suffices to attend to only the range $\interval{-\pi / 2}{\pi / 2}^d$ in order to calculate the supremum in $g_2(p, m)$ with respect to $k$:
  \[
    g_2(p, m)
    = \sup_{\substack{k\in\mathbb{T}^d,\\ z\in\mathbb{C}\colon\abs{z}\in\set{1, m}}}\frac{\abs{\flt\varphi_p(k, z)}}{\abs*{\ft S_{\mu_p(z)}(k)}}
    = \sup_{\substack{k\in\interval{-\pi / 2}{\pi / 2}^d,\\ z\in\mathbb{C}\colon\abs{z}\in\set{1, m}}}\frac{\abs{\flt\varphi_p(k, z)}}{\abs*{\ft S_{\mu_p(z)}(k)}}.
  \]
\end{proposition}

\begin{lemma}  \label{lem:upperBound-Green}
  For every $k\in\mathbb{T}^d$ satisfying $\ft D(k) \geq 0$, $r\in\interval{0}{1}$ and $\theta\in\mathbb{T}$,
  \begin{equation}
    \label{eq:upperBound-Green}
    \abs{1 - r \Napier^{\imag\theta} \ft D(k)} \geq \frac{1}{2} \qty(\frac{\abs{\theta}}{\pi} + 1 - \ft D(k)).
  \end{equation}
\end{lemma}

\begin{remark}
  The proof of Proposition~\ref{prp:restricting-g2} depends on the structure of the BCC lattice $\opSpace$.
  It suffices for $\opSpace$ to reduce the range of $k$ to $\interval{-\pi / 2}{\pi / 2}^d$
  due to reflections in perpendicular hyperplanes to the axes,
  whereas we need extra efforts to prove the same claim for $\mathbb{Z}^d$ due to lack of symmetry.
  The shape of the domain of $k$ such that $\ft D(k)\geq 0$ for $\mathbb{Z}^d$, that is
  $\set{k \in \mathbb{T}^d | d^{-1}\sum_{j=1}^{d} \cos k_j \geq 0}$
  is not as simple as $\set{k \in \mathbb{T}^d | \prod_{j=1}^{d} \cos k_j \geq 0}$ for $\opSpace$.
  However, focusing on $\theta=\arg z$ let us reduce the range over which the supremum is taken on the simple cubic lattice $\mathbb{Z}^d$.
  Specifically, it suffices to attend to only the range $\interval{-\pi/2}{\pi/2}$ in order to calculate the supremum in $g_2(p, m)$ with respect to $\theta$ on $\mathbb{Z}^d$, and a similar inequality to \eqref{eq:upperBound-Green} holds.
  Based on this idea, one can prove the infrared bound \eqref{eq:infraredBound} on $\mathbb{Z}^{d\geq 183}\times\opTime$.
  See \cite{k22} for details.
\end{remark}

The next lemma is almost identical to the above lemma, but note that the ranges of $k$ are different from each other.

\begin{lemma}\label{lem:mu-bound}
  For every $k\in\mathbb{T}^d$ and $\mu\in\mathbb{C}$ satisfying $\abs{\mu}\leq 1$,
  \begin{equation}
    \abs{1 - \Napier^{\imag\arg \mu} \ft D(k)} \leq 2 \abs{1 - \mu \ft D(k)}.
  \end{equation}
\end{lemma}

\begin{proof}[Proof of \eqref{eq:improvedBound-g2} in Lemma~\ref{lem:improvedBounds}]
  Applying the same method as \cite[Lemma~8.11]{hh17}, we rewrite $\flt\varphi_p$ as
  \begin{multline*}
    \frac{\flt\varphi_p(k, z)}{\ft S_{\mu_p(z)}(k)}
    = \frac{1 + \flt\Pi_p(k, z)}{1 + \flt\Pi_p(0, \abs{z})}
    + \frac{\flt\varphi_p(k, z)}{1 + \flt\Pi_p(0, \abs{z})}\\
    \times \qty(
      \flt\Pi_p(0, \abs{z}) \qty\big(1 - \Napier^{\imag\theta} \ft D(k))
      + \flt q_p(k, z) \qty\big(\flt\Pi_p(k, z) - \flt\Pi_p(0, \abs{z}))
    ),
  \end{multline*}
  where $\theta = \arg z$.
  By Proposition~\ref{prp:restricting-g2} and the triangle inequality,
  \begin{align}
    \MoveEqLeft[1] g_2(p, m) \notag\\
    \leq {} & \sup_{\substack{k\in\interval{-\frac{\pi}{2}}{\frac{\pi}{2}}^d,\\ z\in\mathbb{C}\colon\abs{z}\in\set{1, m}}}
      \Biggl(
        \frac{1 + \abs*{\flt\Pi_p(k, z)}}{1 + \flt\Pi_p(0, \abs{z})}
        + \frac{\abs*{\varphi_p(k, z)}}{1 + \flt\Pi_p(0, \abs{z})}
        \biggl(
          \flt\Pi_p(0, \abs{z}) \abs{1 - \Napier^{\imag\theta} \ft D(k)} \notag\\
          & + p \abs{z} \underbrace{\abs\big{\ft D(k)}}_{\leq 1}
          \qty\Big(
            \abs\big{\flt\Pi_p(0, \abs{z}) - \flt\Pi_p(0, z)}
            + \abs\big{\flt\Pi_p(0, z) - \flt\Pi_p(k, z)}
          )
        \biggr)
      \Biggr). \label{eq:crucialStep-g2-bound}
  \end{align}
  By the bootstrap hypotheses $g_i(p, m) \leq K_i$ for $i=1, 2$,
  \begin{align}
    \MoveEqLeft[1] g_2(p, m) \notag\\
    \leq {} & \sup_{\substack{k\in\interval{-\frac{\pi}{2}}{\frac{\pi}{2}}^d,\\ z\in\mathbb{C}\colon\abs{z}\in\set{1, m}}}
      \Biggl(
        \frac{1 + \abs*{\flt\Pi_p(k, z)}}{1 + \flt\Pi_p(0, \abs{z})}
        + \frac{K_2 \ft S_{\mu_p(z)}(k)}{1 + \flt\Pi_p(0, \abs{z})}
        \biggl(
          \flt\Pi_p(0, \abs{z}) \abs{1 - \Napier^{\imag\theta} \ft D(k)} \notag\\
          & + K_1
          \qty\Big(
            \abs\big{\flt\Pi_p(0, \abs{z}) - \flt\Pi_p(0, z)}
            + \abs\big{\flt\Pi_p(0, z) - \flt\Pi_p(k, z)}
          )
        \biggr)
      \Biggr). \label{eq:bootstrapStep-g2-bound}
  \end{align}
  Recall $\Pi_p^\mathrm{even}(\vb*{x}) = \sum_{N=1}^{\infty}\pi_p^{(2N)}(\vb*{x})$
  and $\Pi_p^\mathrm{odd}(\vb*{x}) = \pi_p^{(1)}(\vb*{x}) - \pi_p^{(0)}(\vb*{x}) + \sum_{N=1}^{\infty}\pi_p^{(2N+1)}(\vb*{x})$.
  Applying the triangle inequality to $\abs*{\flt\Pi_p(k, z)}$, $\abs*\big{\flt\Pi_p(0, \abs{z}) - \flt\Pi_p(0, z)}$ and $\abs*\big{\flt\Pi_p(0, z) - \flt\Pi_p(k, z)}$,
  we obtain, respectively,
  \begin{align*}
    \MoveEqLeft[1] \abs{\flt\Pi_p(k, z)}
    \leq \abs{\flt\Pi_p^\mathrm{even}(k, z)} + \abs{\flt\Pi_p^\mathrm{odd}(k, z)}\\
    = {} & \abs{\sum_{(x, t)}\sum_{N=1} \pi_p^{(2N)}(\vb*{x}) \Napier^{\imag k\cdot x} z^t}\\
      &+ \abs{\sum_{(x, t)}\sum_{N=1}
        \Bigl(
          \pi_p^{(1)}(\vb*{x}) - \pi_p^{(0)}(\vb*{x})
          + \pi_p^{(2N+1)}(\vb*{x})
        \Bigr)
        \Napier^{\imag k\cdot x} z^t}\\
    \leq {} & \sum_{(x, t)}\sum_{N=1} \pi_p^{(2N)}(\vb*{x}) \abs{z}^t\\
      &+ \sum_{(x, t)}\sum_{N=1}
        \Bigl(
          \pi_p^{(1)}(\vb*{x}) - \pi_p^{(0)}(\vb*{x})
          + \pi_p^{(2N+1)}(\vb*{x})
        \Bigr)
        \abs{z}^t\\
    = {} & \flt\Pi_p^\mathrm{even}(0, \abs{z}) + \flt\Pi_p^\mathrm{odd}(0, \abs{z}),
  \end{align*}
  \begin{align*}
    &\abs\big{\flt\Pi_p(0, \abs{z}) - \flt\Pi_p(0, z)}\\
    &\leq \sum_{(x, t)}
      \Bigl(
        \pi_p^{(1)}(x, t) - \pi_p^{(0)}(x, t)
        + \sum_{N=2}^{\infty}\pi_p^{(N)}(x, t)
      \Bigr)
      \abs{z}^t \abs*{1 - \Napier^{\imag\theta t}}\\
    &= \sum_{(x, t)} \Bigl(\Pi_p^\mathrm{even}(x, t) + \Pi_p^\mathrm{odd}(x, t)\Bigr) \abs{z}^t \abs*{1 - \Napier^{\imag\theta t}}
  \end{align*}
  and (note that each $\pi_p^{(N)}(x, t)$ is an even function with respect to $x$)
  \begin{align*}
    &\abs\big{\flt\Pi_p(0, z) - \flt\Pi_p(k, z)}\\
    &\leq \sum_{(x, t)}
    \Bigl(
      \pi_p^{(1)}(x, t) - \pi_p^{(0)}(x, t)
      + \sum_{N=2}^{\infty}\pi_p^{(N)}(x, t)
    \Bigr)
    \abs{z}^t (1 - \cos k\cdot x)\\
    &= \sum_{(x, t)} \Bigl(\Pi_p^\mathrm{even}(x, t) + \Pi_p^\mathrm{odd}(x, t)\Bigr) \abs{z}^t (1 - \cos k\cdot x).
  \end{align*}
  Substituting these bounds and \eqref{eq:improvedBound-g1} into \eqref{eq:bootstrapStep-g2-bound} leads to
  \begin{align*}
    \MoveEqLeft[1] g_2(p, m)\\
    \leq {} & \frac{1 + \flt\Pi_p^\mathrm{even}(0, m) + \flt\Pi_p^\mathrm{odd}(0, m)}{1 - \flt\Pi_p^\mathrm{odd}(0, m)}\\
      & + \frac{K_2}{1 - \flt\Pi_p^\mathrm{odd}(0, m)}
        \sup_{\substack{k\in\interval{-\frac{\pi}{2}}{\frac{\pi}{2}}^d,\\ z\in\mathbb{C}\colon\abs{z}\in\set{1, m}}}
        \biggl(
          \flt\Pi_p(0, \abs{z}) \frac{\abs*{1 - \Napier^{\imag\theta} \ft D(k)}}{\abs*{1 - \mu_p(z) \ft D(k)}}\\
          & + K_1 \sum_{(x, t)} \Bigl(\Pi_p^\mathrm{even}(x, t) + \Pi_p^\mathrm{odd}(x, t)\Bigr)
          \abs{z}^t \cdot \frac{\abs*{1 - \Napier^{\imag\theta t}} + 1 - \cos k\cdot x}{\abs*{1 - \mu_p(z) \ft D(k)}}
        \biggr).
  \end{align*}
  By Lemma~\ref{lem:upperBound-Green}--\ref{lem:mu-bound} and the inequality $\abs{1 - \Napier^{\imag\theta t}} = 2 \abs{\sin(\theta t / 2)} \leq t \abs{\theta}$,
  \begin{align*}
    \MoveEqLeft[1] g_2(p, m)\\
    \leq {} & \frac{1 + \flt\Pi_p^\mathrm{even}(0, m) + \flt\Pi_p^\mathrm{odd}(0, m)}{1 - \flt\Pi_p^\mathrm{odd}(0, m)}\\
      & + \frac{K_2}{1 - \flt\Pi_p^\mathrm{odd}(0, m)}
        \sup_{\substack{k\in\interval{-\frac{\pi}{2}}{\frac{\pi}{2}}^d,\\ z\in\mathbb{C}\colon\abs{z}\in\set{1, m}}}
        \biggl(
          2 \underbrace{\flt\Pi_p(0, \abs{z})}_{\mathrlap{\leq \flt\Pi_p^\mathrm{even}(0, \abs{z})}}\\
          & + 2 K_1 \sum_{(x, t)} \Bigl(\Pi_p^\mathrm{even}(x, t) + \Pi_p^\mathrm{odd}(x, t)\Bigr)
          \abs{z}^t\cdot
          \frac{t \abs{\theta} + 1 - \cos k\cdot x}{\pi^{-1} \abs{\theta} + 1 - \ft D(k)}
        \biggr)\\
    \leq {} & \frac{1 + \flt\Pi_p^\mathrm{even}(0, m) + \flt\Pi_p^\mathrm{odd}(0, m)}{1 - \flt\Pi_p^\mathrm{odd}(0, m)}
      + \frac{2 K_2 \flt\Pi_p^\mathrm{even}(0, m)}{1 - \flt\Pi_p^\mathrm{odd}(0, m)}\\
      & + \frac{2 K_1 K_2}{1 - \flt\Pi_p^\mathrm{odd}(0, m)}\\
      & \times
        \sup_{\substack{k\in\interval{-\frac{\pi}{2}}{\frac{\pi}{2}}^d,\\ z\in\mathbb{C}\colon\abs{z}\in\set{1, m}}}
        \sum_{(x, t)} \Bigl(\Pi_p^\mathrm{even}(x, t) + \Pi_p^\mathrm{odd}(x, t)\Bigr)
        \abs{z}^t\cdot
        \frac{t \abs{\theta} + 1 - \cos k\cdot x}{\pi^{-1} \abs{\theta} + 1 - \ft D(k)}\\
    = {} & \frac{1 + \flt\Pi_p^\mathrm{even}(0, m) + \flt\Pi_p^\mathrm{odd}(0, m)}{1 - \flt\Pi_p^\mathrm{odd}(0, m)}
      + \frac{2 K_2 \flt\Pi_p^\mathrm{even}(0, m)}{1 - \flt\Pi_p^\mathrm{odd}(0, m)}\\
      & + \frac{2 K_1 K_2}{1 - \flt\Pi_p^\mathrm{odd}(0, m)} \cdot \frac{1}{\pi^{-1} \abs{\theta} + 1 - \ft D(k)}\\
      & \times
        \sup_{\substack{k\in\interval{-\frac{\pi}{2}}{\frac{\pi}{2}}^d,\\ z\in\mathbb{C}\colon\abs{z}\in\set{1, m}}}
        \biggl(
          \underbrace{\pi \sum_{(x, t)} \Bigl(\Pi_p^\mathrm{even}(x, t) + \Pi_p^\mathrm{odd}(x, t)\Bigr) \abs{z}^t t}_{\mathrm{(a)}}
            \cdot \underbrace{\pi^{-1} \abs{\theta}}_{\mathrm{(b)}}\\
          & + \underbrace{\sum_{(x, t)} \Bigl(\Pi_p^\mathrm{even}(x, t) + \Pi_p^\mathrm{odd}(x, t)\Bigr) \abs{z}^t \frac{1 - \cos k\cdot x}{1 - \ft D(k)}}_{\mathrm{(c)}}
            \cdot \underbrace{\bigl(1 - \ft D(k)\bigr)}_{\mathrm{(d)}}
        \biggr).
  \end{align*}
  Note that the inequality $(a_1 \phi + b_1 \psi) / (a_2 \phi + b_2 \psi) \leq (a_1 \vee b_1) / (a_2 \wedge b_2)$ holds for any $a_1,\allowbreak b_1, a_2, b_2, \phi, \psi \in \mathbb{R}$.
  Setting $a_1 = \mathrm{(a)}$, $\phi = \mathrm{(b)}$, $b_1 = \mathrm{(c)}$, $\psi = \mathrm{(d)}$ and $a_2 = b_2 = 1$,
  we arrive at
  \begin{align}
    \MoveEqLeft[1] g_2(p, m) \notag\\
    \leq {} & \frac{1 + \flt\Pi_p^\mathrm{even}(0, m) + \flt\Pi_p^\mathrm{odd}(0, m)}{1 - \flt\Pi_p^\mathrm{odd}(0, m)}
      + \frac{2 K_2 \flt\Pi_p^\mathrm{even}(0, m)}{1 - \flt\Pi_p^\mathrm{odd}(0, m)} \notag\\
      & + \frac{2 K_1 K_2}{1 - \flt\Pi_p^\mathrm{odd}(0, m)} \notag\\
      & \times \qty\bigg(
        \pi \sum_{(x, t)} \qty\Big(
          \Pi_p^\mathrm{even}(x, t) + \Pi_p^\mathrm{odd}(x, t)
        )
        m^t t
      ) \notag\\
      & \quad \vee \qty\bigg(
        \sum_{(x, t)} \qty\Big(
          \Pi_p^\mathrm{even}(x, t) + \Pi_p^\mathrm{odd}(x, t)
        )
        m^t \frac{1 - \cos k\cdot x}{1 - \ft D(k)}
      ).
  \end{align}
\end{proof}

\begin{remark}\label{rem:role-restricting-g2}
  Notice that, if we do not use Proposition~\ref{prp:restricting-g2}
  and naively apply the triangle inequality as in \eqref{eq:crucialStep-g2-bound},
  \begin{equation}
    \label{eq:op-naiveBound}
    \frac{\abs\big{\flt\Pi_p(0, \abs{z}) - \flt\Pi_p(k, z)}}{\abs\big{1 - \Napier^{\imag\theta} \ft D(k)}}
    \leq \frac{\abs\big{\flt\Pi_p(0, \abs{z}) - \flt\Pi_p(0, z)} + \abs\big{\flt\Pi_p(0, z) - \flt\Pi_p(k, z)}}{\abs\big{1 - \Napier^{\imag\theta} \ft D(k)}}
  \end{equation}
  appears in an upper bound on $g_2(p, m)$, but this bound becomes infinite for some situations.
  When $\ft D(k) = -1$ and $\theta=\pm\pi$, its denominator equals $0$ even though its numerator is non-zero.
  Thanks to Proposition~\ref{prp:restricting-g2}, we can avoid this problem.
  Although one may consider keeping the left hand side in \eqref{eq:op-naiveBound} without using the triangle inequality,
  the way does not work currently
  because we do not have any nice methods to decompose diagrams with complex numbers.
\end{remark}

\begin{remark}
  Nguyen and Yang wrote ``\ldots $1 - \Napier^{\imag t}\ft D(k)$ is bounded below by a positive constant
  for $\set{(k, t) \in \interval{-\pi}{\pi}^d\times\interval{-\pi}{\pi} : \abs{(k, t)} > \varepsilon}$ \ldots''
  below \cite[Equation (13)]{ny93}.
  However, it is impossible for the nearest-neighbor model on the simple cubic lattice $\mathbb{Z}^d$ because it can occur that $\ft D(k)=-1$ when $k_i=\pm\pi$ for some $i=1, \dots, d$.
  That is a reason why we need Proposition~\ref{prp:restricting-g2}.
  The spread-out model does not cause such problem due to $1 + \ft D_L(k) > 0$ with $L\gg 1$ in \cite[Proof of Lemma~10]{ny93},
  where $D_L$ is the random-walk transition probability for the spread-out model.
  This inequality implies that $\abs*{1 - \Napier^{\imag\theta}\ft D(k)} > \mathrm{const.} (\norm{k}_2 L^2 + \theta^2)$,
  hence the naive bound \eqref{eq:op-naiveBound} works for sufficiently large $L$.
  If we consider oriented percolation whose staying probability is positive (considered in, e.g., \cite{gh02}),
  then the naive bound also works for such ``nearest-neighbor'' model
  because, in this case, $D$ turns to the lazy random-walk transition probability $D_\mathrm{lazy}$.
  This transition probability is defined as, for $\lambda\in\interval{0}{1}$ and $x\in\opSpace$,
  \[
    D_\mathrm{lazy}(x) = \lambda \delta_{x, o} + \left(1 - \lambda\right) \frac{\ind{x\in\mathscr{N}^d}}{\card{\mathscr{N}^d}} = \lambda \delta_{x, o} + \left(1 - \lambda\right) D(x),
  \]
  of which the Fourier transform is given by $\ft D_\mathrm{lazy}(k) = \lambda + (1 - \lambda) \ft D(k)$.
  When we choose an arbitrary $\lambda < 1/2$, the inequality $1 + \ft D_\mathrm{lazy}(k) > 0$ holds for every $k\in\mathbb{T}^d$ even at $k$ such that $\ft D(k) = -1$.
  Thus, the lazy model justifies the bound \eqref{eq:crucialStep-g2-bound}.
  However, it is different from so-called nearest-neighbor oriented percolation, so that we did not deal with the lazy model in this paper.
\end{remark}

\subsection{An upper bound for $g_3(p, m)$}\label{sec:improvedBound-g3}

An upper bound on $g_3(p, m)$ require us to modify \cite[Lemma~5.7]{s06}.

\begin{lemma}[A modified version of {\cite[(5.29)--(5.33)]{bchss05}} or {\cite[Lemma~5.7]{s06}}]\label{lem:trigonometricInequality}
  Suppose that $a(-x, t) = a(x, t)$ for all $(x, t)\in\opSpaceTime$, and let
  \[
    \flt A(k, z) = \frac{1}{1 - \flt a(k, z)},
  \]
  Then, for all $k, l\in\mathbb{T}^d$,
  \begin{multline}
    \abs{\frac{1}{2}\DLaplacian{k}\flt A(l, z)}\\
    \leq \Bigl(\FLT{\abs{a}}(0, \abs{z}) - \FLT{\abs{a}}(k, \abs{z})\Bigr)
      \Biggl(
        \frac{\abs*{\flt A(l + k, z)} + \abs*{\flt A(l - k, z)}}{2} \abs{\flt A(l, z)}\\
        \qquad + \abs{\flt A(l, z)} \abs{\flt A(l + k, z)} \abs{\flt A(l - k, z)} \Bigl(\FLT{\abs{a}}(0, \abs{z}) - \FLT{\abs{a}}(2l, \abs{z})\Bigr)
      \Biggr),
    \label{eq:modified-trigonometricInequality}
  \end{multline}
  where $\FLT{\abs*{a}}(k, z) = \sum_{(x, t)} \abs*{a(x, t)} \Napier^{\imag k\cdot x} z^t$.
\end{lemma}

\begin{proof}[Proof of \eqref{eq:improvedBound-g3} in Lemma~\ref{lem:improvedBounds}]
  Recall \eqref{eq:op-bootstrapFunction3} and \eqref{eq:FourierLaplace-recursionEquation}.
  Setting $a(x, t) = q_p(x, t) + (q_p\star\Pi_p)(x, t)$ in Lemma~\ref{lem:trigonometricInequality}, we have
  \begin{equation}
    \abs{\frac{1}{2}\DLaplacian{k}\bigl(\flt q_p(l, z) \flt\varphi_p(l, z)\bigr)}
    = \abs{\frac{1}{2}\DLaplacian{k}\bigg(\frac{1}{1 - \flt a(l, z)} - 1\bigg)}
    =\abs{\frac{1}{2}\DLaplacian{k}\flt A(l, z)}.
    \label{eq:2nd-derivative-line}
  \end{equation}
  By Lemma~\ref{lem:split-of-cosine} and the bootstrap hypothesis $g_1(p, m) \leq K_1$,
  \begin{align}
    &\FLT{\abs{a}}(0, \abs{z}) - \FLT{\abs{a}}(k, \abs{z})
    = \sum_{(x, t)} \abs{q_p(x, t) + (q_p\star\Pi_p)(x, t)} \abs{z}^t (1 - \cos k\cdot x) \notag\\
    &\leq \sum_{(x, t)} q_p(x, t) \abs{z}^t (1 - \cos k\cdot x) \notag\\
      &\quad + \sum_{(x, t)} \sum_{(y, s)} q_p(y, s) \abs{\Pi_p(x - y, t - s)} \abs{z}^t
        \underbrace{\Bigl( 1 - \cos k\cdot \bigl(y + (x - y)\bigr)\Bigr)}_{\leq 2 (1 - \cos k\cdot y) + 2 (1 - \cos k\cdot (x - y))} \notag\\
    &\leq p\abs{z} \bigl( 1 - \ft D(k)\bigr) \left( 1 + 2 \FLT{\abs{\Pi_p}}(0, \abs{z}) + 2 \cdot \frac{\FLT{\abs{\Pi_p}}(0, \abs{z}) - \FLT{\abs{\Pi_p}}(k, \abs{z})}{1 - \ft D(k)} \right) \notag\\
    &\begin{multlined}[b]
      \leq K_1 \bigl( 1 - \ft D(k)\bigr) \Biggl(
        1 + 2 \flt\Pi_p^\mathrm{even}(0, \abs{z}) + 2 \flt\Pi_p^\mathrm{odd}(0, \abs{z})\\
        + 2 \cdot \frac{\flt\Pi_p^\mathrm{even}(0, \abs{z}) - \flt\Pi_p^\mathrm{even}(k, \abs{z})}{1 - \ft D(k)}
        + 2 \cdot \frac{\flt\Pi_p^\mathrm{odd}(0, \abs{z}) - \flt\Pi_p^\mathrm{odd}(k, \abs{z})}{1 - \ft D(k)}
      \Biggr).
    \end{multlined}
    \label{eq:upperBound-diff-perturbation}
  \end{align}
  By the triangle inequality and the definition of $g_2(p, m)$,
  \begin{align}
    \abs{\flt A(l, z)}
    = \frac{\abs*{\flt\varphi_p(k, z)}}{\abs*{1 + \flt\Pi_p(l, z)}}
    \leq \frac{g_2(p, m)}{1 - \flt\Pi_p^\mathrm{even}(0, m) - \flt\Pi_p^\mathrm{odd}(0, m)} \abs{\ft S_{\mu_p(z)}(l)}.
    \label{eq:upperBound-like-g2}
  \end{align}
  Recall the definition \eqref{eq:2nd-derivative-RW} of $\flt U_{\mu_p(z)}(k, l)$.
  Applying Lemma~\ref{lem:trigonometricInequality} to \eqref{eq:2nd-derivative-line}, and combining the bounds \eqref{eq:upperBound-diff-perturbation} and \eqref{eq:upperBound-like-g2}, we obtain
  \begin{multline}
    g_3(p, m)
    \leq \Biggl( 1 \vee \frac{g_2(p, m)}{1 - \flt\Pi_p^\mathrm{even}(0, m) - \flt\Pi_p^\mathrm{odd}(0, m)}\Biggr)^3 K_1^2\\
    \times \Biggl(
      1 + 2 \qty(\flt\Pi_p^\mathrm{even}(0, m) + \flt\Pi_p^\mathrm{odd}(0, m))\\
      + 2 \norm{
        \frac{
          \qty\big(\flt\Pi_p^\mathrm{even}(0, m) + \flt\Pi_p^\mathrm{odd}(0, m))
          - \qty\big(\flt\Pi_p^\mathrm{even}(\bullet, m) + \flt\Pi_p^\mathrm{odd}(\bullet, m))
        }{
          1 - \ft D(\bullet)
        }
      }_\infty
    \Biggr)^2.
  \end{multline}
\end{proof}

\subsection{Proof of lemmas}\label{sec:proof-lemmas-improvedBounds}

\begin{proof}[Proof of Proposition~\ref{prp:restricting-g2}]
  We notice that
  \begin{align}
    &g_2(p, m) \notag\\
    &= \sup_{\substack{k\in\mathbb{T}^d,\\ r\in\set{1, m},\\ \theta\in\mathbb{T}}}
      \frac{\abs{\flt\varphi_p(k, r\Napier^{\imag\theta})}}{\abs*{\ft S_{\mu_p(r\Napier^{\imag\theta})}(k)}} \notag\\
    &= \qty\Bigg(
      \sup_{\substack{k \in \interval{-\pi/2}{\pi/2}^d,\\ r\in\set{1, m},\\ \theta\in\mathbb{T}}}
      \frac{\abs{\flt\varphi_p(k, r\Napier^{\imag\theta})}}{\abs*{\ft S_{\mu_p(r\Napier^{\imag\theta})}(k)}}
    )
    \vee \qty\Bigg(
      \sup_{\substack{k \in \mathbb{T}^d \setminus \interval{-\pi/2}{\pi/2}^d,\\ r\in\set{1, m},\\ \theta\in\mathbb{T}}}
      \frac{\abs{\flt\varphi_p(k, r\Napier^{\imag\theta})}}{\abs*{\ft S_{\mu_p(r\Napier^{\imag\theta})}(k)}}
    ).
    \label{eq:splitting-g2}
  \end{align}
  Given some $k = (k_1, \dots, k_d) \in \mathbb{T}^d \setminus \interval{-\pi/2}{\pi/2}^d$,
  there exists $\tilde k = (\tilde k_1, \dots, \tilde k_d) \in \interval{-\pi/2}{\pi/2}^d$ such that, for $j = 1, \dots, d$,
  \begin{equation*}
    \tilde k_j = \begin{cases}
      k_j & [-\pi/2 \leq k_j \leq \pi/2],\\
      \pi - k_j & [\pi/2 < k_j \leq \pi],\\
      -\pi - k_j & [-\pi \leq k_j < -\pi/2].
    \end{cases}
  \end{equation*}
  The second and third lines mean the reflection across points $\pi/2$ and $-\pi/2$ in one dimension, respectively.
  We show to be able to replace $k$ in the second factor in \eqref{eq:splitting-g2} by $\tilde k$.

  Fix $k \in \mathbb{T}^d \setminus \interval{-\pi/2}{\pi/2}^d$, $r\in\set{1, m}$ and $\theta \in \mathbb{T}$.
  Note that $\tilde k_j = k_j$ for $j=1, \dots, d$ when $-\pi/2 \leq k_j \leq \pi/2$.
  Since $\varphi_p(x, t)$ is an even function with respect to $x$,
  \begin{align*}
    \flt\varphi_p(k, r\Napier^{\imag\theta})
    &= \sum_{(x, t)} \varphi_p(x, t) r^t \Napier^{\imag \sum_{j=1}^{d} k_j x_j} \Napier^{\imag\theta t}\\
    &= \sum_{(x, t)} \varphi_p(x, t) r^t \Napier^{-\imag \tilde k \cdot x} \Napier^{\pi\imag \sum_{j=1}^{d} \ind{\tilde k_j\neq k_j} \sgn(k_j) x_j} \Napier^{\imag\theta t}\\
    &= \sum_{(x, t)} \varphi_p(-x, t) r^t \Napier^{\imag \tilde k \cdot x} (-1)^{-\sum_{j=1}^{d} \ind{\tilde k_j\neq k_j} \sgn(k_j) x_j} \Napier^{\imag\theta t}\\
    &= \sum_{(x, t)} \varphi_p(x, t) r^t \Napier^{\imag \tilde k \cdot x} (-1)^{\sum_{j=1}^{d} \ind{\tilde k_j\neq k_j} \sgn(k_j) x_j} \Napier^{\imag\theta t},
  \end{align*}
  where $\sgn\colon \mathbb{R} \to \set{-1, 0, +1}$ denotes the sign function.
  We divide into the two cases where $\sum_{j=1}^{d} \allowbreak \ind{\tilde k_j\neq k_j}$ is even or odd.

  On the one hand, suppose that $\sum_{j=1}^{d} \ind{\tilde k_j\neq k_j}$ is even.
  Since all coordinates of some $x\in\opSpace$ are either even or odd, $\sum_{j=1}^{d} \ind{\tilde k_j\neq k_j} \sgn(k_j) x_j$ is even.
  It is easy to see that
  \begin{equation*}
    \flt\varphi_p(k, r\Napier^{\imag\theta})
    = \sum_{(x, t)} \varphi_p(x, t) r^t \Napier^{\imag \tilde k \cdot x} \Napier^{\imag\theta t}
    = \flt\varphi_p(\tilde k, r\Napier^{\imag\theta}).
  \end{equation*}

  On the other hand, suppose that $\sum_{j=1}^{d} \ind{\tilde k_j\neq k_j}$ is odd.
  Let $\opSpace_\mathrm{e} = \{x\in\opSpace \mid \forall j,\ x_j\text{ is even}\}$,
  $\opSpace_\mathrm{o} = \set{x\in\opSpace | \forall j,\ x_j\text{ is odd}}$,
  $\mathbb{Z}_{\mathrm{e}+}$ be the set of non-negative even numbers
  and $\mathbb{Z}_{\mathrm{o}+}$ be the set of positive odd numbers.
  Note that $\varphi_p(x, t) = \mathbb{P}_p((o, 0) \rightarrow (x, t)) = 0$ when $(x, t) \neq (\opSpace_\mathrm{e}\times\mathbb{Z}_{\mathrm{e}+}) \cup (\opSpace_\mathrm{o}\times\mathbb{Z}_{\mathrm{o}+})$
  due to the definition of nearest-neighbor oriented percolation.
  We rewrite $\flt\varphi_p(k, r\Napier^{\imag\theta})$ as
  \begin{equation}
    \flt\varphi_p(k, r\Napier^{\imag\theta})
    = \sum_{\substack{x\in\opSpace_\mathrm{e},\\ t\in\mathbb{Z}_{\mathrm{e}+}}} \varphi_p(x, t) r^t \Napier^{\imag \tilde k \cdot x} \Napier^{\imag\theta t}
      - \sum_{\substack{x\in\opSpace_\mathrm{o},\\ t\in\mathbb{Z}_{\mathrm{o}+}}} \varphi_p(x, t) r^t \Napier^{\imag \tilde k \cdot x} \Napier^{\imag\theta t}.
    \label{eq:division-even-odd}
  \end{equation}
  Let $\tilde\theta = -\theta + \pi \sgn(\theta)$, which means the reflection of across points $\pi/2$ and $-\pi/2$ in one dimension.
  This definition is different from that of $\tilde k$ because even the point $\theta \in \interval{-\pi/2}{\pi/2}$ is reflected.
  By using $\tilde \theta$,
  \begin{equation}
    \Napier^{\imag\theta t}
    = \Napier^{-\imag\tilde\theta t} \Napier^{\pi\imag \sgn(\theta) t}
    = (-1)^t \Napier^{-\imag\tilde \theta t}.
    \label{eq:reflected-theta}
  \end{equation}
  Substituting \eqref{eq:reflected-theta} into \eqref{eq:division-even-odd}, we obtain
  \begin{align*}
    \flt\varphi_p(k, r\Napier^{\imag\theta})
    &= \sum_{\substack{x\in\opSpace_\mathrm{e},\\ t\in\mathbb{Z}_{\mathrm{e}+}}} \varphi_p(x, t) r^t \Napier^{\imag \tilde k \cdot x} \Napier^{-\imag\tilde\theta t}
    + \sum_{\substack{x\in\opSpace_\mathrm{o},\\ t\in\mathbb{Z}_{\mathrm{o}+}}} \varphi_p(x, t) r^t \Napier^{\imag \tilde k \cdot x} \Napier^{-\imag\tilde\theta t}\\
    &= \sum_{(x, t)} \varphi_p(x, t) r^t \Napier^{\imag \tilde k \cdot x} \Napier^{-\imag\tilde\theta t}
    = \flt\varphi_p(\tilde k, r\Napier^{-\imag\tilde\theta}).
  \end{align*}
  As for $\theta$, $-\tilde\theta$ takes a value in $\mathbb{T}$.
  It is possible to replace $-\tilde\theta$ in $\flt\varphi_p(\tilde k, r\Napier^{-\imag\tilde\theta})$ by $\theta$
  since it is the dummy variable of the supremum in \eqref{eq:splitting-g2}.

  The above discussions also hold for $\flt S_{\mu_p(r\Napier^{\imag\theta})}(k)$
  because $Q_1(x, t)$ satisfies the same properties ``$Q_1(x, t)$ is an even function with respect to $x$''
  and ``$Q_1(x, t) = D^{\conv t}(x) \ind{t\in\opTime} = 0$ when $(x, t) \neq (\opSpace_\mathrm{e}\times\mathbb{Z}_{\mathrm{e}+}) \cup (\opSpace_\mathrm{o}\times\mathbb{Z}_{\mathrm{o}+})$''
  as $\flt\varphi_p(k, r\Napier^{\imag\theta})$.
  Therefore,
  \begin{equation*}
    g_2(p, m)
    = \qty\Bigg(
      \sup_{\substack{k \in \interval{-\pi/2}{\pi/2}^d,\\ r\in\set{1, m},\\ \theta\in\mathbb{T}}}
      \frac{\abs{\flt\varphi_p(k, r\Napier^{\imag\theta})}}{\abs*{\ft S_{\mu_p(r\Napier^{\imag\theta})}(k)}}
    )
    \vee \qty\Bigg(
      \sup_{\substack{\tilde k \in \interval{-\pi/2}{\pi/2}^d,\\ r\in\set{1, m},\\ \theta \in \mathbb{T}}}
      \frac{\abs{\flt\varphi_p(\tilde k, r\Napier^{\imag\theta})}}{\abs*{\ft S_{\mu_p(r\Napier^{\imag\theta})}(\tilde k)}}
    ).
  \end{equation*}
  This completes the proof of Proposition~\ref{prp:restricting-g2}.
\end{proof}

\begin{proof}[Proof of Lemma~\ref{lem:upperBound-Green}]
  Let $\xi \defeq \ft D(k)$ and
  \begin{align*}
    f_r(\xi, \theta) &\defeq 4 \abs{1 - r \Napier^{\imag\theta} \xi}^2 - \left(\frac{\abs{\theta}}{\pi} + 1 - \xi\right)^2\\
    &= \left(4 r^2 - 1\right) \xi^2
      + 2 \left(\frac{\abs{\theta}}{\pi} + 1 - 4r \cos\theta\right) \xi
      + \underbrace{3 - \frac{\theta^2}{\pi^2} - \frac{2 \abs{\theta}}{\pi}}_{(\dagger)}.
  \end{align*}
  $f_r(\xi, \theta)$ is an even function with respect to $\theta$.
  It suffices to show that $f_r(\xi, \theta)\geq 0$ for $\xi\in\interval{0}{1}$ and $\theta\in\interval{0}{\pi}$.
  Clearly, the part ($\dagger$) is monotonically decreasing with respect to $\theta$ and non-negative, hence $f_r(0, \theta) \geq 0$.
  Solving the equation $\pdv*{f_r(\xi, \theta)}{r}=0$, we obtain $r = \cos\theta / \xi$.
  By substituting this into $f_r(\xi, \theta)$,
  \[
    f_{\cos\theta / \xi}(\xi, \theta) = -\xi^2 + 2 \left(\frac{\theta}{\pi} + 1\right) \xi - 4 \cos^2\theta + 3 - \frac{\theta^2}{\pi^2} - \frac{2\theta}{\pi}.
  \]
  Since the coefficient of $\xi^2$ is negative, $f_{\cos\theta / \xi}(\xi, \theta)$ takes the minimum value at the boundary points $\xi=0, 1$.
  Since $\theta$ takes a value on $\interval{0}{\pi}$, we arrive at $f_{\cos\theta}(1, \theta) = 4 \sin^2\theta + 3 - \theta^2 / \pi^2 \geq 0$.
  This completes the proof of Lemma~\ref{lem:upperBound-Green}.
\end{proof}

\begin{proof}[Proof of Lemma~\ref{lem:mu-bound}]
  Let $r = \abs{\mu} \in \interval{0}{1}$ and $\theta = \arg \mu \in \mathbb{T}$.  We show that
  \begin{align*}
    f(\theta) &\defeq 4 \abs{1 - r\Napier^{\imag\theta} \ft D(k)}^2 - \abs{1 - \Napier^{\imag\theta} \ft D(k)}^2\\
    &= 3 + \qty(4r^2 - 1) \ft D(k)^2 - 2 \qty(4r - 1) \ft D(k) \cos\theta\\
    &\geq 0.
  \end{align*}
  By the symmetry of the cosine function, we restrict the range of $\theta$ to $\interval{0}{\pi}$.  Since
  \[
    f'(\theta) = 2 \qty(4r - 1) \ft D(k) \sin\theta \begin{cases}
      \geq 0 & [r \geq 1 / 4 \qand \ft D(k) \geq 0],\\
      \leq 0 & [r < 1 / 4 \qand \ft D(k) \geq 0],\\
      \leq 0 & [r \geq 1 / 4 \qand \ft D(k) < 0],\\
      \geq 0 & [r < 1 / 4 \qand \ft D(k) < 0],
    \end{cases}
  \]
  $f(\theta)$ is monotonic with respect to $\theta$ whenever $r$ and $k$ are fixed.
  Thus, $f(\theta)$ takes the minimum value at $\theta = 0$, $\pi / 2$ or $\pi$.
  When $\theta = 0$ or $\pi$, we need to show $3 + \qty(4x^2 - 1) y^2 \pm 2 \qty(4x - 1) y \geq 0$
  for all $(x, y)\in\interval{0}{1}\times\interval{-1}{1}$.  When $\theta = \pi / 2$, we need to show
  $3 + \qty(4x^2 - 1) y^2 \geq 0$ for all $(x, y)\in\interval{0}{1}\times\interval{-1}{1}$.
  We omit proofs of these inequalities because it is not hard to see them.
  Therefore, $f(\theta) \geq 0$.
\end{proof}

\begin{proof}[Proof of Lemma~\ref{lem:trigonometricInequality}]
  Let
  \begin{gather*}
    \flt a^\mathrm{cos}(l, k; z) = \sum_{(x, t)} a(x, t) \cos(l\cdot x) \cos(k\cdot x) z^t,\\
    \flt a^\mathrm{sin}(l, k; z) = \sum_{(x, t)} a(x, t) \sin(l\cdot x) \sin(k\cdot x) z^t.
  \end{gather*}
  As in \cite[Equation~(5.21)]{s06},
  \begin{multline}
    -\frac{1}{2}\DLaplacian{k}\flt A(l, z)\\
    = \frac{\flt A(l + k, z) + \flt A(l - k, z)}{2} \flt A(l, z) \bigl(\flt a(l, z) - \flt a^{\cos}(l, k; z)\bigr)\\
      - \flt A(l, z) \flt A(l + k, z) \flt A(l - k, z) \flt a^{\sin}(l, k; z)^2.
  \end{multline}
  By the triangle inequality and arithmetic of the trigonometric functions, combining
  \begin{equation*}
    \abs{\flt a(l, z) - \flt a^{\cos}(l, k; z)} \leq \sum_{(x, t)}\abs{a(x, t)} \abs{z}^t (1 - \cos k\cdot x)
  \end{equation*}
  and
  \begin{align*}
    \abs{\flt a^{\sin}(l, k; z)}^2
    &\leq \qty(\sum_{(x, t)}\abs{a(x, t)} \abs{z}^t \sin^2 l\cdot x) \qty(\sum_{(x, t)}\abs{a(x, t)} \abs{z}^t \sin^2 k\cdot x)\\
    &= \qty(\sum_{(x, t)}\abs{a(x, t)} \abs{z}^t \frac{1 - \cos(2l\cdot x)}{2}) \qty(\sum_{(x, t)}\abs{a(x, t)} \abs{z}^t (1 - \cos^2 k\cdot x))\\
    &\leq \qty(\sum_{(x, t)}\abs{a(x, t)} \abs{z}^t \qty\big(1 - \cos(2l\cdot x))) \qty(\sum_{(x, t)}\abs{a(x, t)} \abs{z}^t (1 - \cos k\cdot x))
  \end{align*}
  leads to \eqref{eq:modified-trigonometricInequality}.
\end{proof}

From \eqref{eq:improvedBound-g1}, \eqref{eq:improvedBound-g2} and \eqref{eq:improvedBound-g3}, upper bounds for $\{g_i(p, m)\}_{i=1}^{3}$ are written in terms of the lace expansion coefficients in the end.
Each term in the upper bounds is also bounded above by Lemma~\ref{lem:precise-diagrammaticBounds}, which is written in terms of the basic diagrams $B_{p, m}^{(\lambda, \rho)}$, $T_{p, m}^{(\lambda, \rho)}$
and $V_{p, m}^{(\lambda, \rho)}(k)$.
The rest of our work is to bound the basic diagrams.

\section{Diagrammatic bounds in terms of random-walk quantities}\label{sec:bounds-basicDiagrams}

In this section, we bound the basic diagrams \eqref{eq:def-bubble}--\eqref{eq:def-weighted-bubble} above.
Recall the random-walk quantities \eqref{eq:def-rwQuantities}.
We first show the proof of Lemma~\ref{lem:basic-diagrams} below in Section~\ref{sec:proof-basicDiagrams} with additional lemmas,
and then we prove their lemmas in Section~\ref{sec:proof-lemmas-basicDiagrams}.

\begin{lemma} \label{lem:basic-diagrams}
  Assume $g_i(p, m)\leq K_i$, $i=1, 2, 3$.
  For $\mu, \rho\in\mathbb{N}$, $p\in\interval[open right]{0}{p_\mathrm{c}}$ and $m\in\interval[open right]{1}{m_p}$,
  \begin{gather}
    \label{eq:bubble-diagram}
    B_{p, m}^{(\lambda, \rho)} \leq K_1^{\lambda + \rho} K_2^2 \varepsilon_1^{(\floor{(\lambda + \rho) / 2})},\\
    \label{eq:triangle-diagram}
    T_{p, m}^{(\lambda, \rho)} \leq \sqrt{2} K_1^{\lambda + \rho} K_2^3 \varepsilon_2^{(\floor{(\lambda + \rho) / 2})},\\
    \label{eq:weighted-bubble}
    \norm{\frac{\flt V_{p, m}^{(\lambda, \rho)}}{1 - \ft D}}_\infty \leq \begin{dcases}
      K_1^2 \norm{D}_\infty + K_1^3 K_2 \sqrt{\varepsilon_1^{(3)}} + \norm{\frac{\flt V_{p, m}^{(2, 1)}}{1 - \ft D}}_\infty& [\lambda = \rho = 1],\\
      \begin{multlined}
        \lambda \left(\lambda - 1\right) K_1^{\lambda + \rho} K_2^2 \varepsilon_1^{(\floor{(\lambda + \rho - 1) / 2})}\\
        + \lambda K_1^{\lambda + \rho - 1} K_2 K_3 \bigl(\sqrt{2} + 4\bigr) \varepsilon_2^{(\floor{(\lambda + \rho - 1) / 2})}
      \end{multlined} & [\lambda \geq 2\text{ or }\rho \geq 2].
    \end{dcases}
  \end{gather}
\end{lemma}

\subsection{Upper bounds on the basic diagrams}\label{sec:proof-basicDiagrams}

The residue theorem in complex analysis yields the next lemma, which is used throughout this section.

\begin{lemma}\label{lem:applicationOfResidue}
  For every $r \in \interval[open right]{0}{1}$ and $k \in \mathbb{T}$,
  \begin{gather}
    \int_{\mathbb{T}} \abs{\ft S_{r \Napier^{\imag\theta}}(k)}^2 \frac{\dd{\theta}}{2\pi}
      = \frac{1}{1 - r^2 \ft D(k)^2},
      \label{eq:integral-ComplexRWGreen-2ndPower}\\
    \int_{\mathbb{T}} \abs{\ft S_{r \Napier^{\imag\theta}}(k)}^4 \frac{\dd{\theta}}{2\pi}
      = \frac{1 + r^2 \ft D(k)^2}{\bigl(1 - r^2 \ft D(k)^2\bigr)^3}.
      \label{eq:integral-ComplexRWGreen-4thPower}
  \end{gather}
\end{lemma}

\begin{proof}[Proof of \eqref{eq:bubble-diagram} and \eqref{eq:triangle-diagram} in Lemma~\ref{lem:basic-diagrams}]
  We first show \eqref{eq:bubble-diagram}.
  By the inverse Fourier-Laplace transform and the bootstrap hypotheses,
  \begin{align*}
    B_{p, m}^{(\lambda, \rho)}
    &= \sup_{(x, t)} \abs{
      \int_{\mathbb{T}^d}\frac{\dd[d]{k}}{(2\pi)^d}\int_{\mathbb{T}}\frac{\dd{\theta}}{2\pi}
      \flt q_p(k, \Napier^{\imag\theta})^\lambda \flt\varphi_p(k, \Napier^{\imag\theta})
      m^\rho \flt q_p(k, \Napier^{-\imag\theta})^\rho \flt\varphi_p(k, m\Napier^{-\imag\theta})
      \Napier^{-\imag k\cdot x} \Napier^{-\imag\theta t}
    }\\
    &\leq \int_{\mathbb{T}^d}\frac{\dd[d]{k}}{(2\pi)^d}\int_{\mathbb{T}}\frac{\dd{\theta}}{2\pi} \abs{\flt q_p(k, \Napier^{\imag\theta})}^\lambda \abs{\flt\varphi_p(k, \Napier^{\imag\theta})} m^\rho \abs{\flt q_p(k, \Napier^{-\imag\theta})}^\rho \abs{\flt\varphi_p(k, m\Napier^{-\imag\theta})}\\
    &= p^{\lambda + \rho} m^\rho \int_{\mathbb{T}^d}\frac{\dd[d]{k}}{(2\pi)^d}\int_{\mathbb{T}}\frac{\dd{\theta}}{2\pi} \abs{\ft D(k)}^{\lambda + \rho} \abs{\flt\varphi_p(k, \Napier^{\imag\theta})} \abs{\flt\varphi_p(k, m\Napier^{-\imag\theta})}\\
    &\leq K_1^{\lambda + \rho} K_2^2 \int_{\mathbb{T}^d}\frac{\dd[d]{k}}{(2\pi)^d}\int_{\mathbb{T}}\frac{\dd{\theta}}{2\pi} \abs{\ft D(k)}^{\lambda + \rho} \abs{\ft S_{\mu_p(\Napier^{\imag\theta})}(k)} \abs{\ft S_{\mu_p(m\Napier^{-\imag\theta})}(k)}.
  \end{align*}
  By the Cauchy-Schwarz inequality,
  \begin{equation*}
    B_{p, m}^{(\lambda, \rho)}
    \leq K_1^{\lambda + \rho} K_2^2 \int_{\mathbb{T}^d}\frac{\dd[d]{k}}{(2\pi)^d} \abs{\ft D(k)}^{\lambda + \rho} \qty(\int_{\mathbb{T}}\frac{\dd{\theta}}{2\pi} \abs{\ft S_{\mu_p(\Napier^{\imag\theta})}(k)}^2)^{1 / 2} \qty(\int_{\mathbb{T}}\frac{\dd{\theta}}{2\pi} \abs{\ft S_{\mu_p(m\Napier^{-\imag\theta})}(k)}^2)^{1 / 2}.
  \end{equation*}
  Recall \eqref{eq:scaledParameter}.
  We apply \eqref{eq:integral-ComplexRWGreen-2ndPower} in Lemma~\ref{lem:applicationOfResidue} to the above to obtain
  \begin{align*}
    B_{p, m}^{(\lambda, \rho)}
    &\leq  K_1^{\lambda + \rho} K_2^2 \int_{\mathbb{T}^d}\frac{\dd[d]{k}}{(2\pi)^d} \abs{\ft D(k)}^{\lambda + \rho} \qty(\frac{1}{1 - \mu_p(1)^2 \ft D(k)^2})^{1 / 2} \qty(\frac{1}{1 - \mu_p(m)^2 \ft D(k)^2})^{1 / 2}\\
    &\leq K_1^{\lambda + \rho} K_2^2 \int_{\mathbb{T}^d}\frac{\dd[d]{k}}{(2\pi)^d} \frac{\ft D(k)^{2\floor{\frac{\lambda + \rho}{2}}}}{1 - \ft D(k)^2}\\
    &= K_1^{\lambda + \rho} K_2^2 \varepsilon_1^{(\floor{(\lambda + \rho) / 2})}.
  \end{align*}
  In the last line, we have used $\mu_p(m) < 1$ for every $m\in\interval[open right]{0}{m_p}$,
  and the trivial inequality $\abs*{\ft D(k)}^{\lambda + \rho} \leq \ft D(k)^{2\floor{\frac{\lambda + \rho}{2}}}$
  since $\abs*{\ft D(k)}\leq 1$ for all $k\in\mathbb{T}^d$.

  Similarly,
  \begin{align*}
    T_{p, m}^{(\lambda, \rho)}
    &\mathmakebox[\widthof{CS ineq.}]{\overset{\text{inv. FL}}{=}}
      \sup_{(x, t)} \abs{
        \int_{\mathbb{T}^d}\frac{\dd[d]{k}}{(2\pi)^d}\int_{\mathbb{T}}\frac{\dd{\theta}}{2\pi}
        \flt q_p(k, \Napier^{\imag\theta})^\lambda \flt\varphi_p(k, \Napier^{\imag\theta})^2
        m^\rho \flt q_p(k, \Napier^{-\imag\theta})^\rho \flt\varphi_p(k, m\Napier^{-\imag\theta})
        \Napier^{-\imag k\cdot x} \Napier^{-\imag\theta t}
      }\\
    &\mathmakebox[\widthof{CS ineq.}]{\leq}
      \int_{\mathbb{T}^d}\frac{\dd[d]{k}}{(2\pi)^d}\int_{\mathbb{T}}\frac{\dd{\theta}}{2\pi}
      \abs{\flt q_p(k, \Napier^{\imag\theta})}^\lambda \abs{\flt\varphi_p(k, \Napier^{\imag\theta})}^2
      m^\rho \abs{\flt q_p(k, \Napier^{-\imag\theta})}^\rho \abs{\flt\varphi_p(k, m\Napier^{-\imag\theta})}\\
    &\mathmakebox[\widthof{CS ineq.}]{=}
      p^{\lambda + \rho} m^\rho
      \int_{\mathbb{T}^d}\frac{\dd[d]{k}}{(2\pi)^d}\int_{\mathbb{T}}\frac{\dd{\theta}}{2\pi}
      \abs{\ft D(k)}^{\lambda + \rho} \abs{\flt\varphi_p(k, \Napier^{\imag\theta})}^2 \abs{\flt\varphi_p(k, m\Napier^{-\imag\theta})}\\
    &\mathmakebox[\widthof{CS ineq.}]{\overset{\text{hypoth.}}{\leq}}
      K_1^{\lambda + \rho} K_2^3
      \int_{\mathbb{T}^d}\frac{\dd[d]{k}}{(2\pi)^d}\int_{\mathbb{T}}\frac{\dd{\theta}}{2\pi}
      \abs{\ft D(k)}^{\lambda + \rho} \abs{\ft S_{\mu_p(\Napier^{\imag\theta})}(k)}^2 \abs{\ft S_{\mu_p(m\Napier^{-\imag\theta})}(k)}\\
    &\mathmakebox[\widthof{CS ineq.}]{\overset{\text{CS ineq.}}{\leq}}
      K_1^{\lambda + \rho} K_2^3
      \int_{\mathbb{T}^d}\frac{\dd[d]{k}}{(2\pi)^d} \abs{\ft D(k)}^{\lambda + \rho}
      \qty(\int_{\mathbb{T}}\frac{\dd{\theta}}{2\pi} \abs{\ft S_{\mu_p(\Napier^{\imag\theta})}(k)}^4)^{1 / 2}
      \qty(\int_{\mathbb{T}}\frac{\dd{\theta}}{2\pi} \abs{\ft S_{\mu_p(m\Napier^{-\imag\theta})}(k)}^2)^{1 / 2}\\
    &\mathmakebox[\widthof{CS ineq.}]{\overset{\text{Lem.~\ref{lem:applicationOfResidue}}}{=}}
      K_1^{\lambda + \rho} K_2^3
      \int_{\mathbb{T}^d}\frac{\dd[d]{k}}{(2\pi)^d} \abs{\ft D(k)}^{\lambda + \rho}
      \qty(\frac{1 + \mu_p(1)^2 \ft D(k)^2}{\qty\big(1 - \mu_p(1)^2 \ft D(k)^2)^3})^{1 / 2}
      \qty(\frac{1}{1 - \mu_p(m)^2 \ft D(k)^2})^{1 / 2}\\
    &\mathmakebox[\widthof{CS ineq.}]{\leq}
      \sqrt{2} K_1^{\lambda + \rho} K_2^3
      \int_{\mathbb{T}^d}\frac{\dd[d]{k}}{(2\pi)^d}
      \frac{\ft D(k)^{2\floor{\frac{\lambda + \rho}{2}}}}{\qty\big(1 - \ft D(k)^2)^2}\\
    &\mathmakebox[\widthof{CS ineq.}]{=}
      \sqrt{2} K_1^{\lambda + \rho} K_2^3
      \sum_{n=0}^{\infty} \qty(n + 1)
      \int_{\mathbb{T}^d}\frac{\dd[d]{k}}{(2\pi)^d} \ft D(k)^{2 \qty(n + \floor{\frac{\lambda + \rho}{2}})}\\
    &\mathmakebox[\widthof{CS ineq.}]{\overset{\text{inv. FL}}{=}}
      \sqrt{2} K_1^{\lambda + \rho} K_2^3 \varepsilon_2^{(\floor{(\lambda + \rho) / 2})}.
  \end{align*}
  This completes the proof of \eqref{eq:triangle-diagram}.
\end{proof}

To prove \eqref{eq:weighted-bubble} in Lemma~\ref{lem:basic-diagrams}, integrating the product of the absolute value of $\ft D(l)$'s and $\ft S_{\mu}(l)$'s is required.
These integrands come from $\hat U_{\mu}(k, l)$ in \eqref{eq:2nd-derivative-RW}.

\begin{lemma}\label{lem:integrals-upperBoundOnGreen}
  For any $\lambda, \rho\in\mathbb{N}$,
  \begin{equation}
    \begin{aligned}[b]
      \MoveEqLeft[1] \int_{\mathbb{T}^d}\frac{\dd[d]{l}}{(2\pi)^d}\int_{\mathbb{T}}\frac{\dd{\theta}}{2\pi} \abs{\ft D(l)}^{\lambda + \rho - 1} \abs{\flt S_{\mu_p(\Napier^{\imag\theta})}(l \pm k)}\\
        &\times \abs{\flt S_{\mu_p(\Napier^{\imag\theta})}(l)} \abs{\flt S_{\mu_p(m\Napier^{-\imag\theta})}(l)}\\
      \leq {} & \sqrt{2} \varepsilon_2^{(\floor{(\lambda + \rho - 1) / 2})},
    \end{aligned}
    \label{eq:integral-type1-upperBoundOnGreen}
  \end{equation}
  and
  \begin{equation}
    \begin{aligned}[b]
      \MoveEqLeft[1] \int_{\mathbb{T}^d}\frac{\dd[d]{l}}{(2\pi)^d}\int_{\mathbb{T}}\frac{\dd{\theta}}{2\pi} \abs{\ft D(l)}^{\lambda + \rho - 1} \abs{\flt S_{\mu_p(\Napier^{\imag\theta})}(l + k)} \abs{\flt S_{\mu_p(\Napier^{\imag\theta})}(l - k)}\\
        &\times \abs{\flt S_{\mu_p(\Napier^{\imag\theta})}(l)} \abs{\flt S_{\mu_p(m\Napier^{-\imag\theta})}(l)} \qty\big(1 - \ft D(2l))\\
      \leq {} & 4 \varepsilon_2^{(\floor{(\lambda + \rho - 1) / 2})}.
    \end{aligned}
    \label{eq:integral-type2-upperBoundOnGreen}
  \end{equation}
\end{lemma}

To prove \eqref{eq:integral-type2-upperBoundOnGreen} in Lemma~\ref{lem:integrals-upperBoundOnGreen}, we note that the period of $\ft D(2k)$ is same as that of $\ft D(k)^2$.
The next lemma plays an important role to prove the finiteness of $\flt V_{p, m}^{(\lambda, \rho)}(k)$ for $d>4$.

\begin{lemma}\label{lem:bounds-wrt-D}
  For every $k\in\mathbb{T}^d$ and $\mu\in\mathbb{C}$ satisfying $\abs{\mu}\leq 1$,
  \begin{gather}
    1 - \ft D(2k) \leq 2 \qty(1 - \ft D(k)^2) \label{eq:equivalent-periodicity}.
  \end{gather}
\end{lemma}

\begin{remark}\label{rem:finiteness-weightedBubble}
  The proof of Lemma~\ref{lem:bounds-wrt-D} depends on the structure of $\opSpace$ as well as Proposition~\ref{prp:restricting-g2},
  but it is not hard to prove this on $\mathbb{Z}^d$ similarly \cite{k21,k22}.
  Thanks to this lemma, we are able to show the finiteness of $\flt V_{p, m}^{(\lambda, \rho)}(k)$
  for $d>4$, cf., \cite[Remark~5]{ny93} in which they proved the finiteness of the similar quantity only for $d>8$.
\end{remark}

\begin{proof}[Proof of \eqref{eq:weighted-bubble} in Lemma~\ref{lem:basic-diagrams} when $\lambda \geq 2$ and $\rho \geq 2$]
  By Lemma~\ref{lem:split-of-cosine} in Appendix~\ref{sec:proof-diagrammaticBounds},
  \begin{align*}
    \flt V_{p, m}^{(\lambda, \rho)}(k)
    = {} \MoveEqLeft[0] \sup_{\vb*{x}}\sum_{\vb*{y}_\lambda} \qty(q_p^{\Conv \lambda} \Conv \varphi_p)(\vb*{y}_\lambda) \qty(1 - \cos k\cdot y) \qty(m^\rho q_p^{\Conv \rho} \Conv \varphi_p^{(m)})(\vb*{y}_\lambda - \vb*{x})\\
    = {} & \sup_{\vb*{x}}\sum_{\vb*{y}_1, \dots, \vb*{y}_\lambda} \qty(\prod_{i=1}^{\lambda - 1} q_p(\vb*{y}_{i} - \vb*{y}_{i - 1})) \qty(q_p \Conv \varphi_p)(\vb*{y}_\lambda - \vb*{y}_{\lambda - 1})\\
      &\times \qty(1 - \cos\qty\bigg(k\cdot \sum_{j=1}^{\lambda}\qty(y_{j} - y_{j - 1})))
      \qty(m^\rho q_p^{\Conv \rho} \Conv \varphi_p^{(m)})(\vb*{y}_\lambda - \vb*{x})\\
    \leq {} & \lambda \sup_{\vb*{x}}\sum_{j=1}^{\lambda - 1} \sum_{\vb*{y}_1, \dots, \vb*{y}_\lambda} \qty(\prod_{i\neq j} q_p(\vb*{y}_{i} - \vb*{y}_{i - 1})) \qty(q_p \Conv \varphi_p)(\vb*{y}_\lambda - \vb*{y}_{\lambda - 1})\\
        &\quad \times q_p(\vb*{y}_{j} - \vb*{y}_{j - 1}) \qty\Big(1 - \cos\qty\big(k\cdot \qty(y_{j} - y_{j - 1})))\\
        &\quad \times \qty(m^\rho q_p^{\Conv \rho} \Conv \varphi_p^{(m)})(\vb*{y}_\lambda - \vb*{x})\\
      & + \lambda \sup_{\vb*{x}} \sum_{\vb*{y}_1, \dots, \vb*{y}_\lambda} \qty(q_p \Conv \varphi_p)(\vb*{y}_\lambda - \vb*{y}_{\lambda - 1})\\
        &\quad \times \qty\Big(1 - \cos\qty\big(k\cdot \qty(y_{\lambda} - y_{\lambda - 1})))\\
        &\quad \times \qty(\prod_{i=1}^{\lambda - 1} q_p(\vb*{y}_{i} - \vb*{y}_{i - 1})) \qty(m^\rho q_p^{\Conv \rho} \Conv \varphi_p^{(m)})(\vb*{y}_\lambda - \vb*{x})\\
    = {} & \sup_{\vb*{x}}\sum_{\vb*{y}_{j - 1}, \vb*{y}_{j}, \vb*{y}_\lambda}
        q_p^{\Conv (j - 1)}(\vb*{y}_{j - 1})
        q_p(\vb*{y}_{j} - \vb*{y}_{j - 1})\\
        &\quad \times \qty\Big(1 - \cos\qty\big(k\cdot \qty(y_{j} - y_{j - 1})))\\
        &\quad \times \qty(q_p^{\Conv (\lambda - j)} \Conv \varphi_p)(\vb*{y}_\lambda - \vb*{y}_{j})
        \qty(m^\rho q_p^{\Conv \rho} \Conv \varphi_p^{(m)})(\vb*{y}_\lambda - \vb*{x})\\
      & + \lambda \sup_{\vb*{x}} \sum_{\vb*{y}_{\lambda - 1}, \vb*{y}_\lambda}
        q_p^{\Conv (\lambda - 1)}(\vb*{y}_{\lambda - 1})\\
        &\quad \times
        \qty(q_p \Conv \varphi_p)(\vb*{y}_\lambda - \vb*{y}_{\lambda - 1}) \qty\Big(1 - \cos\qty\big(k\cdot \qty(y_{\lambda} - y_{\lambda - 1})))\\
        &\quad \times
        \qty(m^\rho q_p^{\Conv \rho} \Conv \varphi_p^{(m)})(\vb*{y}_\lambda - \vb*{x})
  \end{align*}
  where $\vb*{y}_0 = \vb*{o}$.
  By the inverse Fourier-Laplace transform,
  \begin{align}
    \MoveEqLeft \flt V_{p, m}^{(\lambda, \rho)}(k) \notag\\
    \leq {} & \lambda \qty(\lambda - 1) \int_{\mathbb{T}^d}\frac{\dd[d]{l}}{(2\pi)^d}\int_{\mathbb{T}}\frac{\dd{\theta}}{2\pi} \abs{\flt q_p(l, \Napier^{\imag\theta})}^{\lambda - 1} \notag\\
      &\quad \times \abs{\frac{1}{2}\DLaplacian{k} \flt q_p(l, \Napier^{\imag\theta})} \abs{\flt\varphi_p(l, \Napier^{\imag\theta})} m^\rho \abs{\flt q_p(l, \Napier^{-\imag\theta})}^\rho \abs{\flt\varphi_p(l, m\Napier^{-\imag\theta})} \notag\\
      &+ \lambda \int_{\mathbb{T}^d}\frac{\dd[d]{l}}{(2\pi)^d}\int_{\mathbb{T}}\frac{\dd{\theta}}{2\pi} \abs{\flt q_p(l, \Napier^{\imag\theta})}^{\lambda - 1} \notag\\
      &\quad \times \abs{\frac{1}{2}\DLaplacian{k} \qty(\flt q_p(l, \Napier^{\imag\theta}) \flt\varphi_p(l, \Napier^{\imag\theta}))} m^\rho \abs{\flt q_p(l, \Napier^{-\imag\theta})}^\rho \abs{\flt\varphi_p(l, m\Napier^{-\imag\theta})} \notag\\
    \leq {} & \lambda \qty(\lambda - 1) \qty\big(1 - \ft D(k)) K_1^{\lambda + \rho} K_2^2 \notag\\
      &\quad \times \underbrace{\int_{\mathbb{T}^d}\frac{\dd[d]{l}}{(2\pi)^d}\int_{\mathbb{T}}\frac{\dd{\theta}}{2\pi} \abs{\ft D(l)}^{\lambda + \rho - 1} \abs{\flt S_{\mu_p(\Napier^{\imag\theta})}(l)} \abs{\flt S_{\mu_p(m\Napier^{-\imag\theta})}(l)}}_{\leq \varepsilon_1^{(\floor{(\lambda + \rho - 1) / 2})}} \notag\\
      &+ \lambda K_1^{\lambda + \rho - 1} K_2 K_3 \int_{\mathbb{T}^d}\frac{\dd[d]{l}}{(2\pi)^d}\int_{\mathbb{T}}\frac{\dd{\theta}}{2\pi} \abs{\ft D(l)}^{\lambda + \rho - 1} \notag\\
      &\quad \times \abs{\flt U_{\mu_p(\Napier^{\imag\theta})}(k, l)} \abs{\flt S_{\mu_p(m\Napier^{-\imag\theta})}(l)}.
      \label{eq:estimation-weighted-bubble}
  \end{align}
  In the second inequality, we have used
  \[
    \abs{\frac{1}{2}\DLaplacian{k} \ft D(l)}
    = \abs{\sum_{x} D(x) \qty(1 - \cos k\cdot x) \Napier^{\imag l\cdot x}}
    \leq \sum_{x} D(x) \qty(1 - \cos k\cdot x)
    = 1 - \ft D(k)
  \]
  and the definition \eqref{eq:2nd-derivative-RW} of $\flt U_{\mu_p(\Napier^{\imag\theta})}(k, l)$.
  The second term in the right most in \eqref{eq:estimation-weighted-bubble} contains two types of integrals that appear in Lemma~\ref{lem:integrals-upperBoundOnGreen}.
  Combining \eqref{eq:estimation-weighted-bubble}, \eqref{eq:integral-type1-upperBoundOnGreen}
  and \eqref{eq:integral-type2-upperBoundOnGreen} implies \eqref{eq:weighted-bubble} when $\lambda \geq 2$ and $\rho \geq 2$.
\end{proof}

\begin{proof}[Proof of \eqref{eq:weighted-bubble} in Lemma~\ref{lem:basic-diagrams} when $\lambda = 1$ and $\rho = 1$]
  By the trivial inequality \eqref{eq:more-one-step},
  note that $\varphi_p(\vb*{x}) = \KroneckerDelta{\vb*{o}}{\vb*{x}} + \varphi_p(\vb*{x}) \ind{\vb*{o}\neq\vb*{x}} \leq \KroneckerDelta{\vb*{o}}{\vb*{x}} + (q_p\Conv\varphi_p)(\vb*{x})$ for every $\vb*{x}\in\opSpaceTime$.
  Then, we obtain
  \begin{align*}
    \flt V_{p, m}^{(1, 1)}(k)
    = {} \MoveEqLeft[0] \sup_{\vb*{x}}\sum_{\vb*{y}} \qty(q_p \Conv \varphi_p)(\vb*{y}) \qty(1 - \cos k\cdot y) \qty(m q_p \Conv \varphi_p^{(m)})(\vb*{y} - \vb*{x})\\
    \leq {} & \underbrace{
          \sup_{\vb*{x}}\sum_{\vb*{y}} q_p(\vb*{y}) \qty(1 - \cos k\cdot y) m q_p(\vb*{y} - \vb*{x})
        }_{
          \leq K_1^2 \norm{D}_\infty (1 - \ft D(k))
          \quad
          (\because \eqref{eq:decomposingDiagrams})
        }\\
      &+ \sup_{\vb*{x}}\sum_{\vb*{y}} q_p(\vb*{y}) \qty(1 - \cos k\cdot y) \qty(m^2 q_p^{\Conv 2} \Conv \varphi_p^{(m)})(\vb*{y} - \vb*{x})\\
      &+ \underbrace{
          \sup_{\vb*{x}}\sum_{\vb*{y}} \qty(q_p^{\Conv 2} \Conv \varphi_p)(\vb*{y}) \qty(1 - \cos k\cdot y) \qty(m q_p \Conv \varphi_p^{(m)})(\vb*{y} - \vb*{x})
        }_{
          = \flt V_{p, m}^{(2, 1)}(k)
        }.
  \end{align*}
  It is not hard to bound the second term above by $K_1^3 K_2 \sqrt{\varepsilon_1^{(3)}} (1 - \ft D(k))$ by the same method as the proof of the bounds on $B_{p, m}^{(\lambda, \rho)}$ and $T_{p, m}^{(\lambda, \rho)}$.
  Notice that we apply the Cauchy-Schwarz inequality to it twice, which changes the $1$st power of $\abs*{\ft S_{\mu_p(m \Napier^{-\imag\theta})}}$ to the $2$nd power and the $1/2$th power of $1 - \ft D(k)^2$ to the $1$st power.
  This completes the proof of \eqref{eq:weighted-bubble} when $\lambda = 1$ and $\rho = 1$.
\end{proof}

\subsection{Proof of lemmas}\label{sec:proof-lemmas-basicDiagrams}

\begin{proof}[Proof of Lemma~\ref{lem:applicationOfResidue}]
  By the equality $\ft S_{r \Napier^{\imag\theta}}(k) = (1 - r \Napier^{\imag\theta} \ft D(k))^{-1}$ and the change of variables $z = \Napier^{\imag\theta}$,
  \begin{align*}
    \int_{\mathbb{T}} \abs{\ft S_{r \Napier^{\imag\theta}}(k)}^2 \frac{\dd{\theta}}{2\pi}
    &= \int_{\mathbb{T}} \frac{1}{\abs\big{1 - \Napier^{\imag\theta} r \ft D(k)}^2} \frac{\dd{\theta}}{2\pi}\\
    &= \int_{\mathbb{T}} \frac{1}{\bigl(1 - \Napier^{\imag\theta} r \ft D(k)\bigr) \bigl(1 - \Napier^{-\imag\theta} r \ft D(k)\bigr)} \frac{\dd{\theta}}{2\pi}\\
    &= \int_{\mathbb{T}} \frac{\Napier^{\imag\theta}}{\bigl(\Napier^{\imag\theta} - r \ft D(k)\bigr) \bigl(1 - \Napier^{\imag\theta} r \ft D(k)\bigr)} \frac{\dd{\theta}}{2\pi}\\
    &= \frac{1}{2\pi\imag} \oint_{\set{z | \abs*{z} = 1}} \frac{\dd z}{\bigl(z - r \ft D(k)\bigr) \bigl(1 - z r \ft D(k)\bigr)}.
  \end{align*}
  Note that $\abs*{r \ft D(k)} < 1$.
  Only the complex number $z = r \ft D(k)$ is a singularity in the disk $\set{z | \abs*{z} < 1}$.
  Applying the residue theorem to the above line integral leads to \eqref{eq:integral-ComplexRWGreen-2ndPower}.

  Similarly, the equality
  \[
    \int_{\mathbb{T}} \abs{\ft S_{r \Napier^{\imag\theta}}(k)}^4 \frac{\dd{\theta}}{2\pi}
    = \frac{1}{2\pi\imag} \oint_{\set{z | \abs*{z} = 1}} \frac{z \dd z}{\bigl(z - r \ft D(k)\bigr)^2 \bigl(1 - z r \ft D(k)\bigr)^2}
  \]
  and the residue theorem imply \eqref{eq:integral-ComplexRWGreen-4thPower}.
\end{proof}

\begin{proof}[Proof of \eqref{eq:integral-type1-upperBoundOnGreen} in Lemma~\ref{lem:integrals-upperBoundOnGreen}]
  By the Cauchy-Schwarz inequality and Lemma~\ref{lem:applicationOfResidue},
  \begin{align*}
    [\text{LHS of \eqref{eq:integral-type1-upperBoundOnGreen}}]
    \leq {} \MoveEqLeft[0] \int_{\mathbb{T}^d}\frac{\dd[d]{l}}{(2\pi)^d} \abs{\ft D(l)}^{\lambda + \rho - 1}
      \qty(\int_{\mathbb{T}}\frac{\dd{\theta}}{2\pi} \abs{\flt S_{\mu_p(\Napier^{\imag\theta})}(l \pm k)}^2)^{1 / 2}\\
      &\times
      \qty(\int_{\mathbb{T}}\frac{\dd{\theta}}{2\pi} \abs{\flt S_{\mu_p(\Napier^{\imag\theta})}(l)}^4)^{1 / 4}
      \qty(\int_{\mathbb{T}}\frac{\dd{\theta}}{2\pi} \abs{\flt S_{\mu_p(m\Napier^{-\imag\theta})}(l)}^4)^{1 / 4}\\
    = {} & \int_{\mathbb{T}^d}\frac{\dd[d]{l}}{(2\pi)^d} \abs{\ft D(l)}^{\lambda + \rho - 1}
      \qty(\frac{1}{1 - \mu_p(1)^2 \ft D(l \pm k)^2})^{1 / 2}\\
      &\times
      \qty(\frac{1 + \mu_p(1)^2 \ft D(l)^2}{\qty\big(1 - \mu_p(1)^2 \ft D(l)^2)^3})^{1 / 4}
      \qty(\frac{1 + \mu_p(m)^2 \ft D(l)^2}{\qty\big(1 - \mu_p(m)^2 \ft D(l)^2)^3})^{1 / 4}\\
    \leq {} & \sqrt{2} \int_{\mathbb{T}^d}\frac{\dd[d]{l}}{(2\pi)^d}
      \qty(\frac{\abs*{\ft D(l)}^{(\lambda + \rho - 1) / 2}}{1 - \ft D(l \pm k)^2})^{1 / 2}
      \qty(\frac{\abs*{\ft D(l)}^{3 (\lambda + \rho - 1) / 2}}{\qty\big(1 - \ft D(l)^2)^3})^{1 / 2}.
  \end{align*}
  By H\"{o}lder's inequality and the Cauchy-Schwarz inequality again,
  \begin{align*}
    \MoveEqLeft[1] [\text{LHS of \eqref{eq:integral-type1-upperBoundOnGreen}}]\\
    \leq {} & \sqrt{2}
      \qty(\int_{\mathbb{T}^d}\frac{\dd[d]{l}}{(2\pi)^d} \frac{\abs*{\ft D(l)}^{\lambda + \rho - 1}}{\qty\big(1 - \ft D(l \pm k)^2)^2})^{1 / 4}
      \qty(\int_{\mathbb{T}^d}\frac{\dd[d]{l}}{(2\pi)^d} \frac{\abs*{\ft D(l)}^{\lambda + \rho - 1}}{\qty\big(1 - \ft D(l)^2)^2})^{3 / 4}\\
    \leq {} & \sqrt{2}
      \qty(\sum_{n=0}^{\infty} \qty(n + 1) \int_{\mathbb{T}^d}\frac{\dd[d]{l}}{(2\pi)^d} \ft D(l)^{2 \floor{\frac{\lambda + \rho - 1}{2}}} \ft D(l \pm k)^{2n})^{1 / 4}\\
      &\times
      \qty(\sum_{n=0}^{\infty} \qty(n + 1) \int_{\mathbb{T}^d}\frac{\dd[d]{l}}{(2\pi)^d} \ft D(l)^{2 \qty(n + \floor{\frac{\lambda + \rho - 1}{2}})})^{3 / 4}\\
    = {} & \sqrt{2} \qty(
        \sum_{n=0}^{\infty} \qty(n + 1)
        \int_{\mathbb{T}^d}\frac{\dd[d]{l}}{(2\pi)^d} \ft D(l)^{2 \floor{\frac{\lambda + \rho - 1}{2}}}
        \sum_{x} D^{\conv 2n}(x) \Napier^{\imag (l \pm k)\cdot x}
      )^{1 / 4}\\
      &\times
      \qty(\varepsilon_2^{(\floor{(\lambda + \rho - 1) / 2})})^{3 / 4}\\
    = {} & \sqrt{2} \qty(\sum_{n=0}^{\infty} \qty(n + 1) \sum_{x} D^{\conv 2 \floor{\frac{\lambda + \rho - 1}{2}}}(x) D^{\conv 2n}(x) \Napier^{\pm\imag k\cdot x})^{1 / 4}\\
      &\times \qty(\varepsilon_2^{(\floor{(\lambda + \rho - 1) / 2})})^{3 / 4}.
  \end{align*}
  Since the summand with respect to $x$ is an even function, the factor $\Napier^{\imag (l \pm k)\cdot x}$ equals $\cos k\cdot x$, which is bounded above by $1$.
  Therefore, we arrive at \eqref{eq:integral-type1-upperBoundOnGreen}.
\end{proof}

\begin{proof}[Proof of \eqref{eq:integral-type2-upperBoundOnGreen} in Lemma~\ref{lem:integrals-upperBoundOnGreen}]
  By the Cauchy-Schwarz inequality and \eqref{eq:integral-ComplexRWGreen-4thPower} in Lemma~\ref{lem:applicationOfResidue},
  \begin{align*}
    \MoveEqLeft[1] [\text{LHS of \eqref{eq:integral-type2-upperBoundOnGreen}}]\\
    \leq {} & \int_{\mathbb{T}^d}\frac{\dd[d]{l}}{(2\pi)^d}
      \abs{\ft D(l)}^{\lambda + \rho - 1}\\
      &\times
      \qty(\int_{\mathbb{T}}\frac{\dd{\theta}}{2\pi} \abs{\flt S_{\mu_p(\Napier^{\imag\theta})}(l + k)}^4)^{1 / 4}
      \qty(\int_{\mathbb{T}}\frac{\dd{\theta}}{2\pi} \abs{\flt S_{\mu_p(\Napier^{\imag\theta})}(l - k)}^4)^{1 / 4}\\
      &\times
      \qty(\int_{\mathbb{T}}\frac{\dd{\theta}}{2\pi} \abs{\flt S_{\mu_p(\Napier^{\imag\theta})}(l)}^4)^{1 / 4}
      \qty(\int_{\mathbb{T}}\frac{\dd{\theta}}{2\pi} \abs{\flt S_{\mu_p(m\Napier^{-\imag\theta})}(l)}^4)^{1 / 4}
      \qty\big(1 - \ft D(2l))\\
    = {} & \int_{\mathbb{T}^d}\frac{\dd[d]{l}}{(2\pi)^d}
      \abs{\ft D(l)}^{\lambda + \rho - 1}\\
      &\times
      \qty(\frac{1 + \mu_p(1)^2 \ft D(l + k)^2}{\qty\big(1 - \mu_p(1)^2 \ft D(l + k)^2)^3})^{1 / 4}
      \qty(\frac{1 + \mu_p(1)^2 \ft D(l - k)^2}{\qty\big(1 - \mu_p(1)^2 \ft D(l - k)^2)^3})^{1 / 4}\\
      &\times
      \qty(\frac{1 + \mu_p(1)^2 \ft D(l)^2}{\qty\big(1 - \mu_p(1)^2 \ft D(l)^2)^3})^{1 / 4}
      \qty(\frac{1 + \mu_p(m)^2 \ft D(l)^2}{\qty\big(1 - \mu_p(m)^2 \ft D(l)^2)^3})^{1 / 4}
      \qty\big(1 - \ft D(2l))\\
    \leq {} & 2 \int_{\mathbb{T}^d}\frac{\dd[d]{l}}{(2\pi)^d}
      \qty(\frac{\abs*{\ft D(l)}^{3 (\lambda + \rho - 1) / 2}}{\qty\big(1 - \ft D(l + k)^2)^3})^{1 / 4}
      \qty(\frac{\abs*{\ft D(l)}^{3 (\lambda + \rho - 1) / 2}}{\qty\big(1 - \ft D(l - k)^2)^3})^{1 / 4}\\
      &\times
      \qty(\frac{\abs*{\ft D(l)}^{(\lambda + \rho - 1) / 2}}{\qty\big(1 - \ft D(l)^2)^3})^{1 / 2}
      \qty\big(1 - \ft D(2l)).
  \end{align*}
  By Lemma~\ref{lem:bounds-wrt-D},
  \begin{multline*}
    [\text{LHS of \eqref{eq:integral-type2-upperBoundOnGreen}}]
    \leq 4 \int_{\mathbb{T}^d}\frac{\dd[d]{l}}{(2\pi)^d}
      \qty(\frac{\abs*{\ft D(l)}^{3 (\lambda + \rho - 1) / 2}}{\qty\big(1 - \ft D(l + k)^2)^3})^{1 / 4}\\
      \times
      \qty(\frac{\abs*{\ft D(l)}^{3 (\lambda + \rho - 1) / 2}}{\qty\big(1 - \ft D(l - k)^2)^3})^{1 / 4}
      \qty(\frac{\abs*{\ft D(l)}^{(\lambda + \rho - 1) / 2}}{1 - \ft D(l)^2})^{1 / 2}.
  \end{multline*}
  By H\"{o}lder's inequality and the Cauchy-Schwarz inequality again,
  \begin{align*}
    \MoveEqLeft[1] [\text{LHS of \eqref{eq:integral-type2-upperBoundOnGreen}}]\\
    \leq {} & 4 \qty(\int_{\mathbb{T}^d}\frac{\dd[d]{l}}{(2\pi)^d}
      \frac{\abs*{\ft D(l)}^{(\lambda + \rho - 1) / 2} \cdot \abs*{\ft D(l)}^{(\lambda + \rho - 1) / 2}}{\qty\big(1 - \ft D(l + k)^2) \qty\big(1 - \ft D(l - k)^2)})^{3 / 4}\\
      &\times
      \qty(\int_{\mathbb{T}^d}\frac{\dd[d]{l}}{(2\pi)^d} \frac{\abs*{\ft D(l)}^{\lambda + \rho - 1}}{\qty\big(1 - \ft D(l)^2)^2})^{1 / 4}\\
    \leq {} & 4
      \qty(\int_{\mathbb{T}^d}\frac{\dd[d]{l}}{(2\pi)^d} \frac{\abs*{\ft D(l)}^{\lambda + \rho - 1}}{\qty\big(1 - \ft D(l + k)^2)^2})^{3 / 8}
      \qty(\int_{\mathbb{T}^d}\frac{\dd[d]{l}}{(2\pi)^d} \frac{\abs*{\ft D(l)}^{\lambda + \rho - 1}}{\qty\big(1 - \ft D(l - k)^2)^2})^{3 / 8}\\
      &\times
      \qty(\int_{\mathbb{T}^d}\frac{\dd[d]{l}}{(2\pi)^d} \frac{\abs*{\ft D(l)}^{\lambda + \rho - 1}}{\qty\big(1 - \ft D(l)^2)^2})^{1 / 4}\\
    \leq {} & 4 \varepsilon_2^{(\floor{(\lambda + \rho - 1) / 2})}.
  \end{align*}
\end{proof}

\begin{proof}[Proof of Lemma~\ref{lem:bounds-wrt-D}]
  Note that
  \begin{align*}
    2 \qty(1 - \ft D(k)^2) - \qty(1 - \ft D(2k))
    &= 1 - 2 \prod_{j=1}^{d} \cos^2 k_j + \prod_{j=1}^{d} \cos(2k_j)\\
    &= 1 - 2 \prod_{j=1}^{d} \cos^2 k_j + \prod_{j=1}^{d} \qty(2 \cos^2 k_j - 1)
  \end{align*}
  on $\opSpace$.  Let $h_d(\xi_1, \dots, \xi_d) \defeq 1 - 2 \prod_{j=1}^{d} \xi_j + \prod_{j=1}^{d} \qty(2 \xi_j - 1)$.
  It suffices to show that $h_d(\xi_1, \dots, \xi_d) \allowbreak \geq 0$ for every $(\xi_1, \dots, \xi_d) \in \interval{0}{1}^d$.
  To do so, we use the method of induction.

  When $d = 1$, $h_1(\xi_1) = 0$ for every $\xi_1$.  Next, we assume that $h_d(\xi_1, \dots, \xi_d) \allowbreak \geq 0$ holds in some dimension $d$.
  Fix $\xi_j$ for each $j = 1, \dots, d$.  Then, $h_{d+1}(\xi_1, \dots, \xi_{d+1})$ is regarded as a linear function with respect to
  $\xi_{d+1}$ and hence takes the minimum value at either the boundary point $0$ or $1$.  When $\xi_{d+1}=0$, clearly,
  \[
    h_{d+1}(\xi_1, \dots, \xi_d, 0) = 1 - \prod_{j=1}^{d} \qty(2 \xi_j - 1) \geq 0.
  \]
  When $\xi_{d+1}=1$, by the induction hypothesis,
  \[
    h_{d+1}(\xi_1, \dots, \xi_d, 1) = 1 - 2 \prod_{j=1}^{d} \xi_j + \prod_{j=1}^{d} \qty(2 \xi_j - 1) \geq 0.
  \]
  Therefore, $h_d(\xi_1, \dots, \xi_d) \geq 0$ holds in every dimension $d$.
\end{proof}

\section{Completing the bootstrap argument}\label{sec:proof-bootstrapArgument}

In this section, we prove Proposition~\ref{prp:op-bootstrapArgument} by a computer-assisted approach.

\begin{proof}[Proof of Proposition~\ref{prp:op-bootstrapArgument}]
  Set
  \begin{equation}
    \label{eq:numericalValuesForBootstrap}
    d = \num{9}, \quad
    K_1 = \num{1.0020}, \quad
    K_2 = \num{1.0500}, \quad
    K_3 = \num{1.2500},
  \end{equation}
  which satisfy the initial conditions of Proposition~\ref{prp:op-initialCondition}.
  From Lemma~\ref{lem:basic-diagrams} and Table~\ref{tbl:numericalValues-rw-bcc} in Appendix~\ref{sec:numericalComputations-rwQuantities}, the basic diagrams are bounded above as
  \begin{align*}
    B_{p, m}^{(0, 2)} &\leq \num{2.37279e-3},&
    B_{p, m}^{(2, 2)} &\leq \num{2.11688e-4},\\
    T_{p, m}^{(1, 1)} &\leq \num{3.96190e-3},&
    T_{p, m}^{(2, 2)} &\leq \num{4.40247e-4},\\
    \frac{\flt V_{p, m}^{(1, 1)}(k)}{1 - \ft D(k)} &\leq \num{4.81033e-2},&
    \frac{\flt V_{p, m}^{(1, 2)}(k)}{1 - \ft D(k)} &\leq \num{1.71979e-2},\\
    \frac{\flt V_{p, m}^{(2, 1)}(k)}{1 - \ft D(k)} &\leq \num{3.91493e-2},&
    \frac{\flt V_{p, m}^{(2, 2)}(k)}{1 - \ft D(k)} &\leq \num{3.92276e-2},\\
    \frac{\flt V_{p, m}^{(1, 3)}(k)}{1 - \ft D(k)} &\leq \num{1.72315e-2},&
    \frac{\flt V_{p, m}^{(3, 1)}(k)}{1 - \ft D(k)} &\leq \num{6.59882e-2}.
  \end{align*}
  The 4th inequality in the above implies the sufficient condition of Lemma~\ref{lem:precise-diagrammaticBounds}.
  In addition, from Lemma~\ref{lem:precise-diagrammaticBounds},
  \begin{gather*}
    \flt\Pi_p^\mathrm{even}(0, m) \leq \num{2.30399e-6},\\
    \flt\Pi_p^\mathrm{odd}(0, m) \leq \num{1.09838e-4},\\
    \sum_{(x, t)} \qty(\Pi_p^\mathrm{even}(x, t) + \Pi_p^\mathrm{odd}(x, t)) m^t t
      \leq \num{3.84797e-4},\\
    \sum_{(x, t)} \qty(\Pi_p^\mathrm{even}(x, t) + \Pi_p^\mathrm{odd}(x, t)) m^t \frac{1 - \cos k\cdot x}{1 - \ft D(k)}
      \leq \num{2.02991e-2},
  \end{gather*}
  which imply the absolute convergence of the alternating series \eqref{eq:def-sumOfLaceExpansionCoefficients}.
  The 1st and 2nd inequalities in the above imply the sufficient condition of Lemma~\ref{lem:improvedBounds}.
  Then, by \eqref{eq:improvedBound-g1}, \eqref{eq:improvedBound-g2} and \eqref{eq:improvedBound-g3}, we obtain
  \begin{align}
    g_1(p, m) &\leq \num{1.0002} < K_1,  \label{eq:numericalBound-g1}\\
    g_2(p, m) &\leq \num{1.0430} < K_2,\\
    g_3(p, m) &\leq \num{1.2343} < K_3.
  \end{align}

  The random-walk quantities $\varepsilon_i^{(\nu)}$ ($i=1, 2$ and $\nu\in\nnInt$) monotonically decrease with respect to $d$ due to the estimate in Appendix~\ref{sec:numericalComputations-rwQuantities}.
  Recall that the upper bounds in Lemma~\ref{lem:basic-diagrams} depend only on $\{\varepsilon_i^{(\nu)}\}_{i=1, 2; \nu\in\nnInt}$ and $\{K_i\}_{i=1}^{3}$.
  Therefore, the inequalities $g_i(p, m) < K_i$, $i=1, 2, 3$ also hold in $d>9$ with the same numerical values \eqref{eq:numericalValuesForBootstrap} of $K_i$ as in $d=9$.
  This completes the proof of Proposition~\ref{prp:op-bootstrapArgument}.
\end{proof}

\begin{remark}\label{rem:computer-assisted-proof}
  One may wonder why the lace expansion analysis in this paper does not work in $d=8$.
  That is caused by the too large upper bounds in Lemmas~\ref{lem:precise-diagrammaticBounds} and \ref{lem:basic-diagrams} as well as \eqref{eq:improvedBound-g1}, \eqref{eq:improvedBound-g2} and \eqref{eq:improvedBound-g3}.
  To find the numerical values \eqref{eq:numericalValuesForBootstrap}, we get assistance from computers. 
  Specifically, we divide the intervals $\interval[open left]{\num{1.0}}{\num{1.1}}$, $\interval[open left]{\num{1.0}}{\num{1.1}}$ and $\interval[open left]{\num{1.0}}{\num{1.3}}$ for $K_1$, $K_2$ and $K_3$, respectively, into \num{100} and check the inequalities $g_i(p, m) < K_i$, $i=1, 2, 3$ for each dividing point in $d=9$.
  We do not find proper values of $\{K_i\}_{i=1}^{3}$ in $d=8$.
  There is a possibility that our selections of the parameters (intervals and the division number) are bad.
  However, it is essentially more important to improve the upper bounds than the computer-assisted part.
\end{remark}

\section*{Acknowledgments}

YK and SH are grateful to the National Center for Theoretical Sciences (NCTS)
for providing his support and hospitality during his visit to National Chengchi University (NCCU) in the period May--June 2018.
YK is also grateful for the financial support from NCCU during the visits from November 5--17, 2018,
and from December 17--27, 2019.
The work of LCC was supported by the Grant MOST 109-2115-M-004-010-MY3.
Finally, we would like to thank Akira Sakai and the anonymous referees for their valuable comments on an earlier version of the manuscript.

\appendix

\section{Random walk quantities on the BCC lattice}\label{sec:numericalComputations-rwQuantities}

In this section, we show the expressions of the upper bounds on \eqref{eq:def-rwQuantities}
to compute their numerical values.
We use Stirling's formula due to the property \eqref{eq:rwTransitonProbability} of the BCC lattice,
which helps us to obtain highly-precise estimates.
Recall \eqref{eq:rwTransitonProbability}.
Note that
\begin{equation}
  \label{eq:StirlingFormula-rwTransitonProbability}
  D^{\conv 2n}(o) = \qty(\binom{2n}{n} \frac{1}{2^{2n}})^d \leq \qty(\frac{1}{\sqrt{\pi n}})^d,
\end{equation}
for  $n\in\mathbb{N}$.
Now, fix a sufficiently large number $N$.
For $\nu\in\mathbb{N}$, \eqref{eq:StirlingFormula-rwTransitonProbability} and the bounds given by the integral test for convergence imply that
\begin{align*}
  \varepsilon_1^{(\nu)}
  &= \sum_{n=\nu}^{\infty} D^{\conv 2n}(o)\\
  &\leq \sum_{n=\nu}^{\nu + N - 1} D^{\conv 2n}(o) + \sum_{n=\nu + N}^{\infty} \frac{1}{(\pi n)^{d/2}}\\
  &\leq \sum_{n=0}^{N - 1} D^{\conv 2 (n + \nu)}(o) + \frac{1}{\pi^{d/2} (\nu + N)^{d/2}} + \frac{1}{\pi^{d/2}} \int_{\nu + N}^{\infty} s^{-d/2} \dd{s}\\
  &\leq \sum_{n=0}^{N - 1} D^{\conv 2 (n + \nu)}(o) + \frac{1}{\pi^{d/2} (\nu + N)^{d/2}} + \frac{2}{\pi^{d/2} \left(d - 2\right)} \left(\nu + N\right)^{(2 - d)/2}.
\end{align*}
Moreover,
\begin{align*}
  \varepsilon_2^{(\nu)}
  &= \sum_{n=\nu}^{\infty} \left(n - \nu + 1\right) D^{\conv 2n}(o)\\
  &\leq \sum_{n=\nu}^{\nu + N - 1} \left(n - \nu + 1\right) D^{\conv 2n}(o) + \sum_{n=\nu + N}^{\infty} \left(n - \nu + 1\right) \frac{1}{(\pi n)^{d/2}}\\
  &\leq \sum_{n=0}^{N - 1} \left(n + 1\right) D^{\conv 2 (n + \nu)}(o) + \frac{1}{\pi^{d/2} (\nu + N)^{d/2 - 1}}\\
  &\quad + \frac{1}{\pi^{d/2}} \int_{\nu + N}^{\infty} s^{1 - d/2} \dd{s} - \frac{\nu - 1}{\pi^{d/2}} \int_{\nu + N}^{\infty} s^{-d/2} \dd{s}\\
  &\leq \sum_{n=0}^{N - 1} \left(n + 1\right) D^{\conv 2 (n + \nu)}(o) + \frac{1}{\pi^{d/2} (\nu + N)^{(d - 2)/2}}\\
  &\quad + \frac{2}{\pi^{d/2} \left(d - 4\right)} \left(\nu + N\right)^{(4 - d)/2} - \frac{2 \left(\nu - 1\right)}{\pi^{d/2} \left(d - 2\right)} \left(\nu + N\right)^{(2 - d)/2}.
\end{align*}
When $N=500$, we obtain numerical values in Table~\ref{tbl:numericalValues-rw-bcc} by directly computing these expressions.

\begin{table}[bhpt]
  \caption{Numerical values of random-walk quantities on the BCC lattice.}
  \begin{subtable}{\linewidth}
    \centering
    \caption{The random-walk loops.}
    \begin{tabular}{|c|c|c|c|}
      \hline
      $d$ & $\varepsilon_1^{(1)}$ & $\varepsilon_1^{(2)}$ & $\varepsilon_1^{(3)}$\\
      \hline
      \num{3}  & \num{3.932159419e-01} & \num{2.682159060e-01} & \num{2.154814953e-01}\\
      \num{4}  & \num{1.186366900e-01} & \num{5.613668875e-02} & \num{3.636129693e-02}\\
      \num{5}  & \num{4.682555747e-02} & \num{1.557555743e-02} & \num{8.159785906e-03}\\
      $\vdots$ & $\vdots$ & $\vdots$ & $\vdots$\\
      \num{9}  & \num{2.143603149e-03} & \num{1.904781484e-04} & \num{4.382837046e-05}\\
      \num{10} & \num{1.044316187e-03} & \num{6.775368687e-05} & \num{1.276002017e-05}\\
      \num{11} & \num{5.126699321e-04} & \num{2.438868202e-05} & \num{3.766057003e-06}\\
      \num{12} & \num{2.529972928e-04} & \num{8.856667719e-06} & \num{1.123183338e-06}\\
      \hline
    \end{tabular}
  \end{subtable}
  \par\bigskip
  \begin{subtable}{\linewidth}
    \centering
    \caption{The random-walk bubbles}
    \begin{tabular}{|c|c|c|}
      \hline
      $d$ & $\varepsilon_2^{(1)}$ & $\varepsilon_2^{(2)}$ \\
      \hline
      \num{3}  & $\infty$ & $\infty$\\
      \num{4}  & $\infty$ & $\infty$\\
      \num{5}  & \num{1.125786856e-01} & \num{6.575312431e-02}\\
      $\vdots$ & $\vdots$ & $\vdots$\\
      \num{9}  & \num{2.410376015e-03} & \num{2.667728665e-04}\\
      \num{10} & \num{1.132062934e-03} & \num{8.774674704e-05}\\
      \num{11} & \num{5.425031298e-04} & \num{2.983319769e-05}\\
      \num{12} & \num{2.633784960e-04} & \num{1.038120322e-05}\\
      \hline
    \end{tabular}
  \end{subtable}
  \label{tbl:numericalValues-rw-bcc}
\end{table}

\section{Proof of Lemma~\ref{lem:ksp-bound}}\label{sec:proof-diagrammaticBounds}

In this section, we prove Lemma~\ref{lem:ksp-bound}.
It is the result of Lemma~\ref{lem:xsp-bound} below.
In this lemma, we first divide percolation events by cases, which are based on the fact of whether or not double connections are collapsed and where a path intersects a double connection.
Next, these observations yield a lot of disjoint events such as $\set{\vb*{x} \rightarrow \vb*{y}} \circ \set{\vb*{u} \rightarrow \vb*{v}}$, where $\circ$ denotes that the left and right events must occur disjointly.
Then, the BK~\cite{bk85} inequality provides the product of the two-point functions such as
\begin{equation*}
  \mathbb{P}_p\bigl(\set{\vb*{x} \rightarrow \vb*{y}} \circ \set{\vb*{u} \rightarrow \vb*{v}}\bigr)
  \leq \mathbb{P}_p(\vb*{x} \rightarrow \vb*{y}) \mathbb{P}_p(\vb*{u} \rightarrow \vb*{v})
  = \varphi_p(\vb*{y} - \vb*{x}) \varphi_p(\vb*{v} - \vb*{u}).
\end{equation*}
Such a product is represented by a diagram like \eqref{eq:diagrammaticRepresentation}.
Also, the method to obtain Lemma~\ref{lem:ksp-bound} from Lemma~\ref{lem:xsp-bound} is based on decomposing diagrams by the inequality
\begin{equation}
  \label{eq:decomposingDiagrams}
  \norm{f g}_1 \defeq \sum_{\vb*{x}\in\opSpaceTime} f(\vb*{x}) g(\vb*{x}) \leq \norm{f}_\infty \norm{g}_1,
\end{equation}
where $f$ and $g$ are functions depending on a sum of a product of the two-point functions.

\begin{lemma}\label{lem:xsp-bound}
  For $\vb*{x}\in\opSpaceTime$ and $N\geq 2$,
  \begin{align}
    &\begin{multlined}[b]
      \abs{\pi_p^{(0)}(\vb*{x}) - \pi_p^{(1)}(\vb*{x})}\\
      \leq
      \frac{1}{2}\times
      \begin{tikzpicture}[op diagram]
        \laceDraw[out=60, in=120, relative, first=2] (0,0) (0,2);
        \laceDraw[out=-60, in=-120, relative, first=2] (0,0) (0,2);
        \lacePutLabel[anchor=north] (0,0) {$\vb*{o}$};
        \lacePutLabel[anchor=south] (0,2) {$\vb*{x}$};
      \end{tikzpicture}
      +
      \begin{tikzpicture}[op diagram]
        \laceDraw[first=2] (0,0) (0,3);
        \draw[lace connectivity={first=2}] (0,0) to[bend left=80, distance=30] (0,2) node[vertex] {};
        \laceDraw[out=-80, in=-100, relative, last=2] (0,0) (0,3);
        \lacePutLabel[anchor=north] (0,0) {$\vb*{o}$};
        \lacePutLabel[anchor=south] (0,3) {$\vb*{x}$};
      \end{tikzpicture}
      + \frac{3}{2}\times
      \begin{tikzpicture}[op diagram]
        \laceDraw[out=60, in=120, relative, first=2] (0,0) (0,2);
        \laceDraw[out=-60, in=-120, relative, first=2] (0,0) (0,2);
        \laceDraw[first=1, endpoint-shape=line] (0,2) (0,3.5);
        \draw[lace connectivity={last=3}] (0,0) to[bend right=85, distance=60] (0,3.5);
        \lacePutLabel[anchor=north] (0,0) {$\vb*{o}$};
        \lacePutLabel[anchor=south] (0,3.5) {$\vb*{x}$};
      \end{tikzpicture}
      + 3\times
      \begin{tikzpicture}[op diagram]
        \laceDraw[out=60, in=120, relative, first=2] (0,0) (0,2);
        \laceDraw[out=-60, in=-120, relative, first=2] (0,0) (0,2);
        \laceDraw[out=60, in=120, relative, first=2, endpoint-shape=circle] (0,2) (0,4);
        \laceDraw[out=-60, in=-120, relative, first=2] (0,2) (0,4);
        \lacePutLabel[anchor=north] (0,0) {$\vb*{o}$};
        \lacePutLabel[anchor=south] (0,4) {$\vb*{x}$};
      \end{tikzpicture}
      + 2\times
      \begin{tikzpicture}[op diagram]
        \laceDraw[first=2] (0,0) (-1.3,2);
        \laceDraw[first=1, endpoint-shape=line] (0,0) (1.3,1.5);
        \laceDraw[first=1, endpoint-shape=line] (1.3,1.5) (-1.3,2);
        \laceDraw[first=1, endpoint-shape=line] (-1.3,2) (0,3.5);
        \laceDraw[last=2, endpoint-shape=line] (1.3,1.5) (0,3.5);
        \lacePutLabel[anchor=north] (0,0) {$\vb*{o}$};
        \lacePutLabel[anchor=south] (0,3.5) {$\vb*{x}$};
      \end{tikzpicture}
      +
      \begin{tikzpicture}[op diagram]
        \laceDraw[first=2] (0,0) (1.3,2);
        \laceDraw[first=1] (0,0) (-1.3,1.5);
        \laceDraw[first=1, endpoint-shape=line] (-1.3,1.5) (1.3,2);
        \laceDraw[first=1, endpoint-shape=line] (1.3,2) (0,3.5);
        \laceDraw[last=2] (-1.3,1.5) (0,3.5);
        \lacePutLabel[anchor=north] (0,0) {$\vb*{o}$};
        \lacePutLabel[anchor=south] (0,3.5) {$\vb*{x}$};
      \end{tikzpicture},
    \end{multlined}
    \label{eq:xsp-bound-first-zeroth}\\
    &\begin{multlined}[b]
      \pi_p^{(2)}(\vb*{x}) \leq
      \begin{tikzpicture}[op diagram]
        \laceDraw[out=60, in=120, relative, first=2] (0,0) (0,2);
        \laceDraw[out=-60, in=-120, relative, first=2] (0,0) (0,2);
        \laceDraw[first=1, endpoint-shape=line] (0,2) (0,3.5);
        \draw[lace connectivity={last=3}] (0,0) to[bend right=85, distance=60] (0,3.5);
        \lacePutLabel[anchor=north] (0,0) {$\vb*{o}$};
        \lacePutLabel[anchor=south] (0,3.5) {$\vb*{x}$};
      \end{tikzpicture}
      + 2\times
      \begin{tikzpicture}[op diagram]
        \laceDraw[out=60, in=120, relative, first=2] (0,0) (0,2);
        \laceDraw[out=-60, in=-120, relative, first=2] (0,0) (0,2);
        \laceDraw[out=60, in=120, relative, first=2, endpoint-shape=circle] (0,2) (0,4);
        \laceDraw[out=-60, in=-120, relative, first=2] (0,2) (0,4);
        \lacePutLabel[anchor=north] (0,0) {$\vb*{o}$};
        \lacePutLabel[anchor=south] (0,4) {$\vb*{x}$};
      \end{tikzpicture}
      +
      \begin{tikzpicture}[op diagram]
        \laceDraw[first=2] (0,0) (-1.3,2);
        \laceDraw[first=1, endpoint-shape=line] (0,0) (1.3,1.5);
        \laceDraw[first=1, endpoint-shape=line] (1.3,1.5) (-1.3,2);
        \laceDraw[first=1, endpoint-shape=line] (-1.3,2) (0,3.5);
        \laceDraw[last=2, endpoint-shape=line] (1.3,1.5) (0,3.5);
        \lacePutLabel[anchor=north] (0,0) {$\vb*{o}$};
        \lacePutLabel[anchor=south] (0,3.5) {$\vb*{x}$};
      \end{tikzpicture}
      +
      \begin{tikzpicture}[op diagram]
        \laceDraw[first=2] (0,0) (1.3,2);
        \laceDraw[first=1] (0,0) (-1.3,1.5);
        \laceDraw[first=1, endpoint-shape=line] (-1.3,1.5) (1.3,2);
        \laceDraw[first=1, endpoint-shape=line] (1.3,2) (0,3.5);
        \laceDraw[last=2] (-1.3,1.5) (0,3.5);
        \lacePutLabel[anchor=north] (0,0) {$\vb*{o}$};
        \lacePutLabel[anchor=south] (0,3.5) {$\vb*{x}$};
      \end{tikzpicture}\\
      + \frac{1}{2}\times
      \begin{tikzpicture}[op diagram]
        \laceDraw[out=60, in=120, relative, first=2] (1.3,0) (-1.3,1.5);
        \laceDraw[first=2] (1.3,0) (-1.3,1.5);
        \laceDraw[first=1, endpoint-shape=line] (-1.3,1.5) (-1.3,3.5);
        \laceDraw[last=3, endpoint-shape=circle] (1.3,0) (1.3,3);
        \draw (1.3,3) -- (-1.3,3.5);
        \laceDraw[first=1, endpoint-shape=line] (-1.3,3.5) (0,5);
        \laceDraw[first=2] (1.3,3) (0,5);
        \lacePutLabel[anchor=north] (1.3,0) {$\vb*{o}$};
        \lacePutLabel[anchor=south] (0,5) {$\vb*{x}$};
      \end{tikzpicture}
      +
      \begin{tikzpicture}[op diagram]
        \laceDraw[first=2] (0,0) (-1.3,2);
        \laceDraw[first=1] (0,0) (1.3,1.5);
        \draw (1.3,1.5) -- (-1.3,2);
        \laceDraw[first=1, endpoint-shape=line] (-1.3,2) (-1.3,4);
        \laceDraw[first=1, endpoint-shape=line] (1.3,1.5) (1.3,3.5);
        \draw (1.3,3.5) -- (-1.3,4);
        \laceDraw[first=1, endpoint-shape=line] (-1.3,4) (0,5.5);
        \laceDraw[first=2, endpoint-shape=line] (1.3,3.5) (0,5.5);
        \lacePutLabel[anchor=north] (0,0) {$\vb*{o}$};
        \lacePutLabel[anchor=south] (0,5.5) {$\vb*{x}$};
      \end{tikzpicture}
      + \frac{1}{2}\times
      \begin{tikzpicture}[op diagram]
        \laceDraw[out=60, in=120, relative, first=2] (1.3,0) (-1.3,1.5);
        \laceDraw[first=2] (1.3,0) (-1.3,1.5);
        \laceDraw[first=1, endpoint-shape=line] (-1.3,1.5) (1.3,3);
        \laceDraw[last=3, endpoint-shape=circle] (1.3,0) (-1.3,3.5);
        \draw (1.3,3) -- (-1.3,3.5);
        \laceDraw[first=1] (-1.3,3.5) (0,5);
        \laceDraw[first=2, endpoint-shape=line] (1.3,3) (0,5);
        \lacePutLabel[anchor=north] (1.3,0) {$\vb*{o}$};
        \lacePutLabel[anchor=south] (0,5) {$\vb*{x}$};
      \end{tikzpicture}
      +
      \begin{tikzpicture}[op diagram]
        \laceDraw[first=2] (0,0) (-1.3,2);
        \laceDraw[first=1] (0,0) (1.3,1.5);
        \draw (1.3,1.5) -- (-1.3,2);
        \laceDraw[first=1, endpoint-shape=line] (-1.3,2) (1.3,3.5);
        \laceDraw[first=1, endpoint-shape=line] (1.3,1.5) (-1.3,4);
        \draw (1.3,3.5) -- (-1.3,4);
        \laceDraw[first=1, endpoint-shape=line] (-1.3,4) (0,5.5);
        \laceDraw[first=2, endpoint-shape=line] (1.3,3.5) (0,5.5);
        \lacePutLabel[anchor=north] (0,0) {$\vb*{o}$};
        \lacePutLabel[anchor=south] (0,5.5) {$\vb*{x}$};
      \end{tikzpicture},
    \end{multlined}
    \label{eq:xsp-bound-second}\\
    &\begin{multlined}[b]
      \pi_p^{(N)}(\vb*{x}) \leq
      \sum_{\substack{\{\vb*{y}_i\}_{i=1}^{N}, \{\vb*{u}_i\}_{i=1}^{N}\\ (\forall i\colon \TimeOf(\vb*{y}_i) > \TimeOf(\vb*{u}_{i-1}))}}
      \left(
        2 \KroneckerDelta{\vb*{u}_1}{\vb*{o}} \KroneckerDelta{\vb*{y}_1}{\vb*{o}} \times
        \begin{tikzpicture}[op diagram]
          \laceDraw[last=1, endpoint-shape=line] (0,0) (-1.3,2) node[anchor=south] {$\vb*{u}_2$};
          \laceDraw[last=1, endpoint-shape=line] (0,0) (1.3,1.5) node[anchor=south west] {$\vb*{y}_2$};
          \lacePutLabel[anchor=north] (0,0) {$\vb*{o}$};
        \end{tikzpicture}
        + \frac{1}{2} \KroneckerDelta{\vb*{y}_1}{\vb*{o}} \times
        \begin{tikzpicture}[op diagram]
          \laceDraw[out=60, in=120, relative, first=2] (1.3,0) (-1.3,1.5);
          \laceDraw[first=2] (1.3,0) (-1.3,1.5) node[vertex] {} node[anchor=east] {$\vb*{u}_1$};
          \laceDraw[last=1, endpoint-shape=line] (-1.3,1.5) (-1.3,3.5) node[anchor=south] {$\vb*{u}_2$};
          \laceDraw[last=3, endpoint-shape=line] (1.3,0) (1.3,3) node[anchor=south] {$\vb*{y}_2$};
          \lacePutLabel[anchor=north] (1.3,0) {$\vb*{o}$};
        \end{tikzpicture}
        \right.\\ \left.
        +
        \begin{tikzpicture}[op diagram, label distance=-4pt]
          \laceDraw[last=1, endpoint-shape=line] (0,0) (-1.3,2) node[anchor=east] {$\vb*{u}_1$};
          \laceDraw[last=1, endpoint-shape=line] (0,0) (1.3,1.5) node[anchor=west] {$\vb*{y}_1$};
          \draw (1.3,1.5) -- (-1.3,2);
          \laceDraw[last=1, endpoint-shape=line] (-1.3,2) (-1.3,4) node[anchor=south] {$\vb*{u}_2$};
          \laceDraw[last=1, endpoint-shape=line] (1.3,1.5) (1.3,3.5) node[anchor=south] {$\vb*{y}_2$};
          \lacePutLabel[anchor=north] (0,0) {$\vb*{o}$};
        \end{tikzpicture}
        + \frac{1}{2} \KroneckerDelta{\vb*{y}_1}{\vb*{o}} \times
        \begin{tikzpicture}[op diagram]
          \laceDraw[out=60, in=120, relative, first=2] (1.3,0) (-1.3,1.5);
          \laceDraw[first=2] (1.3,0) (-1.3,1.5) node[vertex] {} node[anchor=east] {$\vb*{u}_1$};
          \laceDraw[last=1, endpoint-shape=line] (-1.3,1.5) (1.3,3) node[anchor=south west] {$\vb*{y}_2$};
          \laceDraw[last=3, endpoint-shape=line] (1.3,0) (-1.3,3.5) node[anchor=south] {$\vb*{u}_2$};
          \lacePutLabel[anchor=north] (1.3,0) {$\vb*{o}$};
        \end{tikzpicture}
        +
        \begin{tikzpicture}[op diagram, label distance=-4pt]
          \laceDraw[last=1, endpoint-shape=line] (0,0) (-1.3,2) node[anchor=east] {$\vb*{u}_1$};
          \laceDraw[last=1, endpoint-shape=line] (0,0) (1.3,1.5) node[anchor=west] {$\vb*{y}_1$};
          \draw (1.3,1.5) -- (-1.3,2);
          \laceDraw[last=1, endpoint-shape=line] (-1.3,2) (1.3,3.5) node[anchor=south west] {$\vb*{y}_2$};
          \laceDraw[last=1, endpoint-shape=line] (1.3,1.5) (-1.3,4) node[anchor=south] {$\vb*{u}_2$};
          \lacePutLabel[anchor=north] (0,0) {$\vb*{o}$};
        \end{tikzpicture}
      \right)\\
      \times \prod_{i=2}^{N-1}
      \left(
        \begin{tikzpicture}[op diagram]
          \draw (1.3,0) -- (-1.3,0.5);
          \laceDraw[last=1, endpoint-shape=line] (-1.3,0.5) (-1.3,2.5) node[anchor=south] {$\vb*{u}_{i+1}$};
          \laceDraw[last=1, endpoint-shape=line] (1.3,0) (1.3,2) node[anchor=south] {$\vb*{y}_{i+1}$};
          \lacePutLabel[anchor=north] (1.3,0) {$\vb*{y}_i$};
          \lacePutLabel[anchor=north] (-1.3,0.5) {$\vb*{u}_i$};
        \end{tikzpicture}
        +
        \begin{tikzpicture}[op diagram]
          \draw (1.3,0) -- (-1.3,0.5);
          \laceDraw[last=1, endpoint-shape=line] (-1.3,0.5) (1.3,2) node[anchor=south west] {$\vb*{y}_{i+1}$};
          \laceDraw[last=1, endpoint-shape=line] (1.3,0) (-1.3,2.5) node[anchor=south] {$\vb*{u}_{i+1}$};
          \lacePutLabel[anchor=north] (1.3,0) {$\vb*{y}_i$};
          \lacePutLabel[anchor=north] (-1.3,0.5) {$\vb*{u}_i$};
        \end{tikzpicture}
      \right)
      \times
      \begin{tikzpicture}[op diagram]
        \draw (1.3,0) -- (-1.3,0.5);
        \laceDraw[last=1, endpoint-shape=line] (0,2) (-1.3,0.5) node[anchor=north] {$\vb*{u}_N$};
        \laceDraw[last=1, endpoint-shape=line] (0,2) (1.3,0) node[anchor=north] {$\vb*{y}_N$};
        \lacePutLabel[anchor=south] (0,2) {$\vb*{x}$};
      \end{tikzpicture}.
    \end{multlined}
    \label{eq:xsp-bound-general}
  \end{align}
\end{lemma}

\begin{proof}[Proof of \eqref{eq:xsp-bound-first-zeroth} in Lemma~\ref{lem:xsp-bound}]
  This proof is inspired by \cite[Section~3.1]{hs05}.
  First, we rewrite the event in $\pi_p^{(0)}(\vb*{x})$.
  To do so, we introduce an ordering among bonds as follows.
  Let $\mathscr{B}((x, t)) = \mathopen{\{}((x, t), (y, t+1)) \in (\opSpaceTime)^2 \mid x - y \in \mathscr{N}^d\mathclose{\}}$
  for $(x, t) \in \opSpaceTime$, which is the set of directed bonds whose bottoms are $(x, t)$.
  We can order the elements in $\mathscr{B}((x, t))$ because it is a finite set.
  For a pair of bonds $\vb*{b}_1$ and $\vb*{b}_2$, we write $\vb*{b}_1 \prec \vb*{b}_2$
  if $\vb*{b}_1$ is smaller than $\vb*{b}_2$ in that ordering.
  Then, we obtain
  \begin{align}
    \pi_p^{(0)}(\vb*{x}) = {} \MoveEqLeft[0] \mathbb{P}_p(\vb*{o} \Rightarrow \vb*{x}) - \KroneckerDelta{\vb*{o}}{\vb*{x}} \notag\\
    = {} & \mathbb{P}_p\qty(\bigsqcup_{\vb*{b} \in \mathscr{B}(\vb*{o})} \qty\Big{\qty\big{\Set{\vb*{o} \rightarrow \vb*{x}} \circ \Set{\vb*{b} \rightarrow \vb*{x}}} \cap \Set{\forall \vb*{b}'\prec \vb*{b}, \vb*{b}' \not\rightarrow \vb*{x}}}) \notag\\
    = {} & \mathbb{P}_p\qty(\bigsqcup_{\vb*{b} \in \mathscr{B}(\vb*{o})} \qty\Big{\qty\big{\Set{\vb*{o} \rightarrow \vb*{x}} \circ \Set{\vb*{b} \rightarrow \vb*{x}}} \cap \Set{\forall \vb*{b}'\succ \vb*{b}, \vb*{b}' \not\rightarrow \vb*{x}}}) \notag\\
    = {} & \frac{1}{2} \sum_{\vb*{b} \in \mathscr{B}(\vb*{o})}
      \Bigl(
        \mathbb{P}_p\qty(\qty\big{\Set{\vb*{o} \rightarrow \vb*{x}} \circ \Set{\vb*{b} \rightarrow \vb*{x}}} \cap \Set{\forall \vb*{b}'\prec \vb*{b}, \vb*{b}' \not\rightarrow \vb*{x}}) \notag\\
        & + \mathbb{P}_p\qty(\qty\big{\Set{\vb*{o} \rightarrow \vb*{x}} \circ \Set{\vb*{b} \rightarrow \vb*{x}}} \cap \Set{\forall \vb*{b}'\succ \vb*{b}, \vb*{b}' \not\rightarrow \vb*{x}})
      \Bigr).
    \label{eq:rewriting-pi0}
  \end{align}

  Next, we rewrite the event in $\pi_p^{(1)}(\vb*{x})$.  By definition, we can easily see that
  \begin{equation}
    \label{eq:superset-cuttingLastSausage}
    E\qty(\vb*{b}, \vb*{x}; \tilde{\mathcal{C}}^{\vb*{b}}(\vb*{y})) \subset \Set{\vb*{y} \rightarrow \vb*{x}} \circ \Set{\vb*{b} \rightarrow \vb*{x}}.
  \end{equation}
  By splitting the event $\set{\vb*{o} \Rightarrow \underline{\vb*{b}}}$ into two events
  based on whether $\underline{\vb*{b}}$ equals $\vb*{o}$ or not,
  \begin{align}
    \pi_p^{(1)}(\vb*{x}) = {} \MoveEqLeft[0] \sum_{\vb*{b}} \mathbb{P}_p\qty(\Set{\vb*{o} \Rightarrow \underline{\vb*{b}}} \cap E(\vb*{b}, \vb*{x}; \tilde{\mathcal{C}}^{\vb*{b}}(\vb*{o}))) \notag\\
    = {} & \sum_{\vb*{b} \in \mathscr{B}(\vb*{o})} \mathbb{P}_p\qty(E(\vb*{b}, \vb*{x}; \tilde{\mathcal{C}}^{\vb*{b}}(\vb*{o}))) + \sum_{\vb*{b}} \mathbb{P}_p\qty(\Set{\vb*{o} \Rightarrow \underline{\vb*{b}} \neq \vb*{o}} \cap E(\vb*{b}, \vb*{x}; \tilde{\mathcal{C}}^{\vb*{b}}(\vb*{o}))) \notag\\
    = {} & \sum_{\vb*{b} \in \mathscr{B}(\vb*{o})}
      \qty\bigg(
        \mathbb{P}_p\qty\big(\Set{\vb*{o} \rightarrow \vb*{x}} \circ \Set{\vb*{b} \rightarrow \vb*{x}})
        - \mathbb{P}_p\qty(
          \qty\big{\Set{\vb*{o} \rightarrow \vb*{x}} \circ \Set{\vb*{b} \rightarrow \vb*{x}}}
          \setminus E(\vb*{b}, \vb*{x}; \tilde{\mathcal{C}}^{\vb*{b}}(\vb*{o}))
        )
      ) \notag\\
      & + \sum_{\vb*{b}} \mathbb{P}_p\qty(\Set{\vb*{o} \Rightarrow \underline{\vb*{b}} \neq \vb*{o}} \cap E(\vb*{b}, \vb*{x}; \tilde{\mathcal{C}}^{\vb*{b}}(\vb*{o}))).
    \label{eq:rewriting-pi1}
  \end{align}
  Since both
  \[
    \Set{\forall \vb*{b}'\prec \vb*{b}, \vb*{b}' \not\rightarrow \vb*{x}} \sqcup \Set{\exists \vb*{b}''\prec \vb*{b}, \vb*{b}'' \rightarrow \vb*{x}}
    \qand
    \Set{\forall \vb*{b}'\succ \vb*{b}, \vb*{b}' \not\rightarrow \vb*{x}} \sqcup \Set{\exists \vb*{b}''\succ \vb*{b}, \vb*{b}'' \rightarrow \vb*{x}}
  \]
  are the whole event, respectively,
  \begin{align}
    \MoveEqLeft \mathbb{P}_p\qty\big(\Set{\vb*{o} \rightarrow \vb*{x}} \circ \Set{\vb*{b} \rightarrow \vb*{x}}) \notag\\
    = {} & \underbrace{\mathbb{P}_p
      \qty(
        \qty\big{\Set{\vb*{o} \rightarrow \vb*{x}} \circ \Set{\vb*{b} \rightarrow \vb*{x}}}
        \cap
        \qty\big{
          \Set{\forall \vb*{b}'\prec \vb*{b}, \vb*{b}' \not\rightarrow \vb*{x}}
          \cap \Set{\forall \vb*{b}''\succ \vb*{b}, \vb*{b}'' \not\rightarrow \vb*{x}}
        }
      )}_{=0} \notag\\
      & + \mathbb{P}_p
      \qty(
        \qty\big{\Set{\vb*{o} \rightarrow \vb*{x}} \circ \Set{\vb*{b} \rightarrow \vb*{x}}}
        \cap
        \qty\big{
          \Set{\forall \vb*{b}'\prec \vb*{b}, \vb*{b}' \not\rightarrow \vb*{x}}
          \cap \Set{\exists \vb*{b}''\succ \vb*{b}, \vb*{b}'' \rightarrow \vb*{x}}
        }
      ) \notag\\
      & + \mathbb{P}_p
      \qty(
        \qty\big{\Set{\vb*{o} \rightarrow \vb*{x}} \circ \Set{\vb*{b} \rightarrow \vb*{x}}}
        \cap \qty\big{
          \Set{\forall \vb*{b}'\succ \vb*{b}, \vb*{b}' \not\rightarrow \vb*{x}}
          \cap \Set{\exists \vb*{b}''\prec \vb*{b}, \vb*{b}'' \rightarrow \vb*{x}}
        }
      ) \notag\\
      & + \mathbb{P}_p
      \qty(
        \qty\big{\Set{\vb*{o} \rightarrow \vb*{x}} \circ \Set{\vb*{b} \rightarrow \vb*{x}}}
        \cap \qty\big{
          \Set{\exists \vb*{b}'\prec \vb*{b}, \vb*{b}' \rightarrow \vb*{x}}
          \cap \Set{\exists \vb*{b}''\succ \vb*{b}, \vb*{b}'' \rightarrow \vb*{x}}
        }
      ) \notag\\
    = {} & \mathbb{P}_p
      \qty(
        \qty\big{\Set{\vb*{o} \rightarrow \vb*{x}} \circ \Set{\vb*{b} \rightarrow \vb*{x}}}
        \cap
        \Set{\forall \vb*{b}'\prec \vb*{b}, \vb*{b}' \not\rightarrow \vb*{x}}
      ) \notag\\
      & + \mathbb{P}_p
      \qty(
        \qty\big{\Set{\vb*{o} \rightarrow \vb*{x}} \circ \Set{\vb*{b} \rightarrow \vb*{x}}}
        \cap
        \Set{\forall \vb*{b}'\succ \vb*{b}, \vb*{b}' \not\rightarrow \vb*{x}}
      ) \notag\\
      & + \mathbb{P}_p
      \qty(
        \qty\big{\Set{\vb*{o} \rightarrow \vb*{x}} \circ \Set{\vb*{b} \rightarrow \vb*{x}}}
        \cap \qty\big{
          \Set{\exists \vb*{b}'\prec \vb*{b}, \vb*{b}' \rightarrow \vb*{x}}
          \cap \Set{\exists \vb*{b}''\succ \vb*{b}, \vb*{b}'' \rightarrow \vb*{x}}
        }
      ).
    \label{eq:pi1-leading-term}
  \end{align}
  Substituting \eqref{eq:pi1-leading-term} into \eqref{eq:rewriting-pi1} and subtracting \eqref{eq:rewriting-pi0}, we obtain
  \begin{multline}
    \pi_p^{(1)}(\vb*{x}) - \pi_p^{(0)}(\vb*{x})
    = \underbrace{\mathbb{P}_p(\vb*{o} \Rightarrow \vb*{x} \neq \vb*{o})}_{\mathrm{(a)}}\\
      + \underbrace{\sum_{\vb*{b} \in \mathscr{B}(\vb*{o})} \mathbb{P}_p
        \qty(
          \qty\big{\Set{\vb*{o} \rightarrow \vb*{x}} \circ \Set{\vb*{b} \rightarrow \vb*{x}}}
          \cap \qty\big{
            \Set{\exists \vb*{b}'\prec \vb*{b}, \vb*{b}' \rightarrow \vb*{x}}
            \cap \Set{\exists \vb*{b}''\succ \vb*{b}, \vb*{b}'' \rightarrow \vb*{x}}
          }
        )}_{\mathrm{(b)}}\\
      - \underbrace{\sum_{\vb*{b} \in \mathscr{B}(\vb*{o})} \mathbb{P}_p
        \qty(
          \qty\big{\Set{\vb*{o} \rightarrow \vb*{x}} \circ \Set{\vb*{b} \rightarrow \vb*{x}}}
          \setminus E(\vb*{b}, \vb*{x}; \tilde{\mathcal{C}}^{\vb*{b}}(\vb*{o}))
        )}_{\mathrm{(c)}}\\
      + \underbrace{\sum_{\vb*{b}} \mathbb{P}_p
        \qty(
          \Set{\vb*{o} \Rightarrow \underline{\vb*{b}} \neq \vb*{o}}
          \cap E(\vb*{b}, \vb*{x}; \tilde{\mathcal{C}}^{\vb*{b}}(\vb*{o}))
        )}_{\mathrm{(d)}}.
    \label{eq:pi1-minus-pi0}
  \end{multline}

  Finally, we lead \eqref{eq:pi1-minus-pi0} to the upper bound \eqref{eq:xsp-bound-first-zeroth}.
  In the following, we repeatedly use the trivial inequality
  \begin{equation}
    \varphi_p(\vb*{x}) \ind{\vb*{x} \neq \vb*{o}} \leq \qty(q_p \Conv \varphi_p)(\vb*{x})
    \label{eq:more-one-step}
  \end{equation}
  and the fact that, if there are two disjoint connections, then their lengths are at least two for oriented percolation.
  By Bool's and the BK inequalities, (a) in \eqref{eq:pi1-minus-pi0} is bounded above as
  \begin{align*}
    &\mathbb{P}_p(\vb*{o} \Rightarrow \vb*{x} \neq \vb*{o})\\
    &= \mathbb{P}_p\qty\Bigg(\bigcup_{\substack{\vb*{b}_1, \vb*{b}_2\in \mathscr{B}(\vb*{o})\\ (\vb*{b}_1 \prec \vb*{b}_2)}} \bigcup_{\substack{\vb*{b}_1'\in \mathscr{B}(\overline{\vb*{b}}_1),\\ \vb*{b}_2'\in \mathscr{B}(\overline{\vb*{b}}_2)}} \qty\big{\Set{\vb*{b}_1\text{ is occupied} \mathrel{\&} \vb*{b}_1'\rightarrow\vb*{x}} \circ \Set{\vb*{b}_2\text{ is occupied} \mathrel{\&} \vb*{b}_2'\rightarrow\vb*{x}}})\\
    &\leq \sum_{\substack{\vb*{b}_1, \vb*{b}_2\in \mathscr{B}(\vb*{o})\\ (\vb*{b}_1 \prec \vb*{b}_2)}} \sum_{\substack{\vb*{b}_1'\in \mathscr{B}(\overline{\vb*{b}}_1),\\ \vb*{b}_2'\in \mathscr{B}(\overline{\vb*{b}}_2)}} q_p(\vb*{b}_1) q_p(\vb*{b}_1') \varphi_p(\vb*{x} - \overline{\vb*{b}}_1') \cdot q_p(\vb*{b}_2) q_p(\vb*{b}_2') \varphi_p(\vb*{x} - \overline{\vb*{b}}_2')\\
    &\leq \frac{1}{2} \qty(q_p^{\Conv 2} \Conv \varphi_p)(\vb*{x})^2,
  \end{align*}
  which corresponds to the 1st term in \eqref{eq:xsp-bound-first-zeroth}.  The factor $1 / 2$ in the last line is due to ignoring the ordering.
  To bound (b) in \eqref{eq:pi1-minus-pi0}, we note that, for $\vb*{b}\in \mathscr{B}(\vb*{o})$,
  \begin{multline*}
    \qty\big{\Set{\vb*{o} \rightarrow \vb*{x}} \circ \Set{\vb*{b} \rightarrow \vb*{x}}}
    \cap \qty\big{
      \Set{\exists \vb*{b}'\prec \vb*{b}, \vb*{b}' \rightarrow \vb*{x}}
      \cap \Set{\exists \vb*{b}''\succ \vb*{b}, \vb*{b}'' \rightarrow \vb*{x}}
    }\\
    \subset
    \bigcup_{\vb*{b}', \vb*{b}'' \in \mathscr{B}(\vb*{o})} \bigcup_{\vb*{y}} \qty\big{
      \Set{\vb*{b} \rightarrow \vb*{y} \rightarrow \vb*{x}}
      \circ \Set{\vb*{b}' \rightarrow \vb*{x}}
      \circ \Set{\vb*{b}'' \rightarrow \vb*{y}}
    }.
  \end{multline*}
  By Bool's and the BK inequalities, we have
  \begin{multline*}
    \sum_{\vb*{b} \in \mathscr{B}(\vb*{o})} \mathbb{P}_p
      \qty(
        \qty\big{\Set{\vb*{o} \rightarrow \vb*{x}} \circ \Set{\vb*{b} \rightarrow \vb*{x}}}
        \cap \qty\big{
          \Set{\exists \vb*{b}'\prec \vb*{b}, \vb*{b}' \rightarrow \vb*{x}}
          \cap \Set{\exists \vb*{b}''\succ \vb*{b}, \vb*{b}'' \rightarrow \vb*{x}}
        }
      )\\
    \leq \sum_{\vb*{y}}
      \qty(q_p^{\Conv 2} \Conv \varphi_p)(\vb*{y})^2
      \qty(q_p^{\Conv 2} \Conv \varphi_p)(\vb*{x})
      \varphi_p(\vb*{x} - \vb*{y}),
  \end{multline*}
  which corresponds to the 2nd term in \eqref{eq:xsp-bound-first-zeroth}.
  To bound (c) in \eqref{eq:pi1-minus-pi0}, we note that, for $\vb*{b}\in \mathscr{B}(\vb*{o})$,
  \begin{align}
    \MoveEqLeft \qty\big{\Set{\vb*{o} \rightarrow \vb*{x}} \circ \Set{\vb*{b} \rightarrow \vb*{x}}} \setminus E(\vb*{b}, \vb*{x}; \tilde{\mathcal{C}}^{\vb*{b}}(\vb*{o})) \notag\\
    \subset {} & \bigcup_{\vb*{y}} \bigcup_{\vb*{b}'}
      \Bigl\{
        \qty\big{
          \Set{\vb*{o} \rightarrow \vb*{y} \rightarrow \vb*{x}}
          \circ \Set{\vb*{b} \rightarrow \vb*{b}' \rightarrow \vb*{x}}
          \circ \Set{\vb*{y} \rightarrow \underline{\vb*{b}}'}
        } \notag\\
        & \sqcup
        \qty\big{
          \Set{\vb*{o} \rightarrow \vb*{b}' \rightarrow \vb*{x}}
          \circ \Set{\vb*{b} \rightarrow \vb*{y} \rightarrow \underline{\vb*{b}}'}
          \circ \Set{\vb*{y} \rightarrow \vb*{x}}
        }
      \Bigr\} \notag\\
    = {} & \bigcup_{\vb*{b}'}
      \qty\big{
        \Set{\vb*{o} \rightarrow \vb*{x}}
        \circ \Set{\vb*{o} \rightarrow \underline{\vb*{b}}'}
        \circ \Set{\vb*{b} \rightarrow \vb*{b}' \rightarrow \vb*{x}}
      } \notag\\
      & \sqcup \bigcup_{\vb*{b}'}
      \qty\big{
        \Set{\vb*{o} \rightarrow \underline{\vb*{b}}' \rightarrow \vb*{x}}
        \circ \Set{\vb*{b} \rightarrow \vb*{b}' \rightarrow \vb*{x}}
      } \notag\\
      & \sqcup \bigcup_{\vb*{b}'} \bigcup_{\vb*{y} \neq \vb*{o}, \underline{\vb*{b}}'}
      \qty\big{
        \Set{\vb*{o} \rightarrow \vb*{y} \rightarrow \vb*{x}}
        \circ \Set{\vb*{b} \rightarrow \vb*{b}' \rightarrow \vb*{x}}
        \circ \Set{\vb*{y} \rightarrow \underline{\vb*{b}}'}
      } \notag\\
      & \sqcup \bigcup_{\vb*{b}'}
      \qty\big{
        \Set{\vb*{o} \rightarrow \vb*{b}' \rightarrow \vb*{x}}
        \circ \Set{\vb*{b} \rightarrow \underline{\vb*{b}}'}
        \circ \Set{\underline{\vb*{b}}' \rightarrow \vb*{x}}
      } \notag\\
      & \sqcup \bigcup_{\vb*{b}'} \bigcup_{\vb*{y} \neq \underline{\vb*{b}}'}
      \qty\big{
        \Set{\vb*{o} \rightarrow \vb*{b}' \rightarrow \vb*{x}}
        \circ \Set{\vb*{b} \rightarrow \vb*{y} \rightarrow \underline{\vb*{b}}'}
        \circ \Set{\vb*{y} \rightarrow \vb*{x}}
      }.
    \label{eq:errorEstimate-pi1}
  \end{align}
  By Bool's and the BK inequalities, we have
  \begin{align*}
    \MoveEqLeft \sum_{\vb*{b} \in \mathscr{B}(\vb*{o})}\mathbb{P}_p\qty(
      \qty\big{\Set{\vb*{o} \rightarrow \vb*{x}} \circ \Set{\vb*{b} \rightarrow \vb*{x}}}
      \setminus E(\vb*{b}, \vb*{x}; \tilde{\mathcal{C}}^{\vb*{b}}(\vb*{o}))
    )\\
    \leq {} &
      \sum_{\vb*{u}} \qty(q_p^{\Conv 3} \Conv \varphi_p)(\vb*{x}) \qty(q_p^{\Conv 2} \Conv \varphi_p)(\vb*{u})^2 \qty(q_p \Conv \varphi_p)(\vb*{x} - \vb*{u})\\
      & + \sum_{\vb*{u}} \qty(q_p^{\Conv 2} \Conv \varphi_p)(\vb*{u})^2 \qty(q_p^{\Conv 2} \Conv \varphi_p)(\vb*{x} - \vb*{u})^2\\
      & + \sum_{\vb*{u}, \vb*{y}} \qty(q_p^{\Conv 2} \Conv \varphi_p)(\vb*{u}) \qty(q_p \Conv \varphi_p)(\vb*{y}) \qty(q_p \Conv \varphi_p)(\vb*{u} - \vb*{y}) \qty(q_p \Conv \varphi_p)(\vb*{x} - \vb*{u}) \qty(q_p^{\Conv 2} \Conv \varphi_p)(\vb*{x} - \vb*{y})\\
      & + \sum_{\vb*{u}} \qty(q_p^{\Conv 2} \Conv \varphi_p)(\vb*{u})^2 \qty(q_p^{\Conv 2} \Conv \varphi_p)(\vb*{x} - \vb*{u})^2\\
      & + \sum_{\vb*{u}, \vb*{y}} \qty(q_p^{\Conv 2} \Conv \varphi_p)(\vb*{u}) \qty(q_p \Conv \varphi_p)(\vb*{y}) \qty(q_p \Conv \varphi_p)(\vb*{u} - \vb*{y}) \qty(q_p^{\Conv 2} \Conv \varphi_p)(\vb*{x} - \vb*{y}) \qty(q_p \Conv \varphi_p)(\vb*{x} - \vb*{u}).
  \end{align*}
  Each term in the above upper bound corresponds to contributions of the 3rd term, the 4th term, the 5th term, the 4th term and the 6th term in \eqref{eq:xsp-bound-first-zeroth}, respectively.
  To bound (d) in \eqref{eq:pi1-minus-pi0}, we apply \eqref{eq:superset-cuttingLastSausage}
  and split the below event into two events based on where the branching point is assigned:
  \begin{multline}
    \label{eq:splitting-doubleConnection}
    \Set{\vb*{o} \Rightarrow \vb*{u} \neq \vb*{o}} \cap \Set{\vb*{o} \rightarrow \vb*{x}}
    \subset \qty\big{\Set{\vb*{o} \Rightarrow \vb*{u} \neq \vb*{o}} \circ \Set{\vb*{o} \rightarrow \vb*{x}}}\\
      \cup \bigcup_{\substack{\vb*{y} \neq \vb*{o}\\ (\TimeOf(\vb*{y}) \leq \TimeOf(\vb*{u}))}} \qty\big{\Set{\vb*{o} \rightarrow \vb*{u}} \circ \Set{\vb*{o} \rightarrow \vb*{y} \rightarrow \vb*{u}} \circ \Set{\vb*{y} \rightarrow \vb*{x}}}.
  \end{multline}
  Then, we obtain
  \begin{align}
    \MoveEqLeft \sum_{\vb*{b}} \mathbb{P}_p\qty(\Set{\vb*{o} \Rightarrow \underline{\vb*{b}} \neq \vb*{o}} \cap E(\vb*{b}, \vb*{x}; \tilde{\mathcal{C}}^{\vb*{b}}(\vb*{o}))) \notag\\
    \overset{\eqref{eq:superset-cuttingLastSausage}}{\leq} {} & \sum_{\vb*{b}} \mathbb{P}_p
      \qty(
        \Set{\vb*{o} \Rightarrow \underline{\vb*{b}} \neq \vb*{o}}
        \cap \qty\big{
          \Set{\vb*{o} \rightarrow \vb*{x}}
          \circ \Set{\vb*{b} \rightarrow \vb*{x}}
        }
      ) \notag\\
    \mathmakebox[\widthof{(00)}]{=} {} & \sum_{\vb*{b}} \mathbb{P}_p
      \qty(
        \qty\big{
          \Set{\vb*{o} \Rightarrow \underline{\vb*{b}} \neq \vb*{o}}
          \cap \Set{\vb*{o} \rightarrow \vb*{x}}
        }
        \circ \Set{\vb*{b} \rightarrow \vb*{x}}
      ) \notag\\
    \overset{\eqref{eq:splitting-doubleConnection}}{\leq} {} & \sum_{\vb*{b}}
      \biggl(
        \mathbb{P}_p
          \qty\big(
            \Set{\vb*{o} \Rightarrow \underline{\vb*{b}} \neq \vb*{o}}
            \circ \Set{\vb*{o} \rightarrow \vb*{x}}
            \circ \Set{\vb*{b} \rightarrow \vb*{x}}
          ) \notag\\
        & + \sum_{\substack{\vb*{y} \neq \vb*{o}\\ (\TimeOf(\vb*{y}) \leq \TimeOf(\underline{\vb*{b}}))}} \mathbb{P}_p
          \qty\big(
            \Set{\vb*{o} \rightarrow \underline{\vb*{b}}}
            \circ \underbrace{\Set{\vb*{o} \rightarrow \vb*{y} \rightarrow \underline{\vb*{b}}}}_{
              =\set{\vb*{o} \rightarrow \vb*{y} = \underline{\vb*{b}}}
              \sqcup \set{\vb*{o} \rightarrow \vb*{y} \rightarrow \underline{\vb*{b}} \neq \vb*{y}}
            }
            \circ \Set{\vb*{y} \rightarrow \vb*{x}}
            \circ \Set{\vb*{b} \rightarrow \vb*{x}}
          )
      \biggr).
    \label{eq:pre-xsp-bound-first}
  \end{align}
  Applying the BK~\cite{bk85} inequality and the same method for (a) to the right-most in the above,
  and paying attention to the disjointness of the connections, we arrive in
  \begin{align*}
    \MoveEqLeft \sum_{\vb*{b}} \mathbb{P}_p\qty(\Set{\vb*{o} \Rightarrow \underline{\vb*{b}} \neq \vb*{o}} \cap E(\vb*{b}, \vb*{x}; \tilde{\mathcal{C}}^{\vb*{b}}(\vb*{o})))\\
    \leq {} &
      \frac{1}{2} \sum_{\vb*{u}} \qty(q_p^{\Conv 2} \Conv \varphi_p)(\vb*{u})^2 \qty(q_p \Conv \varphi_p)(\vb*{x} - \vb*{u}) \qty(q_p^{\Conv 3} \Conv \varphi_p)(\vb*{x})\\
      & + \sum_{\vb*{u}} \qty(q_p^{\Conv 2} \Conv \varphi_p)(\vb*{u})^2 \qty(q_p^{\Conv 2} \Conv \varphi_p)(\vb*{x} - \vb*{u})^2\\
      & + \sum_{\vb*{u}, \vb*{y}} \qty(q_p^{\Conv 2} \Conv \varphi_p)(\vb*{u}) \qty(q_p \Conv \varphi_p)(\vb*{y}) \qty(q_p \Conv \varphi_p)(\vb*{u} - \vb*{y}) \qty(q_p \Conv \varphi_p)(\vb*{x} - \vb*{u}) \qty(q_p^{\Conv 2} \Conv \varphi_p)(\vb*{x} - \vb*{y}).
  \end{align*}
  Each term in the above upper bound corresponds to contributions of the 3rd term, the 4th term and the 5th term in \eqref{eq:xsp-bound-first-zeroth}, respectively.
  Combining the upper bounds on (a)--(d) completes the proof of \eqref{eq:xsp-bound-first-zeroth}.
\end{proof}

\begin{proof}[Proof of \eqref{eq:xsp-bound-second} in Lemma~\ref{lem:xsp-bound}]
  It is not hard to prove the upper bound by using \eqref{eq:superset-cuttingLastSausage}, \eqref{eq:errorEstimate-pi1}
  and \eqref{eq:superset-lastSausage-branch}, so that we omit it.
\end{proof}

\begin{proof}[Proof of \eqref{eq:xsp-bound-general} in Lemma~\ref{lem:xsp-bound}]
  By definition,
  \begin{multline}
    E(\vb*{b}, \vb*{u}; \tilde{\mathcal{C}}^{\vb*{b}}(\vb*{v})) \cap \Set{\overline{\vb*{b}} \rightarrow \vb*{x}}
    \subset \bigcup_{\substack{\vb*{y}\\ (\TimeOf(\underline{\vb*{b}}) < \TimeOf(\vb*{y}) \leq \TimeOf(\vb*{u}))}} \Bigl\{
      \qty\big{\Set{\vb*{b} \rightarrow \vb*{u}} \circ \Set{\vb*{v} \rightarrow \vb*{y} \rightarrow \vb*{u}} \circ \Set{\vb*{y} \rightarrow \vb*{x}}}\\
      \cup \qty\big{\Set{\vb*{b} \rightarrow \vb*{y} \rightarrow \vb*{u}} \circ \Set{\vb*{v} \rightarrow \vb*{u}} \circ \Set{\vb*{y} \rightarrow \vb*{x}}}
    \Bigr\}.
    \label{eq:superset-lastSausage-branch}
  \end{multline}
  Paying attention to the disjointness of connections and the magnitude relationship between times, we obtain
  \begin{align*}
    \MoveEqLeft \pi_p^{(N)}(\vb*{x})
    = {} \sum_{\va*{b}_N} \mathbb{P}_p\qty(\tilde{E}_{\va*{b}_N}^{(N)}(\vb*{x}))
    = \sum_{\va*{b}_N} \mathbb{P}_p\qty(\tilde{E}_{\va*{b}_{N-1}}^{(N-1)}(\underline{\vb*{b}}_N) \cap E(\vb*{b}_N, \vb*{x}; \tilde{\mathcal{C}}^{\vb*{b}_N}(\overline{\vb*{b}}_{N-1})))\\
    \mathmakebox[\widthof{\eqref{eq:superset-lastSausage-branch}}]{\overset{\eqref{eq:superset-cuttingLastSausage}}{\leq}} {} & \sum_{\va*{b}_N} \mathbb{P}_p\qty(\tilde{E}_{\va*{b}_{N-2}}^{(N-2)}(\underline{\vb*{b}}_{N-1}) \cap E(\vb*{b}_{N-1}, \underline{\vb*{b}}_N; \tilde{\mathcal{C}}^{\vb*{b}_{N-1}}(\overline{\vb*{b}}_{N-2})) \cap \Set{\vb*{b}_{N-1} \rightarrow \vb*{x}} \circ \Set{\vb*{b}_N \rightarrow \vb*{x}})\\
    \overset{\eqref{eq:superset-lastSausage-branch}}{\leq} {} & \sum_{\va*{b}_{N-2}}\sum_{\vb*{v}_{N-1}, \vb*{v}_N}\sum_{\substack{\vb*{y}_N, \vb*{u}_{N-1}, \vb*{u}_N\\ (\TimeOf(\vb*{u}_{N-1}) < \TimeOf(\vb*{y}_N) \leq \TimeOf(\vb*{u}_N))}} \biggl(
      \mathbb{P}_p\qty(\tilde{E}_{\va*{b}_{N-2}}^{(N-2)}(\vb*{u}_{N-1}) \cap \Set{\overline{\vb*{b}}_{N-2} \rightarrow \vb*{y}_N})\\
      & \quad \times \mathbb{P}_p\qty\big(\Set{(\vb*{u}_{N-1}, \vb*{v}_{N-1}) \rightarrow \vb*{u}_N} \circ \Set{\vb*{y}_N \rightarrow \vb*{u}_N} \circ \Set{\vb*{y}_N \rightarrow \vb*{x}} \circ \Set{(\vb*{u}_N, \vb*{v}_N) \rightarrow \vb*{x}})\\
      & + \mathbb{P}_p\qty(\tilde{E}_{\va*{b}_{N-2}}^{(N-2)}(\vb*{u}_{N-1}) \cap \Set{\overline{\vb*{b}}_{N-2} \rightarrow \vb*{u}_N})\\
      & \quad \times \mathbb{P}_p\qty\big(\Set{(\vb*{u}_{N-1}, \vb*{v}_{N-1}) \rightarrow \vb*{y}_N \rightarrow \vb*{u}_N} \circ \Set{\vb*{y}_N \rightarrow \vb*{x}} \circ \Set{(\vb*{u}_N, \vb*{v}_N) \rightarrow \vb*{x}})
    \biggr)\\
    \mathmakebox[\widthof{\eqref{eq:superset-lastSausage-branch}}]{\leq} {} & \sum_{\va*{b}_{N-2}}\sum_{\substack{\vb*{y}_N, \vb*{u}_{N-1}, \vb*{u}_N\\ (\TimeOf(\vb*{u}_{N-1}) < \TimeOf(\vb*{y}_N) \leq \TimeOf(\vb*{u}_N))}}
      \biggl(
        \mathbb{P}_p\qty(\tilde{E}_{\va*{b}_{N-2}}^{(N-2)}(\vb*{u}_{N-1}) \cap \Set{\overline{\vb*{b}}_{N-2} \rightarrow \vb*{y}_N})\\
        & \quad \times
        \qty(q_p \Conv \varphi_p)(\vb*{u}_N - \vb*{u}_{N-1})\\
        & + \mathbb{P}_p\qty(\tilde{E}_{\va*{b}_{N-2}}^{(N-2)}(\vb*{u}_{N-1}) \cap \Set{\overline{\vb*{b}}_{N-2} \rightarrow \vb*{u}_N})
        \qty(q_p \Conv \varphi_p)(\vb*{y}_N - \vb*{u}_{N-1})
      \biggr)\\
      & \times \underbrace{\tcboxmath[
        enhanced,
        frame hidden,
        colback=blue!10,
        size=minimal
      ]{
        \varphi_p(\vb*{u}_N - \vb*{y}_N) \qty(q_p \Conv \varphi_p)(\vb*{x} - \vb*{y}_N) \qty(q_p \Conv \varphi_p)(\vb*{x} - \vb*{u}_N)
      }}_{
        \eqdef \Xi(\vb*{y}_N, \vb*{u}_N; \vb*{x}, \vb*{x}) / 2
      }\\
    \overset{\eqref{eq:superset-lastSausage-branch}}{\leq} {} & \sum_{\va*{b}_{N-3}}\sum_{\substack{\vb*{y}_{N-1}, \vb*{y}_N, \vb*{u}_{N-2}, \vb*{u}_{N-1}, \vb*{u}_N\\ (\forall i\colon \TimeOf(\vb*{u}_{i-1}) < \TimeOf(\vb*{y}_i) \leq \TimeOf(\vb*{u}_i))}}\\
      & \biggl(
        \mathbb{P}_p\qty(\tilde{E}_{\va*{b}_{N-3}}^{(N-3)}(\vb*{u}_{N-2}) \cap \Set*{\overline{\vb*{b}}_{N-3} \rightarrow \vb*{y}_{N-1}})\\
        & \quad \times
        \qty(q_p \Conv \varphi_p)(\vb*{u}_{N-1} - \vb*{u}_{N-2})\\
        & + \mathbb{P}_p\qty(\tilde{E}_{\va*{b}_{N-3}}^{(N-3)}(\vb*{u}_{N-2}) \cap \Set{\overline{\vb*{b}}_{N-3} \rightarrow \vb*{u}_{N-1}})
        \qty(q_p \Conv \varphi_p)(\vb*{y}_{N-1} - \vb*{u}_{N-2})
      \biggr)\\
      & \times \underbrace{\tcboxmath[
        enhanced,
        frame hidden,
        colback=blue!10,
        size=minimal
      ]{
        \begin{aligned}[t]
          & \bigl(\varphi_p(\vb*{u}_{N-1} - \vb*{y}_{N-1}) \qty(q_p \Conv \varphi_p)(\vb*{y}_N - \vb*{y}_{N-1}) \qty(q_p \Conv \varphi_p)(\vb*{u}_N - \vb*{u}_{N-1})\\[\jot]
          & \quad + \varphi_p(\vb*{u}_{N-1} - \vb*{y}_{N-1}) \qty(q_p \Conv \varphi_p)(\vb*{u}_N - \vb*{y}_{N-1}) \qty(q_p \Conv \varphi_p)(\vb*{y}_N - \vb*{u}_{N-1})\bigr)
        \end{aligned}
      }}_{
        \scriptstyle\eqdef \Xi(\vb*{y}_{N-1}, \vb*{u}_{N-1}; \vb*{y}_N, \vb*{u}_N)
      }\\
      & \times \frac{1}{2} \Xi(\vb*{y}_N, \vb*{u}_N; \vb*{x}, \vb*{x}).
  \end{align*}
  By applying \eqref{eq:superset-lastSausage-branch} to the above repeatedly,
  \begin{multline*}
    \pi_p^{(N)}(\vb*{x})
    \leq \sum_{\vb*{v}_1}\sum_{\substack{\{\vb*{y}_i\}_{i=2}^{N}, \{\vb*{u}_i\}_{i=1}^{N}\\ (\forall i\colon \TimeOf(\vb*{u}_{i-1}) < \TimeOf(\vb*{y}_i) \leq \TimeOf(\vb*{u}_i))}} \Bigl(
      \mathbb{P}_p\qty\big(\Set{\vb*{o} \Rightarrow \vb*{u}_1} \cap \Set{(\vb*{u}_1, \vb*{v}_1) \rightarrow \vb*{u}_2} \circ \Set{\vb*{o} \rightarrow \vb*{y}_2})\\
      + \mathbb{P}_p\qty\big(\Set{\vb*{o} \Rightarrow \vb*{u}_1} \cap \Set{(\vb*{u}_1, \vb*{v}_1) \rightarrow \vb*{y}_2} \circ \Set{\vb*{o} \rightarrow \vb*{u}_2})
    \Bigr)\\
    \times \qty(\prod_{j=2}^{N-1} \Xi(\vb*{y}_j, \vb*{u}_j; \vb*{y}_{j+1}, \vb*{u}_{j+1})) \frac{1}{2} \Xi(\vb*{y}_N, \vb*{u}_N; \vb*{x}, \vb*{x}).
  \end{multline*}
  Finally, using Bool's and the BK~\cite{bk85} inequality, and applying a similar method to \eqref{eq:xsp-bound-first-zeroth}, we arrive in \eqref{eq:xsp-bound-general}.
\end{proof}

Multiplying the diagrammatic bounds in Lemma~\ref{lem:xsp-bound} by the factors $m^t$, $t$ or $1 - \cos k\cdot x$, and taking the sum of them, we obtain Lemma~\ref{lem:ksp-bound}.
The upper bounds \eqref{eq:sum-diagrammaticBounds-mcos-01}--\eqref{eq:sum-diagrammaticBounds-m-N} also require a telescopic inequality for the cosine function.
Since its proof is quite the same as the literature, we omit it.

\begin{lemma}[{\cite[Lemma~2.13]{fh17p}} or {\cite[Lemma~7.3]{hh17}}]\label{lem:split-of-cosine}
  Let $J\geq 1$ and $t_j\in\mathbb{R}$ for $j=1, \dots, J$.
  Then,
  \begin{equation}
    \label{eq:telescope}
    0 \leq 1 - \cos\sum_{j=1}^{J}t_j \leq J \sum_{j=1}^{J}\qty(1 - \cos t_j).
  \end{equation}
\end{lemma}

\begin{proof}[Sketch proof of Lemma~\ref{lem:ksp-bound}]
  We only deal with three examples of the bounds on the sum of a diagram multiplied by the factors $m^t$, $t$ or $1 - \cos k\cdot x$
  because one can easily calculate the other bounds on the analogy of such examples.
  The following proof is almost identical to the proof of \cite[Lemma~5.3]{Handa-Kamijima-Sakai}.

  First, we consider the bounds \eqref{eq:sum-diagrammaticBounds-m-01}--\eqref{eq:sum-diagrammaticBounds-m-N}.
  By the translation-invariance, for example,
  \begin{align*}
    \sum_{\vb*{x}}
    \begin{tikzpicture}[op diagram]
      \laceDraw[first=2] (0,0) (1.3,2);
      \laceDraw[first=1] (0,0) (-1.3,1.5);
      \laceDraw[first=1, endpoint-shape=line] (-1.3,1.5) (1.3,2);
      \laceDraw[first=1, endpoint-shape=line] (1.3,2) (0,3.5);
      \laceDraw[last=2] (-1.3,1.5) (0,3.5);
      \lacePutLabel[anchor=north] (0,0) {$\vb*{o}$};
      \lacePutLabel[anchor=south] (0,3.5) {$\vb*{x}$};
    \end{tikzpicture}
    m^{\TimeOf(\vb*{x})}
    &=
    \sum_{\vb*{x}, \vb*{y}, \vb*{w}}
    \begin{tikzpicture}[op diagram]
      \laceDraw[first=2] (0,0) (1.3,2);
      \laceDraw[first=1] (0,0) (-1.3,1.5);
      \laceDraw[first=1, endpoint-shape=circle] (-1.3,1.5) (1.3,2);
      \laceDraw[first=1, endpoint-shape=circle] (1.3,2) (0,3.5);
      \laceDraw[last=2] (-1.3,1.5) (0,3.5);
      \lacePutLabel[anchor=north] (0,0) {$\vb*{o}$};
      \lacePutLabel[anchor=south] (0,3.5) {$\vb*{x}$};
      \lacePutLabel[anchor=west] (1.3,2) {$\vb*{y}$};
      \lacePutLabel[anchor=east] (-1.3,1.5) {$\vb*{w}$};
    \end{tikzpicture}
    m^{\TimeOf(\vb*{y})} m^{\TimeOf(\vb*{x}) - \TimeOf(\vb*{y})}\\
    &\leq
    \sum_{\vb*{w}} \qty(
      \sup_{\vb*{y}} \sum_{\vb*{x}}
      \begin{tikzpicture}[op diagram]
        \laceDraw[first=1, endpoint-shape=circle] (1.3,2) (0,3.5);
        \laceDraw[last=2] (-1.3,1.5) (0,3.5);
        \lacePutLabel[anchor=south] (0,3.5) {$\vb*{x}$};
        \lacePutLabel[anchor=west] (1.3,2) {$\vb*{y}$};
        \lacePutLabel[anchor=east] (-1.3,1.5) {$\vb*{w}$};
      \end{tikzpicture}
      m^{\TimeOf(\vb*{x}) - \TimeOf(\vb*{y})}
    )
    \qty(
      \sum_{\vb*{y}}
      \begin{tikzpicture}[op diagram]
        \laceDraw[first=2] (0,0) (1.3,2);
        \laceDraw[first=1] (0,0) (-1.3,1.5);
        \laceDraw[first=1, endpoint-shape=circle] (-1.3,1.5) (1.3,2);
        \lacePutLabel[anchor=north] (0,0) {$\vb*{o}$};
        \lacePutLabel[anchor=west] (1.3,2) {$\vb*{y}$};
        \lacePutLabel[anchor=east] (-1.3,1.5) {$\vb*{w}$};
      \end{tikzpicture}
      m^{\TimeOf(\vb*{y})}
    )\\
    &=
    \qty(
      \sup_{\vb*{y}} \sum_{\vb*{x}}
      \begin{tikzpicture}[op diagram]
        \laceDraw[first=1, endpoint-shape=circle] (1.3,2) (0,3.5);
        \laceDraw[last=2] (-1.3,1.5) (0,3.5);
        \lacePutLabel[anchor=south] (0,3.5) {$\vb*{x}$};
        \lacePutLabel[anchor=west] (1.3,2) {$\vb*{y}$};
        \lacePutLabel[anchor=east] (-1.3,1.5) {$\vb*{o}$};
      \end{tikzpicture}
      m^{\TimeOf(\vb*{x}) - \TimeOf(\vb*{y})}
    )
    \qty(
      \sum_{\vb*{y}}
      \begin{tikzpicture}[op diagram]
        \laceDraw[first=2] (0,0) (1.3,2);
        \laceDraw[first=1] (0,0) (-1.3,1.5);
        \laceDraw[first=1, endpoint-shape=line] (-1.3,1.5) (1.3,2);
        \lacePutLabel[anchor=north] (0,0) {$\vb*{o}$};
        \lacePutLabel[anchor=west] (1.3,2) {$\vb*{y}$};
      \end{tikzpicture}
      m^{\TimeOf(\vb*{y})}
    )\\
    &= \qty(
      \sup_{\vb*{y}} \sum_{\vb*{x}}
      \qty(q_p^{\Conv 2} \Conv \varphi_p)(\vb*{x})
      \qty(q_p \Conv \varphi_p)(\vb*{x} - \vb*{y})
      m^{\TimeOf(\vb*{x}) - \TimeOf(\vb*{y})}
    )\\
    &\quad \times \qty(
      \sum_{\vb*{y}}
      \qty(q_p^{\Conv 2} \Conv \varphi_p^{\Conv 2})(\vb*{y})
      \qty(q_p^{\Conv 2} \Conv \varphi_p)(\vb*{y})
      m^{\TimeOf(\vb*{y})}
    )\\
    &\leq B_{p, m}^{(2, 1)} T_{p, m}^{(2, 2)},
  \end{align*}
  which corresponds to the last term in the right hand side in \eqref{eq:sum-diagrammaticBounds-m-01}.
  In the second equality, we have used the translation invariance.

  Next, we consider the bounds \eqref{eq:sum-diagrammaticBounds-mt-01}--\eqref{eq:sum-diagrammaticBounds-mt-N}.
  Note that
  \begin{gather*}
    \qty(q_p \Conv \varphi_p)(x, t) t \leq \qty(q_p \Conv \varphi_p^{\Conv 2})(x, t),\\
    \qty(q_p^{\Conv 2} \Conv \varphi_p)(x, t) t \leq \qty(q_p^{\Conv 2} \Conv \varphi_p^{\Conv 2})(x, t) + \qty(q_p \Conv \varphi_p^{\Conv 2})(x, t).
  \end{gather*}
  By using the above inequalities, for example,
  \begin{align*}
    &\sum_{\vb*{x}}
    \begin{tikzpicture}[op diagram]
      \laceDraw[first=2] (0,0) (-1.3,2);
      \laceDraw[first=1] (0,0) (1.3,1.5);
      \draw (1.3,1.5) -- (-1.3,2);
      \laceDraw[first=1, endpoint-shape=line] (-1.3,2) (-1.3,4);
      \laceDraw[first=1, endpoint-shape=line] (1.3,1.5) (1.3,3.5);
      \draw (1.3,3.5) -- (-1.3,4);
      \laceDraw[first=1, endpoint-shape=line] (-1.3,4) (0,5.5);
      \laceDraw[first=2, endpoint-shape=line] (1.3,3.5) (0,5.5);
      \lacePutLabel[anchor=north] (0,0) {$\vb*{o}$};
      \lacePutLabel[anchor=south] (0,5.5) {$\vb*{x}$};
    \end{tikzpicture}
    m^{\TimeOf(\vb*{x})} \TimeOf(\vb*{x})
    =
    \sum_{\vb*{x}, \vb*{y}, \vb*{w}}
    \begin{tikzpicture}[op diagram]
      \laceDraw[first=2] (0,0) (-1.3,2);
      \laceDraw[first=1] (0,0) (1.3,1.5);
      \draw (1.3,1.5) -- (-1.3,2);
      \laceDraw[first=1, endpoint-shape=line] (-1.3,2) (-1.3,4);
      \laceDraw[first=1, endpoint-shape=line] (1.3,1.5) (1.3,3.5);
      \draw (1.3,3.5) -- (-1.3,4);
      \laceDraw[first=1, endpoint-shape=line] (-1.3,4) (0,5.5);
      \laceDraw[first=2, endpoint-shape=line] (1.3,3.5) (0,5.5);
      \lacePutLabel[anchor=north] (0,0) {$\vb*{o}$};
      \lacePutLabel[anchor=south] (0,5.5) {$\vb*{x}$};
      \lacePutLabel[anchor=east] (-1.3,2) {$\vb*{y}$};
      \lacePutLabel[anchor=east] (-1.3,4) {$\vb*{w}$};
    \end{tikzpicture}
    m^{\TimeOf(\vb*{x})} \qty\Big(\qty\big(\TimeOf(\vb*{x}) - \TimeOf(\vb*{w})) + \qty\big(\TimeOf(\vb*{w}) - \TimeOf(\vb*{y})) + \vb*{y})\\
    &\leq
    \sum_{\vb*{x}, \vb*{y}, \vb*{w}} \qty(
      \begin{tikzpicture}[op diagram]
        \coordinate (O) at (0,0);
        \coordinate (X) at (1.3,6);
        \coordinate (Y1) at (-1.3,2);
        \coordinate (Y2) at (-1.3,4);
        \coordinate (Y3) at (-1.3,5.5);
        \coordinate (Y1') at (1.3,1.5);
        \coordinate (Y2') at (1.3,3.5);
        \laceDraw[first=2] (O) (Y1);
        \laceDraw[first=1] (O) (Y1');
        \draw (Y1') -- (Y1);
        \laceDraw[first=1, endpoint-shape=line] (Y1) (Y2);
        \laceDraw[first=1, endpoint-shape=line] (Y1') (Y2');
        \draw (Y2') -- (Y2);
        \laceDraw[first=1, endpoint-shape=line] (Y2) (Y3);
        \laceDraw[first=2, endpoint-shape=line] (Y2') (X);
        \draw (Y3) node[vertex] {} -- (X);
        \draw[decorate, decoration={brace, amplitude=5pt, raise=0.7ex}] (X) -- (X |- O) node [midway, sloped, above, yshift=1.5ex] {$m$};
        \lacePutLabel[anchor=north] (O) {$\vb*{o}$};
        \lacePutLabel[anchor=south] (X) {$\vb*{x}$};
        \lacePutLabel[anchor=east] (Y1) {$\vb*{y}$};
        \lacePutLabel[anchor=east] (Y2) {$\vb*{w}$};
      \end{tikzpicture}
      +
      \begin{tikzpicture}[op diagram]
        \coordinate (O) at (0,0);
        \coordinate (X) at (0,6);
        \coordinate (Y1) at (-1.3,2);
        \coordinate (Y2) at (-1.6,3.25);
        \coordinate (Y3) at (-1.3,4.5);
        \coordinate (Y1') at (1.3,1.5);
        \coordinate (Y2') at (1.3,4);
        \laceDraw[first=2] (O) (Y1);
        \laceDraw[first=1] (O) (Y1');
        \draw (Y1') -- (Y1);
        \laceDraw[first=1, endpoint-shape=line] (Y1) (Y2);
        \draw (Y2) node[vertex] {} -- (Y3);
        \laceDraw[first=1, endpoint-shape=line] (Y1') (Y2');
        \draw (Y2') -- (Y3);
        \laceDraw[first=1, endpoint-shape=line] (Y3) (X);
        \laceDraw[first=2, endpoint-shape=line] (Y2') (X);
        \draw[decorate, decoration={brace, amplitude=5pt, raise=0.7ex}] (Y2' |- X) -- (Y1' |- O) node [midway, sloped, above, yshift=1.5ex] {$m$};
        \lacePutLabel[anchor=north] (O) {$\vb*{o}$};
        \lacePutLabel[anchor=south] (X) {$\vb*{x}$};
        \lacePutLabel[anchor=east] (Y1) {$\vb*{y}$};
        \lacePutLabel[anchor=east] (Y3) {$\vb*{w}$};
      \end{tikzpicture}
      +
      \begin{tikzpicture}[op diagram]
        \coordinate (O) at (1.3,0);
        \coordinate (X) at (0,6);
        \coordinate (Y1) at (-1.3,0.5);
        \coordinate (Y2) at (-1.3,2.5);
        \coordinate (Y3) at (-1.3,4.5);
        \coordinate (Y1') at (1.3,2);
        \coordinate (Y2') at (1.3,4);
        \laceDraw[first=2] (O) (Y1);
        \draw (Y1) node[vertex] {} -- (Y2);
        \laceDraw[first=1] (O) (Y1');
        \draw (Y1') -- (Y2);
        \laceDraw[first=1, endpoint-shape=line] (Y2) (Y3);
        \laceDraw[first=1, endpoint-shape=line] (Y1') (Y2');
        \draw (Y2') -- (Y3);
        \laceDraw[first=1, endpoint-shape=line] (Y3) (X);
        \laceDraw[first=2, endpoint-shape=line] (Y2') (X);
        \draw[decorate, decoration={brace, amplitude=5pt, raise=0.7ex}] (Y2' |- X) -- (O) node [midway, sloped, above, yshift=1.5ex] {$m$};
        \lacePutLabel[anchor=north] (O) {$\vb*{o}$};
        \lacePutLabel[anchor=south] (X) {$\vb*{x}$};
        \lacePutLabel[anchor=east] (Y2) {$\vb*{y}$};
        \lacePutLabel[anchor=east] (Y3) {$\vb*{w}$};
      \end{tikzpicture}
      +
      \begin{tikzpicture}[op diagram]
        \coordinate (O) at (0,0);
        \coordinate (X) at (0,5.5);
        \coordinate (Y1) at (-1.3,2);
        \coordinate (Y2) at (-1.3,4);
        \coordinate (Y1') at (1.3,1.5);
        \coordinate (Y2') at (1.3,3.5);
        \laceDraw[first=2] (O) (Y1);
        \laceDraw[first=1] (O) (Y1');
        \draw (Y1') -- (Y1);
        \laceDraw[first=1, endpoint-shape=line] (Y1) (Y2);
        \laceDraw[first=1, endpoint-shape=line] (Y1') (Y2');
        \draw (Y2') -- (Y2);
        \laceDraw[first=1, endpoint-shape=line] (Y2) (X);
        \laceDraw[first=2, endpoint-shape=line] (Y2') (X);
        \draw[decorate, decoration={brace, amplitude=5pt, raise=0.7ex}] (Y2' |- X) -- (Y1' |- O) node [midway, sloped, above, yshift=1.5ex] {$m$};
        \lacePutLabel[anchor=north] (O) {$\vb*{o}$};
        \lacePutLabel[anchor=south] (X) {$\vb*{x}$};
        \lacePutLabel[anchor=east] (Y1) {$\vb*{y}$};
        \lacePutLabel[anchor=east] (Y2) {$\vb*{w}$};
      \end{tikzpicture}
    )\\
    &\leq 3 T_{p, m}^{(1, 2)} T_{p, m}^{(1, 1)} T_{p, m}^{(2, 1)}
      + B_{p, m}^{(1, 2)} T_{p, m}^{(1, 1)} T_{p, m}^{(2, 1)},
  \end{align*}
  which corresponds to the 5th term in the right hand side in \eqref{eq:sum-diagrammaticBounds-mt-2}.

  Finally, we consider \eqref{eq:sum-diagrammaticBounds-mcos-01}--\eqref{eq:sum-diagrammaticBounds-mcos-N}.
  By Lemma~\ref{lem:split-of-cosine}, for example,
  \begin{align*}
    &\sum_{(x, t)}
    \begin{tikzpicture}[op diagram]
      \laceDraw[first=2] (0,0) (-1.3,2);
      \laceDraw[first=1] (0,0) (1.3,1.5);
      \draw (1.3,1.5) -- (-1.3,2);
      \laceDraw[first=1, endpoint-shape=line] (-1.3,2) (1.3,3.5);
      \laceDraw[first=1, endpoint-shape=line] (1.3,1.5) (-1.3,4);
      \draw (1.3,3.5) -- (-1.3,4);
      \laceDraw[first=1, endpoint-shape=line] (-1.3,4) (0,5.5);
      \laceDraw[first=2, endpoint-shape=line] (1.3,3.5) (0,5.5);
      \lacePutLabel[anchor=north] (0,0) {$(o, 0)$};
      \lacePutLabel[anchor=south] (0,5.5) {$(x, t)$};
    \end{tikzpicture}
    m^t \qty(1 - \cos k\cdot x)\\
    &\leq
    3 \sum_{\substack{(x, t),\\ (y, s), (w, r),\\ (y', s'), (w', r')}}
    \begin{tikzpicture}[op diagram]
      \laceDraw[first=2] (0,0) (-1.3,2);
      \laceDraw[first=1] (0,0) (1.3,1.5);
      \draw (1.3,1.5) -- (-1.3,2);
      \laceDraw[first=1, endpoint-shape=line] (-1.3,2) (1.3,3.5);
      \laceDraw[first=1, endpoint-shape=line] (1.3,1.5) (-1.3,4);
      \draw (1.3,3.5) -- (-1.3,4);
      \laceDraw[first=1, endpoint-shape=line] (-1.3,4) (0,5.5);
      \laceDraw[first=2, endpoint-shape=line] (1.3,3.5) (0,5.5);
      \lacePutLabel[anchor=north] (0,0) {$(o, 0)$};
      \lacePutLabel[anchor=south] (0,5.5) {$(x, t)$};
      \lacePutLabel[anchor=west] (1.3,1.5) {$(y, s)$};
      \lacePutLabel[anchor=east] (-1.3,2) {$(y', s')$};
      \lacePutLabel[anchor=east] (-1.3,4) {$(w, r)$};
      \lacePutLabel[anchor=west] (1.3,3.5) {$(w', r')$};
    \end{tikzpicture}
    \begin{aligned}
      &\times m^{t - r'} m^{r' - s'} m^{s'}\\
      &\times \Bigl(\qty\big(1 - \cos k\cdot (x - w))\\
      &\quad + \qty\big(1 - \cos k\cdot (w - y)) + \qty\big(1 - \cos k\cdot y)\Bigr)
    \end{aligned}\\
    &\leq 3 \qty(
      \flt V_{p, m}^{(1, 2)}(k) T_{p, m}^{(1, 1)} T_{p, m}^{(1, 2)}
      + T_{p, m}^{(1, 2)}(k) \flt V_{p, m}^{(1, 1)} T_{p, m}^{(1, 2)}
      + T_{p, m}^{(1, 2)}(k) T_{p, m}^{(1, 1)} \flt V_{p, m}^{(1, 2)}
    ),
  \end{align*}
  which corresponds to the last term on the right hand side in \eqref{eq:sum-diagrammaticBounds-mcos-2}.
\end{proof}

\end{document}